\documentclass{article}
\usepackage{geometry}
\usepackage[hyphens]{url}  
\usepackage{graphicx}
\usepackage[english]{babel}
\usepackage[T1]{fontenc}

\usepackage[utf8]{inputenc}

\usepackage{tcolorbox}

\usepackage{amssymb}
\usepackage{amsthm}
\usepackage{mathtools,thmtools}
\usepackage{nicefrac}

\usepackage{booktabs}
\usepackage{tabularx}

\usepackage{tikz,pgfplots,xcolor}
\pgfplotsset{compat=1.16}
\usepgfplotslibrary{fillbetween}
\usetikzlibrary{calc}
\definecolor{shadingblue}{RGB}{197,209,227}
\definecolor{labelblue}{RGB}{20, 52, 164}

\usepackage{natbib}
\usepackage[capitalize,nameinlink]{cleveref}

\usepackage{caption}
\usepackage{subcaption}

\usepackage[inline]{enumitem}
\setlist{nosep,nolistsep,labelsep=3pt,labelindent=0pt,leftmargin=7pt}

\newtheorem{theorem}{Theorem}
\newtheorem{lemma}{Lemma}
\newtheorem{proposition}{Proposition}

\theoremstyle{definition}
\newtheorem{definition}{Definition}

\newcommand{\naturals}{\ensuremath{\mathbb{N}}}

\newcommand{\np}{{\sffamily NP}}

\newcommand{\nphard}{\np-hard}

\newcommand{\dotcup}{\mathbin{\dot{\cup}}}

\newcommand{\agentsperm}{\ensuremath{\pi_\textrm{a}}}
\newcommand{\resperm}{\ensuremath{\pi_\mathit{goods}}}
\newcommand{\allPermAgents}{\ensuremath{\Pi_\textrm{A}}}
\newcommand{\allPermgoods}{\ensuremath{\Pi_\textrm{m}}}

\newcommand{\idealdist}{\ensuremath{d_\textrm{v}}}
\newcommand{\demdist}{\ensuremath{d_\textrm{d}}}

\DeclareMathOperator{\dem}{dem}
\newcommand{\sordem}{\ensuremath{\overrightarrow{\dem}}}

\newcommand{\goods}{\ensuremath{\mathcal{R}}}

\newcommand{\task}{\ensuremath{U}}

\newcommand{\IND}{{{\mathrm{IND}}}}
\newcommand{\SEP}{{{\mathrm{SEP}}}}
\newcommand{\CON}{{{\mathrm{CON}}}}
\newcommand{\BIC}{{{\mathrm{BIC}}}}
\newcommand{\WSEP}{{{\mathrm{WSEP}}}}
\newcommand{\WSEPf}{{{\mathrm{WSEPf}}}}

\newcommand{\norm}[1]{\lVert #1 \rVert}

\newcommand{\thespan}[1]{\mathit{span}(\{#1\})}

\title{Putting Fair Division on the Map}

\date{}
\author{
  \normalsize{\textbf{Paula B\"ohm}}\\[0pt]
  {\small Institut f\"ur Informatik}\\[-3pt]
  \small{TU Clausthal} \\[-3pt]
  \small{Clausthal-Zellerfeld, Germany} \\[-3pt]
  \small{\texttt{paula.boehm@tu-clausthal.de}}
  \and
  \normalsize{\textbf{Robert Bredereck}} \\[0pt]
  \small{Institut f\"ur Informatik} \\[-3pt]
  \small{TU Clausthal} \\[-3pt]
  \small{Clausthal-Zellerfeld, Germany} \\[-3pt]
  \small{\texttt{robert.bredereck@tu-clausthal.de}}\\
  \and
  \normalsize{\textbf{Paul G\"olz}} \\[0pt]
  \small{ORIE %
  } \\[-3pt]
  \small{Cornell University} \\[-3pt]
  \small{Ithaca, NY, USA} \\[-3pt]
  \small{\texttt{paulgoelz@cornell.edu}}
  \and
  \normalsize{\textbf{Andrzej Kaczmarczyk}} \\[0pt]
  \small{Department of Computer Science} \\[-3pt]
  \small{The University of Chicago} \\[-3pt]
  \small{Chicago, IL, USA} \\[-3pt]
  \small{\texttt{akaczmarczyk@uchicago.edu}}\\
  \and
  \normalsize{\textbf{Stanis\l{}aw Szufa}} \\[0pt]
  \small{CNRS, LAMSADE} \\[-3pt]
  \small{Université Paris Dauphine - PSL} \\[-3pt]
  \small{Paris, France} \\[-3pt]
  \small{\texttt{s.szufa@gmail.com}}
}

\begin{document}

\newtcbox{\tagbox}{size=fbox,nobeforeafter,colback=black,colupper=white,
left skip=2pt, right skip=2pt}

\newsavebox{\threesixbox}
\sbox{\threesixbox}{%
\tagbox{$3 \times 6$}%
}
\newlength{\threesixboxwidth}
\settowidth{\threesixboxwidth}{\usebox{\threesixbox}}

\newsavebox{\fivefivebox}
\sbox{\fivefivebox}{%
\tagbox{$5 \times 5$}%
}
\newlength{\fivefiveboxwidth}
\settowidth{\fivefiveboxwidth}{\usebox{\fivefivebox}}

\maketitle
\begin{abstract}
The fair division of indivisible goods is not only a subject of theoretical research, but also an important problem in practice, with solutions being offered on several online platforms.
Little is known, however, about the characteristics of real-world allocation instances and how they compare to synthetic instances.
Using dimensionality reduction, we compute a \emph{map} of allocation instances: a 2-dimensional embedding such that an instance's location on the map is predictive of the instance's origin and other key instance features.
Because the axes of this map closely align with the utility matrix's two largest singular values, we define a second, explicit map, which we theoretically characterize.
\end{abstract}
\pagestyle{plain}
\section{Introduction}
Over the past 20 years, the field of fair division has made great advances in
studying allocations of indivisible goods~\citep{aab+23}.
To illustrate this progress, consider the axiom of \emph{envy-freeness}, which
demands that no agent prefer another agent's bundle of allocated goods to their
own. 
By the end of the 20th century, economists already understood envy-freeness well in settings with \emph{divisible} goods.
For example\,---\,assuming that preferences are additive, as we do throughout
the paper\,---\,the allocation maximizing the Nash welfare is envy-free in these
settings~\citep{varian74,ss82}.
Little was known, however, about indivisible goods, a domain whose combinatorial
structure poses additional challenges to mathematical investigation.
Since envy-free allocations need not exist for all indivisible allocation
instances\footnote{Under the standard assumption that all goods must be allocated, which we also make throughout the paper.}, could envy at
least be limited, or were large amounts of envy unavoidable?

Since then, 
we have gained a refined understanding of the
degree to which envy can (and cannot) be avoided.
Notably, the field has coalesced around an attractive relaxation of
envy-freeness\,---\,\emph{envy-freeness up to one good}
(\emph{EF1})~\citep{budish11}\,---\,and
identified elegant algorithms~\citep{lmm+04,ckm+19} that construct EF1 allocations for any instance.
Though intriguing open questions remain,\footnote{For example, whether
envy-freeness up to \emph{any} good (EFX) can be
guaranteed~\citep{cgm20,pr20,abf+21}.} these questions sharpen and extend
a solid understanding of the landscape of allocations.  %
For alternative families of fairness axioms (such as the maximin share,
proportionality, and equitability), fair division has made similar progress in
understanding which axioms can be guaranteed even on worst-case allocation
instances~\citep{aab+23}.

In parallel to these theoretical advances, algorithms for allocating indivisible goods have entered practical usage, raising new questions for fair division.
The Course Match system, for example, assigns course seats to MBA students at
Wharton~\citep{bck+17}, and thousands of users have used the website
Spliddit~\citep{gp14a} to divide up estates or joint possessions.
The deployment of such systems makes it more pressing to study not only
worst-case instances but instances typically encountered in practice as
well.
For example, though envy-free allocations do not exist for all instances, should
algorithms not aim for envy-free allocations for %
the 71\% of Spliddit instances~\citep{bfg+22}
where envy-freeness is possible? If so, how to choose
among envy-free allocations?
Or, as a second example, which algorithms can be implemented in practice?
After all, the algorithms deployed on Course Match~\citep{bgo+23} and
Spliddit~\citep{ckm+19}, as well as other proposed methods~\citep{bfkkr21}, run
fast on practical inputs seemingly defying (worst-case) computational hardness
results. Answers to these questions cannot be found through worst-case analysis
alone.

Whereas most work in fair division follows the worst-case paradigm, a noteworthy
exception is some work in the paradigm of \emph{distributional analysis}, which
assumes that allocation instances are drawn from a probability
distribution~\citep{roughgarden20}.
Typically, these distributional models assume that all agent--good values are
drawn independently, either from a single distribution~\citep{amn+17,ms19,ms21},
a distribution that depends on the agent~\citep{kpw16,fgh+19,bg22}, or a
distribution that depends on the alternative~\citep{dgk+14,fgh+19}.

On the upside, when $m$ and $n$ are large, these distributions generate highly
structured instances, for which fair allocations are more prevalent. For
example, basic algorithms yield envy-free allocations for these instances, with
high probability~\citep{dgk+14,ms21}.
On the flip side, we are not aware of any empirical work that has tested if the
structures of these random instances are present in practical allocation
problems.\footnote{\citet{bfg+22} voice doubts, but also do not provide data.}
In recent work, \citet{bfg+22} try to overcome some of these concerns through \emph{smoothed analysis}, which means that their probability distributions are defined by adding random noise to a worst-case utility profile.
While they extend the possibility results for envy-freeness to more general distributions, they leave ``whether our smoothed model accurately describes the properties of real-world utility profiles that possess envy-free allocations'' to future empirical analysis.

Recently, social choice theory
addressed a similar need for bridging the gap between theoretical work and
practical instances.
Famously riddled with worst-case impossibilities~\citep{ck02},
modern social choice theory has been largely divorced from the analysis of election data.
Part of the theory considered distributional models of elections, but these models were overly prescriptive and, anyway, not easy to relate to real-world elections.
To address these concerns, \citet{sfs+20} created a \emph{map of elections}:
A two-dimensional embedding of election instances (originally from
distributional models and, in subsequent
work~\citep{boe-bre-fal-nie-szu:c:compass,fal-kac-sor-szu-was:c:microscope},
also from real-world data) with the following properties:
\begin{itemize}
\item This map recovers tell-tale features of the different data sources,
	turning elections generated by a specific probability distribution into a compact
	cluster, which evidently only covers a small subset of interesting
	elections.
\item Key features of election instances vary continuously over the map,
  showing that an election's position on the map is highly
	informative\,---\,despite rendering the high-dimensional space of
	elections in just two dimensions.
\item The map's two axes can be given a conceptual interpretation, and the
	boundary of the map be traced by natural ``extreme'' elections.
\end{itemize}
One achievement of this line of work was to highlight distributional models which, at least on the axes of the map, capture the range of interesting elections.
We apply a similar approach to fair division.

\subsection{Our Approach and Results}\label{sec:aproachandresults}
In fact, we create \emph{two} maps of allocation instances for indivisible goods in this paper.
In \cref{sec:embeddingmap}, we create a map by following the methodology
of \citet{sfs+20}:
we define a natural distance between fair division instances (the $\ell^1$ distance between utility matrices up to row and column permutations),
and use multi-dimensional scaling~\citep{kru:j:mds} (a common technique for dimensionality reduction) to find a 2-dimensional
\emph{distance-embedding map} of a given set of allocation instances that
approximately preserves the pairwise distances.
To scale this approach to allocation instances with many agents and goods, we also propose a computationally tractable proxy distance, which leads to almost identical maps at a much smaller computational cost.

Using data from three real-world data sources and several synthetic distributions over approval instances, we show that the map picks up on common properties of instances from the same data source.
We also show that key features of the allocation instance (e.g., the maximum achievable Nash welfare or the existence of envy-free allocations) are distributed in clear patterns across the map.
Whereas two natural probability distributions over instances fail to cover the
whole range of real-world instances, a new distribution we propose covers the whole
map in our experiments and is therefore a natural choice for future synthetic
experiments in fair division.

In \cref{sec:explicit}, we go beyond the heuristically generated embedding described above, by providing an \emph{explicit map}, i.e., an explicit function from allocation instances into $\mathbb{R}^2$, which reproduces the general layout of the distance-embedding map.
Since an instance's position on this explicit map is given by the two largest singular values of its utility matrix, our explicit map is amenable to theoretical analysis.
In particular, we tightly characterize the range of the map, and identify (up to
rounding terms) the most extreme instances in the map's four corners.
We conclude by showing that the explicit map can similarly segment instance sources and features.

\section{Preliminaries}\label{sec:prelim}

Let $[n]$ be a set of agents and $[m]$ be a set of goods.
For ease of exposition, we assume throughout the paper that $m \geq n \geq 2$, which arguably includes all interesting allocation instances.
An \emph{allocation instance} (of indivisible goods) is described by a \emph{utility matrix}
$U \in \mathbb{R}_{\geq 0}^{n \times m}$ whose entries $u_{i,j}$ describe agent
$i$'s utility for good~$j$. We refer to the $i$th row of
this matrix as $i$'s \emph{utility vector} $\vec{u}_i \in \mathbb{R}_{\geq
0}^m$. We assume that preferences are additive, so that agent $i$'s utility for a \emph{bundle} $S \subseteq [m]$ of
goods is given by $u_i(S) \coloneqq \sum_{j \in S}
u_{i,j}$. Since we only consider tasks in which agents' utilities $u_i([m])$ for
the whole bundle are normalized to 1, we consider exactly the set of \emph{row-stochastic} matrices $U$.

\subsection{Characteristic Instances}
\label{sec:prelim:instances}
As useful signposts for navigating through the space of allocation instances, we define several \emph{characteristic instances}.
Each of these instances represents an intuitively extreme scenario with easily-understood, symmetric structure.
Due to space limitations, we introduce these instances in words here, and refer
the reader to \cref{fig:characteristic-instances} in the appendix for a matrix
representation. For any $n$ and $m$, our three characteristic instances are the
following:
\begin{description}
	\item[Indifference (IND)] models the situation where each good is equally
		valuable to each agent.  Thus, all entries of its utility matrix are
		$\nicefrac{1}{m}$.

	\item[Separability (SEP)] captures the scenario in which each agent values
		exactly one
		good, distinct from the goods that all other agents value.
Thus, its utility matrix
is a matrix with ones on the diagonal and zeros everywhere else. In particular,
if $m>n$, all but the first $n$ columns have all-zero entries.

\item[Contention (CON)]
considers one single good
valued by all agents, and all
other goods
yielding no value to any agent.
Hence, its utility matrix has a first column of ones, and is zero everywhere else.

\end{description}

In \cref{sec:explicit:theory}, the explicit map will lead us to introduce three
new characteristic instances: two variants of separability and an entirely new
instance called \emph{bicontention}.

\subsection{Real-World Instances}
\label{sec:prelim:realdata}
Three of the instance sources we consider generate instances derived from real-world preferences over goods. Two of these sources have not been previously analyzed in the fair-division literature:
\begin{description}
\item[Spliddit.]
The heart of our real-world data is a dataset~\citep{splidditdata} of all
allocation instances submitted to Spliddit as of 2022.
This dataset is particularly valuable because it represents instances that
real Spliddit users were hoping to solve.  Since most of the 3000 Spliddit
instances are small, our evaluation will focus on two combinations of $n$,
$m$ that are relatively well represented: First, we will study the setting
$n=3, m=6$, for which the number of instances (namely, 1847) is highest.  Since, for
larger dimensions, the number of Spliddit instances drops precipitously,
only 16 Spliddit instances exist for our second evaluation scenario of
$n=m=5$.
\item[Island.]
To obtain this dataset~\citep{islanddata}, \citet{bis+18} elicited additive utilities for private goods (though they were motivated by a
public-goods setting), by asking
572 crowd workers to spread 100 points between 10 items in proportion to
how much they would value these items (a map, pocket knife, compass, etc.) if
they were stranded on a deserted island.
By sampling sets of $n$ agents and $m$ goods and rescaling agents'
utilities, we simulate (hypothetical) allocation scenarios in which only those
$m$ goods stand to be allocated between those $n$ (real but fictitiously
stranded) agents.
\item[Candy.]
	Our final dataset~\citep{candydata} has a similar shape and consists of the
	additive preferences over 10
	types of snacks indicated by 48 teenagers attending a summer camp. We again obtain instances by subsampling, assuming that only
	one snack of each type is available.
\end{description}

\subsection{Distributions over Synthetic Instances}
In addition to the above instances derived from practical data, we consider
synthetic instances drawn from three types of distributions:
\begin{description}
	\item[I.i.d.]
As described in the introduction, i.i.d.\ valuations have been empirically and theoretically studied in the literature.
Each agent $i$'s utility vector is independently generated by sampling $i$'s values for the $m$ goods independently from some fixed distribution $\mathcal{D}$, and then rescaling this vector to sum up to 1.
In our experiments, we choose $\mathcal{D}$ as the uniform distribution over
$[0,1]$ and as the exponential distribution. (The exponential
distribution's rate is inconsequential since utilities are normalized.)
\item[Attributes.]
This describes a natural explanatory model of how agents' utilities arise (also used by \citet{boe-hee-szu:t:map-roommates}).
Let $d\in\naturals$ be a fixed number of \emph{attributes}.  For each good $j$, we
sample a vector $\vec{g}_j$ from $[0,1]^d$ uniformly at random (higher coordinates
indicate that the good is more desirable along an
attribute). %
For each agent $i$, we sample a \emph{priority vector}~$\vec{a}_i$ over the attributes also
from $[0,1]^d$
(higher coordinates indicate that the agent cares more about an attribute).
Then, agent $i$'s utility for good $j$ is proportional to the dot product
of~$\vec{a}_i$ and $\vec{g}_j$.
\item[Resampling.]
The former two distributions will end up only covering small areas of the
map. Hence, we introduce a third distribution, inspired by that over approval
elections~\citep{szu-fal-jan-lac-sli-sor-tal:c:map-approval}, which will cover
the range of real-world allocation instances.
For each agent, we generate a set of \emph{approved} goods over which the
agent splits the total utility of~$1$ equally.
Given two parameters $p \in
[0,1]$ and $\phi \in [0,1]$, we first choose the instance's \emph{central approval set} $V^*$
by uniformly drawing $\lfloor p \cdot m \rfloor$~goods. 
Then, we determine whether agent $i$ approves a good $j$ as follows (independently across $i,j$):
with probability $1 - \phi$, $i$ approves $j$ iff $j \in V^*$; else, $i$ approves $j$ with probability $p$.
Should this process leave $i$ without any approved goods, they approve one good uniformly at random.

\end{description}

\section{Distance-Embedding Map}
\label{sec:embeddingmap}
We create our first map by performing the following two steps, introduced by
\citet{sfs+20}.
First, we define a notion of
distance between pairs of allocation instances, and compute all pairwise
distances for a collection of instances from the sources described above.\footnote{\Cref{app:datasets} describes the number and parameters of the instances we study from each source.}
Second, we use \emph{multi-dimensional scaling}~\citep{kru:j:mds} to embed this
collection of instances into the plane in a way that approximately preserves
their pairwise distances,
which will allow us to see patterns in the similarity of instances.

\subsection{Distances between Allocation Instances}
Na\"ively, we would like to measure the distance between two instances (with
equal $n$ and $m$) through the entry-wise $\ell^1$ norm.
That is, if the instances' utility matrices are $U^1$ and $U^2$, we would
calculate their distance $||U^1 - U^2||_{1,1}$ by summing up, over all $n \cdot
m$ coordinates, the absolute difference between $U^1$'s and $U^2$'s entry at
this coordinate.
This distance is, however, not desirable, since the ordering of rows and columns
in a utility matrix is arbitrary but would greatly impact this distance.
Instead, an appropriate distance between allocation instances ought to remain
invariant when reordering the utility matrices' rows (i.e., agents) or columns
(i.e., goods).

Our \emph{valuation distance} achieves these goals of anonymity and neutrality by explicitly minimizing over all row and column permutations.
To express this without matrix notation, suppose that we have bijections
$\pi_{\mathit{agents}} \colon [1, n] \to [1, n]$ and $\pi_{\mathit{goods}}
\colon [1, m] \to [1, m]$
between, respectively, the agents and the goods of the first and second
allocation instance.
By summing up, for each agent $i$ and each good~$j$, the absolute difference
$|U^1_{i,j} - U^2_{\pi_{\mathit{agents}}(i), \pi_{\mathit{goods}}(j)}|$ between
$i$'s utility for $j$ and the utility of $i$'s matched agent
$\pi_{\mathit{agents}}(i)$ for $j$'s matched good $\pi_{\mathit{goods}}(j)$, we
calculate the entry-wise $\ell^1$ distance; the valuation distance is defined as
the minimum of this distance, taken over all matchings $\pi_{\mathit{agents}}$
and $\pi_{\mathit{goods}}$.
A nice property of the valuation distance is that two instances $U^1$ and $U^2$
have distance zero exactly if they are identical up to relabeling agents and
goods, i.e., they are \emph{isomorphic}. %

\begin{figure*}[tbh]
\centering
	\def\ratio{.5}
	\includegraphics[height=18pt]{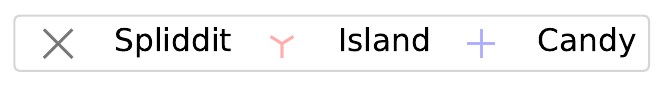}
	\includegraphics[height=18pt]{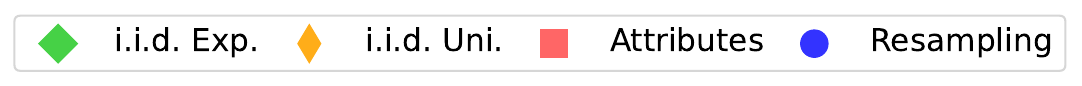}\\
	\begin{minipage}{.44\linewidth}%
	  \begin{subfigure}[t]{\ratio\linewidth}
	  	\centering
	  		\includegraphics[width=\linewidth, trim=0cm 1cm 0cm 0cm, clip]{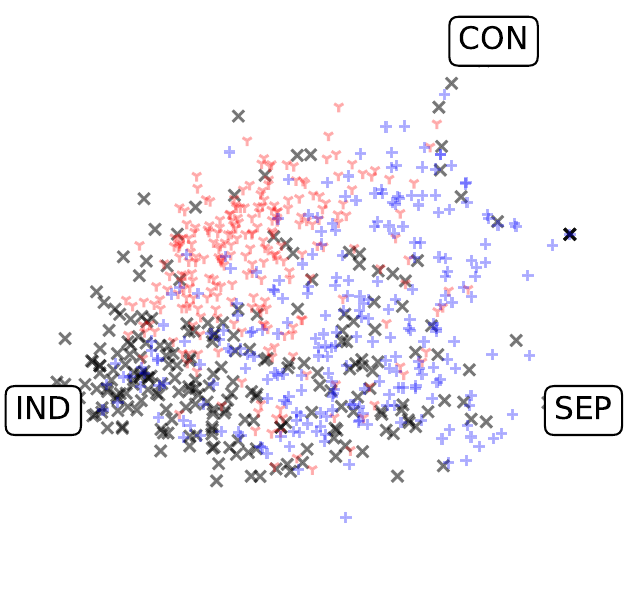}
				\caption{Real-world sources}%
     \label{fig:real36}
	  \end{subfigure}%
	  \begin{subfigure}[t]{\ratio\linewidth}
	  	\centering
	  		\includegraphics[width=\linewidth, trim=0cm 1cm 0cm 0cm, clip]{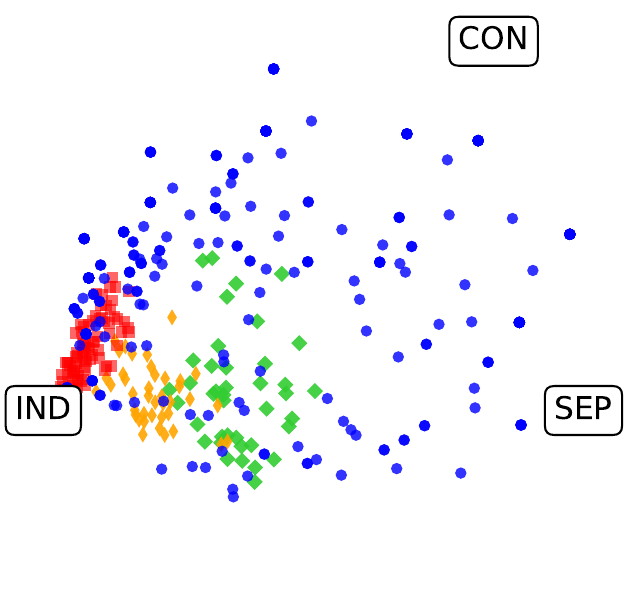}
			\caption{Synthetic sources}%
     \label{fig:synthetic36}
	  \end{subfigure}\\
	  \begin{subfigure}[t]{\ratio\linewidth}
	  	\centering
	  		\includegraphics[width=\linewidth]{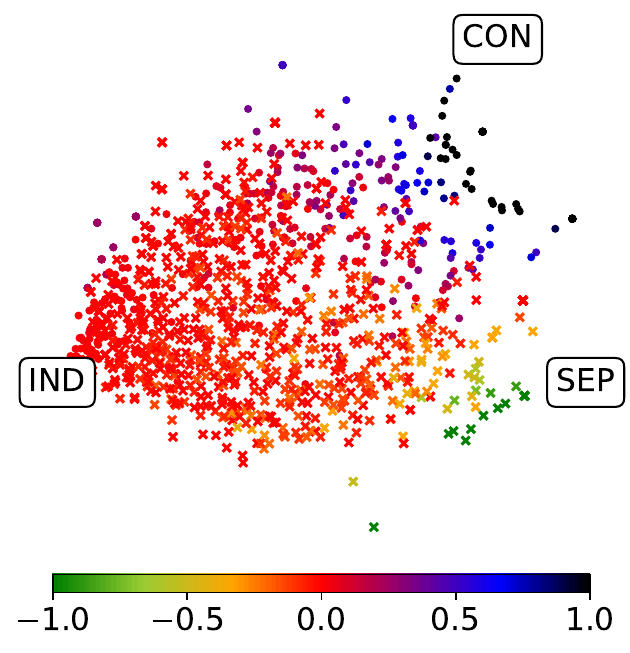}
			\caption{Minimax Envy}%
              \label{fig:envy36}
    \end{subfigure}%
	  \begin{subfigure}[t]{\ratio\linewidth}
	  	\centering
	  		\includegraphics[width=\linewidth]{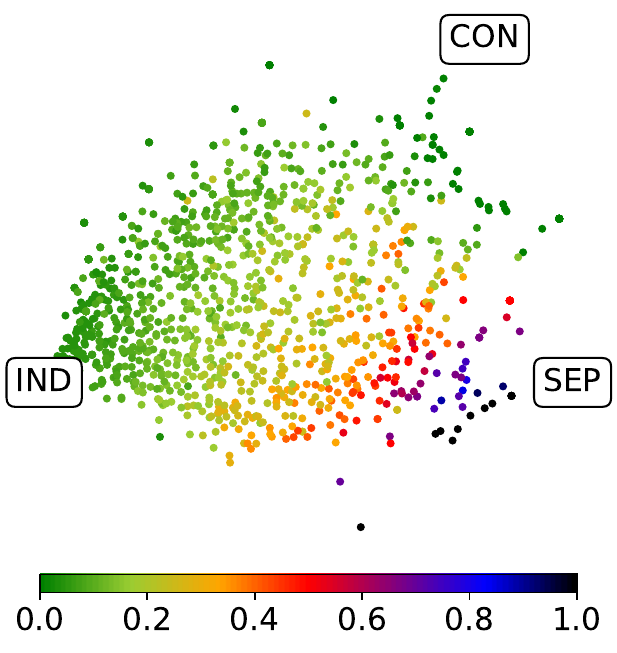}
				\caption{Max. Nash Welf.%
	  	 \vphantom{$\text{brown star}\mathop{\hat{=}}0$}}
       \label{fig:nash36}
	  \end{subfigure}%
	\end{minipage}\hfill%
	\hspace{-\threesixboxwidth}\raisebox{-14em}{\usebox{\threesixbox}}%
	\rule[-14em]{2pt}{29em}%
	\raisebox{-14em}{\usebox{\fivefivebox}}\hspace{-\fivefiveboxwidth}%
	\hfill%
	\begin{minipage}{.49\linewidth}
	  \begin{subfigure}[t]{\ratio\linewidth}
	  	\centering
	  	\includegraphics[width=\linewidth, trim=0cm 1cm 0cm 0cm, clip]{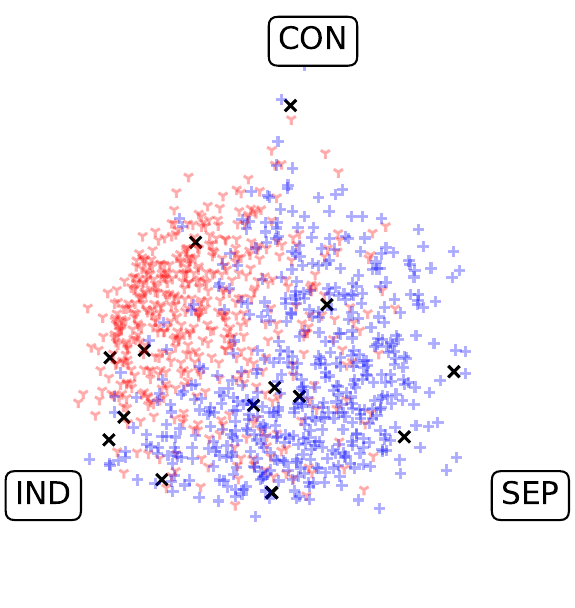}
		\caption{Real-world sources}%
    \label{fig:real55}
	  \end{subfigure}%
	  \begin{subfigure}[t]{\ratio\linewidth}
	  	\centering
	  	\includegraphics[width=\linewidth, trim=0cm 1cm 0cm 0cm, clip]{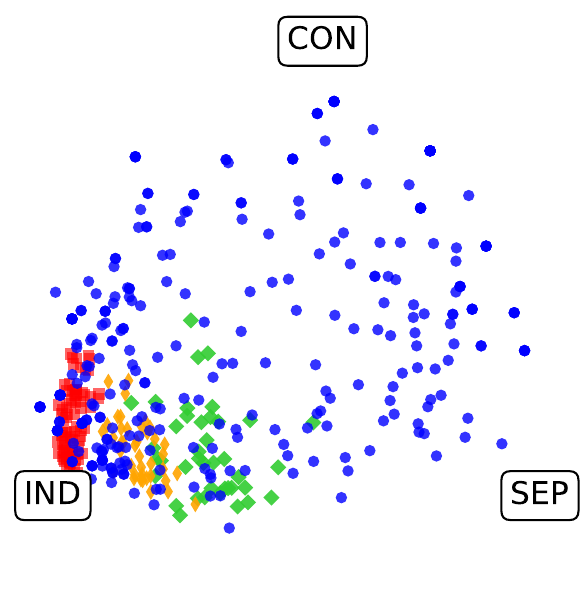}
		\caption{Synthetic sources}%
    \label{fig:synthetic55}
	  \end{subfigure}\\
	  \begin{subfigure}[t]{\ratio\linewidth}
	  	\centering
	  	\includegraphics[width=\linewidth]{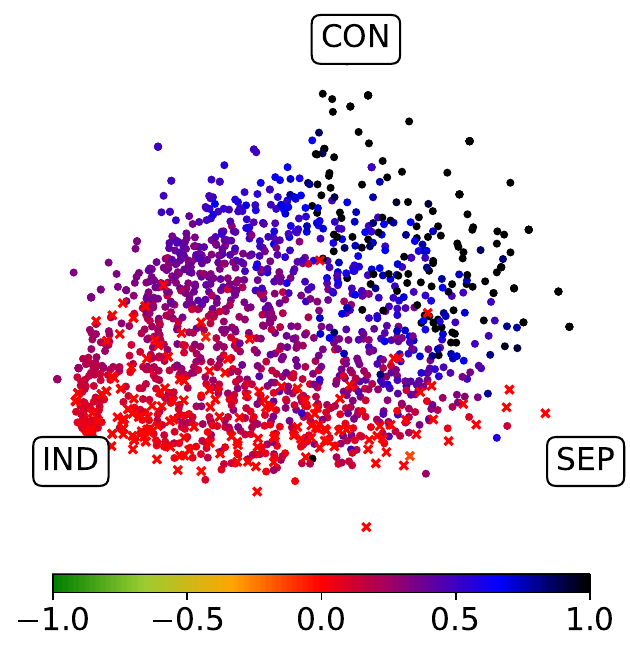}
		\caption{Minimax Envy}%
              \label{fig:envy55}
	  \end{subfigure}%
	  \begin{subfigure}[t]{\ratio\linewidth}
	  		\centering
	  		\includegraphics[width=\linewidth]{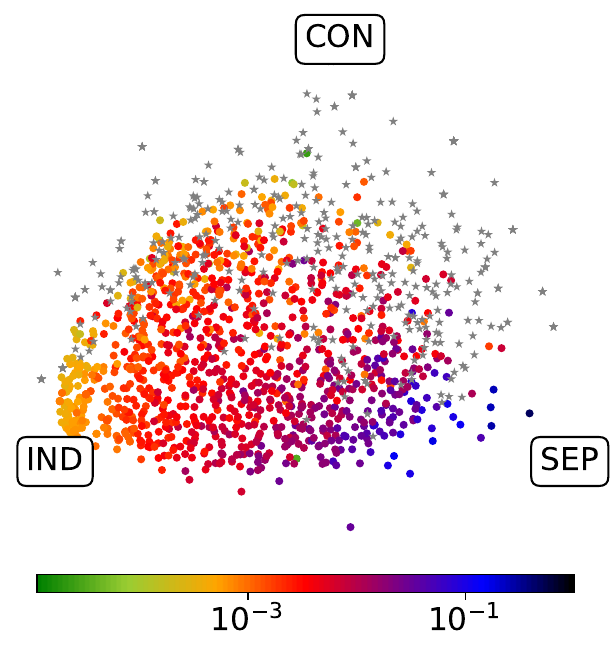}
			\caption{Max. Nash Welf. (grey star indicates 0)}%
      \label{fig:nash55}
	\end{subfigure}%
	\end{minipage}%
	\caption{Distribution of instance sources and two features on
	 the distance-embedding maps. %
 }
  \label{exp:main_maps} 
	\let\ratio\undefined
\end{figure*}

The valuation distance confirms our intuition that the
characteristic instances introduced in~\cref{sec:prelim:instances} are indeed ``extreme points''
in the space of allocation instances in the sense that, if $n=m$, the characteristic
instances are mutually equidistant at distance~$2\,(m-1)$; and this  distance is maximal (see \cref{prop:distupperbound} in \cref{sec:demand-distance}).

\subsection{Studying the Distance-Embedding Map}\label{subsec:stude}
To plot distance-embedding maps, we first generate a collection of instances
from all of our sources\,---\,as mentioned in \cref{sec:prelim:realdata}, we
consider 
two combinations of $n,m$: for $n=3, m=6$ (``$3 \times 6$'' from now on) and
for $n=m=5$ (``$5 \times 5$''); see~\cref{app:main-datasets} for details.
Then, we compute the valuation distance between all pairs of instances in
each collection,
and embed the distances
in 2D Euclidean space using multi-dimensional scaling (implemented
in scikit-learn, using default parameters).\footnote{We implemented the
	generation of our distance-embedding maps as a
	module of
\emph{mapel} (\url{https://mapel.simple.ink}), a framework
for computing maps of elections~\citep{sfs+20}, and will release this module as
open source.}
\cref{exp:main_maps} displays the resulting embeddings, where each point represents an instance, and an instance's location stays fixed within all maps of the same dimension.
Note that the characteristic instances
lie
in distinct corners of the map, reflecting
our observation at the end of the previous subsection.
Most of the map is spanned in a triangle between these three instances, making them useful points of reference.

We can also immediately make out that the instance sources are
spread in different patterns across the map.
Among the real-world sources (\cref{fig:real36,fig:real55}),
the Spliddit instances are concentrated near indifference for $3 \times 6$ instances but spread evenly across the map among $5 \times 5$ instances.
The island instances tend to vary between indifference and contention whereas the candy instances tend to lie closer to separability, which suggests that the agents' preferences over survival items are more aligned than the children's preferences over snacks.
Among the synthetic distributions (\cref{fig:synthetic36,fig:synthetic55}), the
i.i.d.\ distributions and attributes model generate concentrated clusters close
to indifference (\Cref{fig:exponstri,fig:unistri,fig:attristri,fig:rsmplstri,fig:spldstri,fig:islstri,fig:candiescompar}
in \cref{app:moreexperiments} display each distribution separately).  
The observation that these synthetic distributions do not cover the range of the
real-world data (and even of just the Spliddit data) raises concerns about the degree to which distributional-analysis results for i.i.d.\ instances (or attributes instances) can be applied to real-world fair division problems.
Since the resampling instances cover the map to a much higher degree, the
resampling distribution appears to be a more fruitful proxy for real-world
instances for future studies.\footnote{In \cref{app:moreexperiments}, we show that resampling instances continue to cover the embedding map for substantially larger instance dimensions ($n=10,m=20$), using the methodology of \cref{sec:demanddistancestillinbody}.}

Next, we discuss how several natural features of allocation instances vary
across the map, which we compute using constraint
programming.
In \cref{fig:envy36,fig:envy55}, we study to which degree the instances allow
for (almost) envy-free allocations.
Specifically, denote an allocation of all goods over the agents by $[m] = S_1
\dotcup S_2 \dotcup \cdots \dotcup S_n$, where $S_i$ denotes agent $i$'s bundle.
The \emph{minimax envy} is the minimum, over all
allocations, of the largest amount by which some agent envies another, i.e.,
$\max_{i \neq i'} u_i(S_{i'}) - u_i(S_i)$.
An instance has envy-free allocations if and only if the minimax envy is at most 0 (we highlight these instances with cross markers).
But the minimax envy gives a gradual measure of how far envy-freeness is from being achievable (or how much it can be overattained).
As we can see, an instance's position on the map is highly informative for the
minimax envy and the existence of envy-free allocations.
For $3 \times 6$ instances, envy-freeness seems to be hopeless near contention
(which is also the case for contention itself) and easy near separability. For
the rest of the map, minimax envy is close to zero, which means that almost envy-free allocations exist widely, and exactly envy-free allocations generally exist below the upper outline of the map.
$5 \times 5$ instances are less hospitable to envy-freeness: envy-free
allocations exist only near the lower border of the map, and the minimax envy
becomes higher (i.e., worse), the further up on the map an instance lies.

Finally, \cref{fig:nash36,fig:nash55} show that the maximum Nash welfare
achievable by any allocation also varies smoothly over the map, increasing the closer an instance lies to separability.
It is noteworthy that this map differs only slightly from the maximum
utilitarian welfare that can be achieved (see \cref{fig:explicit:comparison}
below). We show the distribution of various additional features
in~\cref{app:moreexperiments}.

\subsection{A Faster Distance for Large Instances}
\label{sec:demanddistancestillinbody}
Creating these distance-embedding maps required computing the valuation distance for many pairs of instances.
We were able to do this because the instance dimensions we have focused on, i.e., the dimensions that regularly appear on Spliddit, are rather small.
In general, however, computing the valuation distance is NP-hard (\cref{sec:valuation-distance}), and would be prohibitively slow to compute for, say, instances with dimensions $n=m=10$ (even with an integer linear programming solver). %

To verify that the patterns we described above extend to large (synthetic) instances, and to ready our mapping approach for a future in which larger fair-division problems might be routinely solved,
we define in \cref{sec:demand-distance} the \emph{demand
distance}, a heuristic approximation to the valuation distance.
Crucially, the demand distance between two instances can be computed in
polynomial time by finding a maximum-weight bipartite matching, which is also
fast in practice. Though the demand distance may, in principle, deviate
substantially from the valuation distance, we find that both correlate very
well, with a Pearson correlation
coefficient of at least 97\% across our dimensions (\cref{fig:distances-correlation} in
\cref{sec:valuation-distance} correlation diagrams).
Maps resulting from both distances are virtually
indistinguishable, which is demonstrated by numerous juxtaposition figures
in \cref{app:moreexperiments}.

Using this demand distance, we compute a distance-embedding map for $n=10$ and $m=20$, which we defer to \cref{app:moreexperiments} due to space limitations.\footnote{Since none of the real-world instances are this large, and since i.i.d.\ and attributes instances are even more clustered around indifference, we focus on resampling instances. In \cref{app:bigger-instances-dataset}, we give details about the instances chosen, and show how the model's parameters determine the instances location on the map.}
The broad patterns we observed in $3 \times 6$ and $5 \times 5$ instances continue to hold, which we have also confirmed for even larger instances (\cref{app:large-instances-datasets}).
Creating maps using the demand distance scales readily to larger sizes\,---\,even, say, to $n=m=100$.

\section{Explicit Maps}
\label{sec:explicit}
Generating maps through a distance embedding
entails several inherent disadvantages:
\begin{description}
	\item[Instability.] The distance-embedding map may change non-continuously as
the result of slight changes to the random seed or the set of mapped
instances (though we did not observe this, \cref{app:main-datasets}).
\item[Data dependence.] Suppose that you want to place an allocation
instance on the map to predict its properties. This would require data for all
other instances and computing pairwise distances, which would be more difficult
than directly computing your instance's properties.
\item[Theoretical intractability.] Which instances are ``most
extreme''? Where do instances from a probability distribution lie on
the map? One can answer such questions
empirically, but not
theoretically.
\end{description}

To overcome these challenges, we propose an \emph{explicit map} of fair division instances: a function $\mu$ from allocation instances to $\mathbb{R}^2$, which replicates the general layout of the distance-embedding map.
Specifically, this function maps $n \times m$ utility matrices as follows:
\[ \mu : \mathbb{R}^{n \times m} \to \mathbb{R}^2 \quad \quad U \mapsto \big(\sigma_1(U), \sigma_2(U)\big), \]
where $\sigma_1(U)$ and $\sigma_2(U)$ are the largest and second-largest singular values of the matrix $U$, respectively.
As \cref{fig:explicit:embedding} shows (on the same $5 \times 5$ map as
in~\cref{exp:main_maps}), these two %
values closely capture the vertical and horizontal ordering
of instances in our distance-embedding map, ensuring that the two maps are
closely aligned (\cref{app:fig:explicit:embedding:singvalues} in \cref{app:sec:explicit} shows the corresponding $3 \times 6$ map).
\begin{figure}\centering%
	\def\ratio{.47}
  \begin{subfigure}[t]{.45\linewidth}
		\includegraphics[width=\ratio\linewidth]{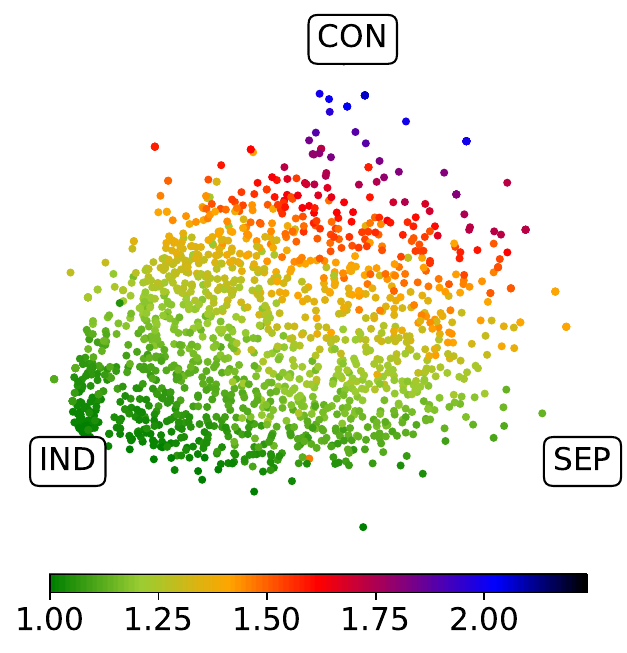}\hspace{.5em}%
	  \includegraphics[width=\ratio\linewidth]{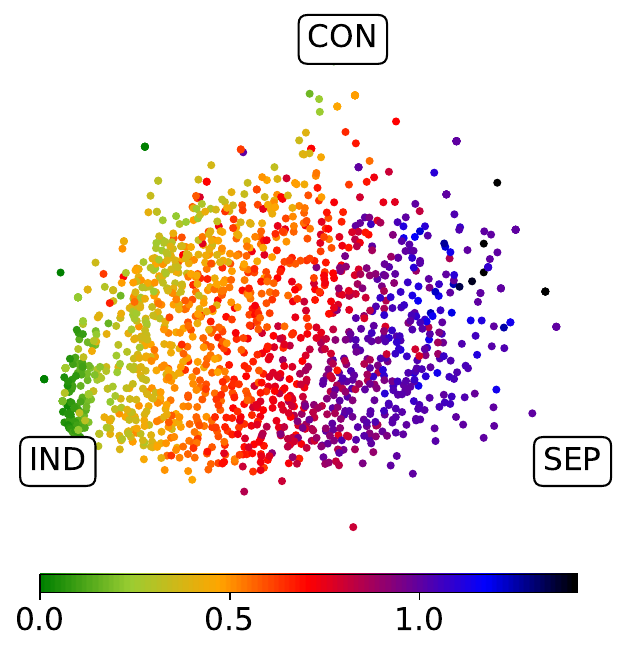}
		\caption{Distribution of $\sigma_1$ (left) and $\sigma_2$ (right) on our
		distance-embedding map for $5 \times 5$.}
    \label{fig:explicit:embedding}
	\end{subfigure}\hfill%
  \begin{subfigure}[t]{.55\linewidth}%
		\hspace{5px}%
		\includegraphics[width=0.41\linewidth, trim={0 0 0 0em}, clip]{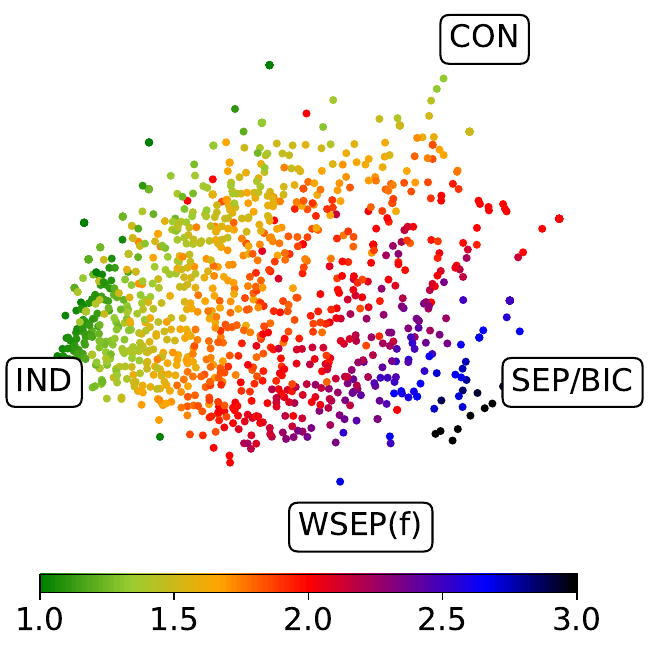}\hfill%
	  \includegraphics[width=0.54\linewidth]{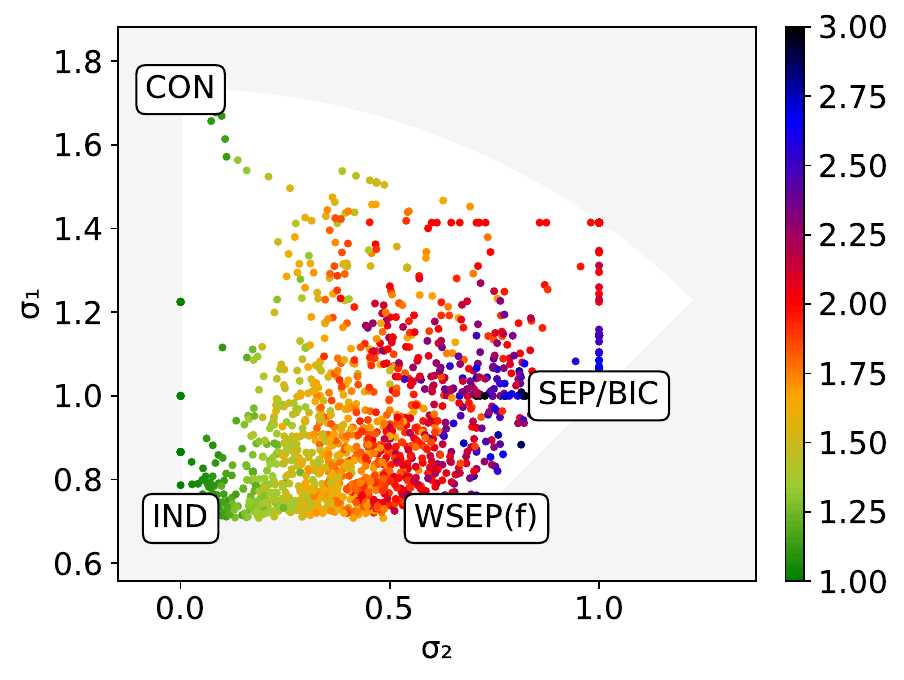}
    \caption{Distribution of max.\ utilitarian welf.\ on our distance-embedding (left) and explicit (right) map, $3 \times 6$ instances.}
	  \label{fig:explicit:comparison}
	\end{subfigure}
	\caption{Distributions of the $\sigma_1$, $\sigma_2$ values and comparison of
		the two maps we introduce.}
	\let\ratio\undefined
\end{figure}
In this section,
we show that the explicit map is similarly informative as the distance-embedding map, while being stable, data independent, and theoretically tractable by design.

\subsection{Demystifying the Singular Value Map}
\label{sec:explicit:demystifying}
We begin by recalling facts about singular values that make them suitable components for our explicit map function.
First, the singular values are invariant under permutations of rows or columns in the utility matrix, so that relabeling agents or goods will not change the map embedding. %
Second, $\sigma_1$ and $\sigma_2$ are %
1-Lipschitz continuous in the entries of the matrix, which
together with the previous point implies that two instances with small valuation distance must be placed near each other on the explicit map.
Third, adding a column of zeros, i.e., a good that no agent values, does not change the singular values, which means that instances can be naturally compared across different $m$.
Finally, implementations of efficient algorithms for computing singular values are readily available (e.g., in numpy), which makes it easy to compute a given instance's position on the map.

We now aim to give the reader an intuition for what information $\sigma_1$ and $\sigma_2$ express about an allocation instance and why.
We begin with $\sigma_1$, which can be expressed as 
\begin{equation}
\sigma_1 = \max_{\vec{v}_1 \in \mathbb{R}^m, \norm{\vec{v}_1}=1} \norm{U \, \vec{v}_1},
\label{eq:sigma1rayleigh}
\end{equation}
where $\norm{\cdot}$ is the Euclidean ($\ell^2$) norm.
Since we rarely think about utility matrices as linear functions over unitary vectors, it is instructive to pretend
that the norms in \cref{eq:sigma1rayleigh} were $\ell^1$-norms.
In this case (choosing $\vec{v}_1$ nonnegative w.l.o.g.), the $U \, \vec{v}_1$ being optimized over are the convex combination of $U$'s columns, for the coefficients given by $\vec{v}_1$.
If we were indeed maximizing the $\ell^1$-norm of $U \, \vec{v}_1$, $\sigma_1$ would be the largest column sum, or \emph{maximum demand}.
Though the $\ell^2$ norm slightly complicates the picture,\footnote{It gives an advantage to combinations of columns in which several columns have positive coefficients, and it encourages making a few coordinates of $U \, \vec{v}_1$ large rather than all.}
$\sigma_1$ and the maximum demand are very highly correlated:
across our $3 \times 6$ instances, for example, the correlation coefficient is $97\%$. %
Thus, $\sigma_1$ can be understood as good approximation of the maximum demand, up to shifting and rescaling.

To interpret the second-largest singular value $\sigma_2$, we recall how the
singular value decomposition of an $\mathbb{R}^{n \times m}$ matrix $U$ can be
used to find a low-dimensional embedding of the row vectors (in our case, the
agents' utility vectors).\footnote{See Chapter 3 by~\citet{bhk20} for a detailed
	explanation. Though singular values are closely connected to dimensionality
	reduction, our use is non-standard.  Applying value decomposition directly to
	find a 2D embedding of utility matrices would result in embeddings highly
	sensitive to row and column permutations and would thus not be fruitful.  One way to
	understand the discussion above is that we map each utility matrix to the
	square roots of the top-two eigenvalues in its principal component analysis;
	except that we do not shift column sums to zero, since this would, e.g., make
	IND and CON indistinguishable.  }
For example, the line through the origin $\thespan{\vec{v}_1}$, spanned by the argmax of \cref{eq:sigma1rayleigh}, is the best 1-dimensional space to embed the rows in, in the following sense: if we sum up, for each row $\vec{u}_i \in \mathbb{R}^m$, the squared length of its projection onto this space%
,
$\thespan{\vec{v}_1}$ maximizes this sum across all 1-dimensional subspaces.
In fact, this sum of squared projection lengths is $\sigma_1^2$, which means that $\sigma_1$ measures ``how much'' of the row vectors can be captured by a 1-dimensional embedding.
Similarly, $\sigma_2$, which can be calculated as
\[ \max_{\vec{v}_2 \in \mathbb{R}^m, \norm{\vec{v}_2}=1, \vec{v}_2 \perp \vec{v}_1} \norm{U \, \vec{v}_2}, \]
measures how much the row embedding improves when going from the optimal 1-dimensional space $\thespan{\vec{v}_1}$ to the optimal 2-dimensional space $\thespan{\vec{v}_1, \vec{v}_2}$.

Thus, as a first approximation, $\sigma_2$ measures how diverse the agents' utilities are.
It is zero if all agents have the same utility vector, and large when there are
blocks of agents that completely disagree on which goods have nonzero value.
To again find a more elementary correlate, we define an instance's
\emph{preference diversity} as the mean $\ell^2$ distance between utility
vectors, averaged over all pairs of agents in the instance. Again, we find a very high
correlation (96\% correlation coefficient for $3 \times 6$).

\subsection{Theoretical Properties of the Map}
\label{sec:explicit:theory}

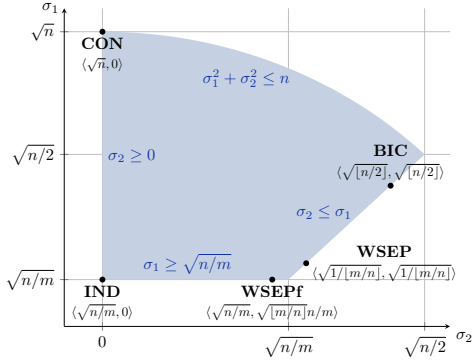
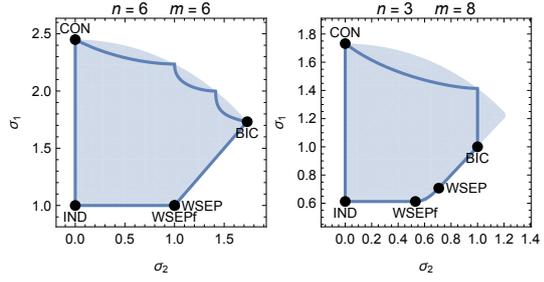
\begin{figure}%
\begin{subfigure}[b]{.5\linewidth}
  \centering%
  \resizebox{.9\linewidth}{!}{\begin{tikzpicture}[
    declare function={
      n=5;
      m=6;
      z=0.25;
      eps=0.001;
      rootnm = sqrt(n/m);
      rootn = sqrt(n);
      rootnover2 = sqrt(n/2);
  },
  constraint/.style={labelblue},
  dotcircle/.style={radius=1.5pt},
  scale=1.2]
  
  \pgfmathsetmacro\WSEPx{sqrt(1/floor(m/n))}
  \pgfmathsetmacro\WSEPfx{sqrt(floor(m/n)*n/m)*rootnm}
  \pgfmathsetmacro\BICx{sqrt(floor(n/2))}
  \pgfmathsetmacro\WSEPy{\WSEPx}
  \pgfmathsetmacro\WSEPfy{rootnm}
  \pgfmathsetmacro\BICy{\BICx}
  \pgfmathsetmacro\CONx{0}
  \pgfmathsetmacro\CONy{rootn}
  \pgfmathsetmacro\INDx{0}
  \pgfmathsetmacro\INDy{rootnm}
  \pgfmathsetmacro\midcirclex{cos(67.5)*sqrt(n)}
  \pgfmathsetmacro\midcircley{sin(67.5)*sqrt(n)}
  
  \begin{axis}[
    name=main plot,
    axis lines = left,
    xmax=rootnover2+z/2,xmin=-3/4*z,
    ymax=rootn+z/2,ymin=rootnm-z,
    xtick = {0,rootnm,rootnover2},
    ytick = {rootnm,rootnover2,rootn},
    xticklabels = {0, {$\sqrt{n/m}$}, {$\sqrt{n/2}$}},
    yticklabels = {$\sqrt{n/m}$, {$\sqrt{n/2}$}, $\sqrt{n}$},
    domain=-3/4*z:rootnover2+z/2,
    samples=300,
    clip=false,
    grid=both
  ]
  
  \addplot[draw=none,name path = circle] {sqrt(n-x^2)};
  \addplot[draw=none,name path = x] {max(x,rootnm)};
  
  \addplot [shadingblue] fill between [of=circle and x, soft clip={domain=0:rootnover2}];
  
  \fill [black] (\WSEPx, \WSEPy) circle[dotcircle];
  \node [right, align=center] at (\WSEPx, \WSEPy) {\textbf{WSEP} \\ \scalebox{.8}{$\langle \sqrt{1 / \lfloor m/n \rfloor}, \sqrt{1 / \lfloor m/n \rfloor} \rangle$}};
  
  \fill [black] (\WSEPfx, \WSEPfy) circle[dotcircle];
  \node [below, align=center] at (\WSEPfx, \WSEPfy) {\textbf{WSEPf} \\ \scalebox{.8}{$\langle \sqrt{n/m}, \sqrt{\lfloor m/n \rfloor} n/m\rangle$}};
  
  \fill [black] (\BICx, \BICy) circle[dotcircle];
  \node [above, align=center] at (\BICx, \BICy) {\textbf{BIC} \\ \scalebox{.8}{$\langle \sqrt{\lfloor n/2 \rfloor}, \sqrt{\lfloor n/2 \rfloor} \rangle$}};
  
  \fill [black] (\CONx, \CONy) circle[dotcircle];
  \node [below, align=center] at (\CONx, \CONy) {\textbf{CON} \\ \scalebox{.8}{$\langle \sqrt{n}, 0 \rangle$}};
  
  \fill [black] (\INDx, \INDy) circle[dotcircle];
  \node [below, align=center] at (\INDx, \INDy) {\textbf{IND} \\ \scalebox{.8}{$\langle \sqrt{n/m}, 0 \rangle$}};
  
  \node [constraint,right] at ($(\INDx, \INDy)!0.5!(\CONx, \CONy)$) {$\sigma_2 \geq 0$};
  \node [constraint,above] at ($(\INDx, \INDy)!0.5!(\WSEPfx, \WSEPfy)$) {$\sigma_1 \geq \sqrt{n/m}$};
  \node [constraint,above left,xshift=.2cm] at ($(\BICx, \BICy)!0.5!(\WSEPx, \WSEPy)$) {$\sigma_2 \leq \sigma_1$};
  \node [constraint,below,xshift=-.65cm] at (\midcirclex, \midcircley) {$\sigma_1^2 + \sigma_2^2 \leq n$};
  
  \end{axis}
  
  \node [below right] at (main plot.right of origin) {$\sigma_2$};
  \node [left,xshift=-.8cm] at (main plot.above origin) {$\sigma_1$};
  
  \end{tikzpicture}}
\caption{Bounding inequalities of the map, and locations $\langle \sigma_1,
\sigma_2 \rangle$ of characteristic instances. %
}
\label{fig:explicit:theory}
\end{subfigure}\hfill%
\begin{subfigure}[b]{.5\linewidth}
  \includegraphics[width=\linewidth]{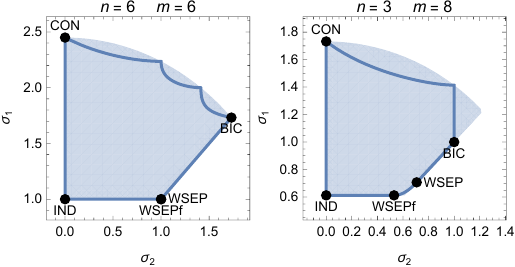}
    \caption{Explicit map for $\langle n,m \rangle = \langle 6,6 \rangle$ and $\langle 3,8 \rangle$. By \cref{thm:explicit:e,thm:explicit:s,thm:explicit:w,thm:explicit:n}, the map is contained in the shaded area. Lines trace interpolations between named instances (see \cref{app:explicit:theory}).}
    \label{fig:explicit:outlines}
\end{subfigure}%
\caption{The general shape of the explicit map (left) and those for selected
values of $n$ and~$m$ (right).}
\end{figure}
We now theoretically characterize the image of our map function $\mu$ for given dimensions $n, m$.
Our task\,---\,characterizing the combinations of singular values in stochastic
rectangular matrices\,---\,is of interest independently to our fair-division
setting, but, to our knowledge, has not previously been undertaken.
This process will give us a more precise understanding of what makes instances extreme along either dimension of our map.
\Cref{fig:explicit:theory} summarizes both the outlines of the map and the positions of characteristic instances, which can guide the reader through this section.
We orient on the page such that $\sigma_1$ grows in the ``North'' and $\sigma_2$ in the ``East'' direction, which by \cref{fig:explicit:embedding} generally aligns with how we have presented the distance-embedding map.

Whereas $\CON$ and $\IND$ still mark the left corners of our map, the other two corners lead us to new characteristic instances.
For the lower-right corner, we refine our definition of separability since $\SEP$ (with $\sigma_1 = \sigma_2 = 1$) only lies on the lower boundary if $n=m$.
If $m$ is a proper multiple of $n$, the lower-left corner is instead inhabited by \emph{wide separability}, in which every agent values $m/n$ disjoint goods, giving equal value $n/m$ to each of them.
In \cref{app:explicit:theory:wsep}, we extend wide separability to $n
\mathrel{\not|} m$ in two slightly different ways: one, $\WSEP$, always lies on
the right border while the other, $\WSEPf$ always lies on the lower border.

The final characteristic instance is \emph{bicontention} ($\BIC$); here, half of
the agents place all utility on one common good and half of the agents on a
second common good (for odd $n$, one agent places all value on a third good).
Since this instance combines highly demanded goods with sharply distinct utility
vectors, it always lies on the right border and, for even $n$, is exactly
located in the upper-right corner.

The main results of this section address all four sides of the map.
For each side, we bound the map by an inequality and show that the inequality is sharp using our characteristic instances.
For three of the sides, we give simple, necessary-and-sufficient conditions for an instance lying on the boundary.
If $n$ is even and divides $m$, as in the left subplot of \cref{fig:explicit:outlines},
our characteristic instances lie exactly in the four corner points of the map, and we can exactly trace three of the four sides by interpolating between corner instances.
If these divisibility conditions do not hold, as illustrated in~\cref{fig:explicit:theory} and the right subplot of \cref{fig:explicit:outlines},
the characteristic instances lie in the corner up to rounding terms.
Proofs of our characterizations tend to be short and cute, but are deferred to \cref{app:explicit:theory:proofs}.

\begin{restatable}[``West'']{theorem}{thmexplicitw}
\label{thm:explicit:w}
$\sigma_2$ is at least $0$.
An instance lies on this boundary iff all agents have the same utility vector.
In particular, $\IND$, $\CON$, and their convex combinations lie on this boundary.
\end{restatable}
\begin{restatable}[``South'']{theorem}{thmexplicits}
\label{thm:explicit:s}
$\sigma_1$ is at least $\sqrt{\nicefrac{n}{m}}$.
An instance lies on this boundary iff all columns of its utility matrix have an
equal sum (namely, $\nicefrac{n}{m}$).
In particular, $\IND$, $\WSEPf$, and their convex combinations lie on this boundary.
\end{restatable}
\begin{restatable}[``North'']{theorem}{thmexplicitn}
\label{thm:explicit:n}
$\sigma_1$ is at most $\sqrt{n - \sigma_2^2} \leq \sqrt{n}$.
An instance lies on this boundary iff each agent values a single good, and if at most two goods are valued by any agent.
In particular, $\CON$ and, if $n$ is even, $\BIC$ lie on this boundary.
\end{restatable}

\begin{restatable}[``East'']{theorem}{thmexplicite}
\label{thm:explicit:e}
$\sigma_2$ is at most $\sigma_1$.
If\ \ $U$, after row and column permutation, has the block matrix structure
\scalebox{.8}{$\begin{pmatrix}
A & 0 & 0\\
0 & A & 0\\
0 & 0 & B
\end{pmatrix}$}
for rectangular matrices $A, B$ and $\sigma_1(A) \geq \sigma_1(B)$, this is sufficient for lying on the boundary.
(If $B$ has height 0, we set $\sigma_1(B) = 0$.)
In particular, $\WSEP$, $\BIC$, and a suitable interpolation lie on this boundary.
\end{restatable}

We conclude the theoretical discussion by pointing out that existing and future
results in the theory of nonnegative random matrices have implications for our
explicit map.
For example, consider a random process in which a single utility vector is drawn
from a flat Dirichlet distribution and
duplicated for all agents (thus,
$\sigma_2 = 0$).
For this distribution over instances, \citet{cfv22} recently derived that
$\mathbb{E}[\sigma_1^2] = \nicefrac{2 \, n}{(n+1)}$ as well as formulas for
$\sigma_1^2$'s higher moments.
\citet{bdh22} study a random process, in  which, for fixed integers $d_2 \geq
d_1 \geq 3$, an instance is uniformly chosen in which each agent values $d_1$
goods at value $\nicefrac{1}{d_1}$, and each good is valued by $d_2$ agents.
In this case,
$\sigma_1$ is always
$\sqrt{n/m}$, and the authors show that, as $m,n \to \infty$,
$\sigma_2$ converges to $\nicefrac{(\sqrt{d_1-1} + \sqrt{d_2-1})}{d_1}$ in
probability.

\subsection{Comparison of the Maps}
\label{sec:explicit:empirical}

Comparing the explicit map to our distance-embedding map (see, e.g.,
\cref{fig:explicit:comparison} for the largest achievable utilitarian welfare), we see that the two maps
have a similar layout and communicate similar information overall.
In \cref{app:moreexperiments}, we provide extensive diagrams showing that this similarity extends to other features, the identifiability of instance sources, and to the $5 \times 5$ instances as well.

\looseness=-1
One major difference is in how the density of instances varies across both maps.
Whereas the distance-embedding map fills the map at a rather uniform density, which helps legibility, the explicit map clusters some instances very densely (e.g., near the South boundary and the $\sigma_2 = 1$ line in \cref{fig:explicit:comparison}).
But these dense areas of the map seem to highlight meaningful clusters of similar instances, given that instance features tend to be homogeneous within these dense areas.
The shape of instances in the explicit map can similarly highlight noteworthy patterns.
For example, the straight lines at $\sigma_1 = \sqrt{2}$ and $\sigma_2 = 1$ in \cref{fig:explicit:comparison} are formed by instances in which several agents only value one good (see \cref{fig:all:fracsingleminded} in \cref{app:moreexperiments}).
The distance-embedding map makes such phenomena much harder to spot.

Hence, and because of the advantages of stability, data independence, and theoretical tractability, we see the explicit map as broadly preferable over the distance-embedding map on our data.
Nevertheless, the distance-embedding map plays a crucial role by justifying the explicit map: the relevance of the explicit map rests in large part on the fact that a general, previously established approach surfaced the two largest singular values as the most salient dimensions of difference between instances.

\section{Conclusion}
\label{sec:conclusion}
We hope that our exploration of allocation instances initiates discussions about
which
assumptions on such instances are supported by practice, and how
fair-division theory can leverage these assumptions to provide algorithms with
stronger fairness properties for the bulk of practical allocation instances.

The main limitation of our study is that\,---\,despite tapping into unconventional data sources\,---\,we were unable to test our approach on large, real-world allocation instances.
This seems rooted in a broader limitation of the practice of fair division: large allocation problems are hardly ever solved, or the preference data are not made available.
We believe that our community should strive to collect and share such datasets,
as has been recently done for election data \citep{mat-wal:c:preflib}.

\bibliographystyle{abbrvnat}
\bibliography{bib}

\clearpage

\appendix

\section*{\Large APPENDIX}
\section{Details on Datasets and Experiments}\label{app:datasets}
We combined various statistical cultures and the real-world data to construct
datasets that we focus on in the paper. The real-world datasets are available
upon request from the authors of the corresponding works cited
in~\cref{sec:prelim:realdata}, from whom we also got permission to use the data
in our study.

\subsection{Datasets $5 \times 5$ and $3 \times 6$}\label{app:main-datasets}
For the $5 \times 5$ dataset consisting of instances with $5$~agents and
$5$~goods, we generated $40$~instances according to:
the attributes models with $2$ and $5$ attributes; the resampling 
model with all values~$\{0.2, 0.4, 0.6, 0.8\}$ of parameter~$p$
and all values~$\{0.2, 0.8\}$ of parameter $\phi$. Next, we generated
$40$~instances with i.i.d. valuations taking the uniform distribution over $[0, 1]$ and the
same number of instances with i.i.d. valuations with the exponential
distribution (recall that here the distribution’s rate does not change the
outcome due to rescaling). Furthermore, we added
$500$~instances sampled (as described in the previous section) from the Island
data, $500$~instances from the Candies data, and all $16$ Spliddit instances.
Finally, we put the respective $\CON$, $\IND$, $\SEP$, $\WSEP$, and $\BIC$ instance.
Analogously, we constructed the $3 \times 6$
dataset consisting of instances with $3$~agents and $6$~goods. The only
exception being that in the $3 \times 6$
dataset, we took $250$ instances each of the
Island, Candies, and Spliddit data, which is possible since there were enough Spliddit instances of this size.

We generated each of the datasets multiple times (note that
generating real-life inspired data is a random process) and repeated all our
experiments. The obtained results were qualitatively the same.

\subsection{Dataset $10 \times 20$}\label{app:bigger-instances-dataset}

To verify that the instances coming from the resampling model cover diverse areas
of our maps, we also constructed a dataset consisting of~$n=10$~agents
and~$20$~goods. In this dataset, besides the CON, IND, and SEP instances, we
generated $4$~instances per each combination of parameters $p \in
\{0.05, 0.1, 0.2, 0.4, 0.6, 0.8\}$ and~$\phi \in \{0.05, 0.1, 0.25, 0.5, 0.75,
0.9, 0.95\}$. We did not include instances sampled from real-world distributions
as none of them have more than~$10$~goods, and cloning them could change the
structure of the instances in an unpredictable manner.
As our goal was to cover the whole
space, we also omitted the i.i.d\ and attributes models since they tend
to cluster in very specific and small areas of the map. Indeed, this behavior
can already be observed for the~$3 \times 6$ and~$5 \times 5$~datasets
in~\cref{fig:exponstri,fig:unistri,fig:attristri} in~\cref{app:moreexperiments}.
Note that numerous juxtapositions of our datasets in~\cref{app:moreexperiments}
present that the $10 \times 20$~dataset yields a similarly structured
distribution of values of natural features across the map, as compared to our
canonical~$3 \times 6$ and~$5 \times 5$ datasets. Since computing the valuation
distance for instances in the~$10 \times 20$~dataset is computationally too
demanding, for this dataset we only computed maps using the demand distance. 

In \cref{fig:resamplingparameters10x20}, we show how each value of $p$ and $\phi$ traces a shifting band on the map, which enables the distribution to cover the entire distance-embedding map.
\begin{figure*}[htb]
\centering
\includegraphics[width=.35\textwidth]{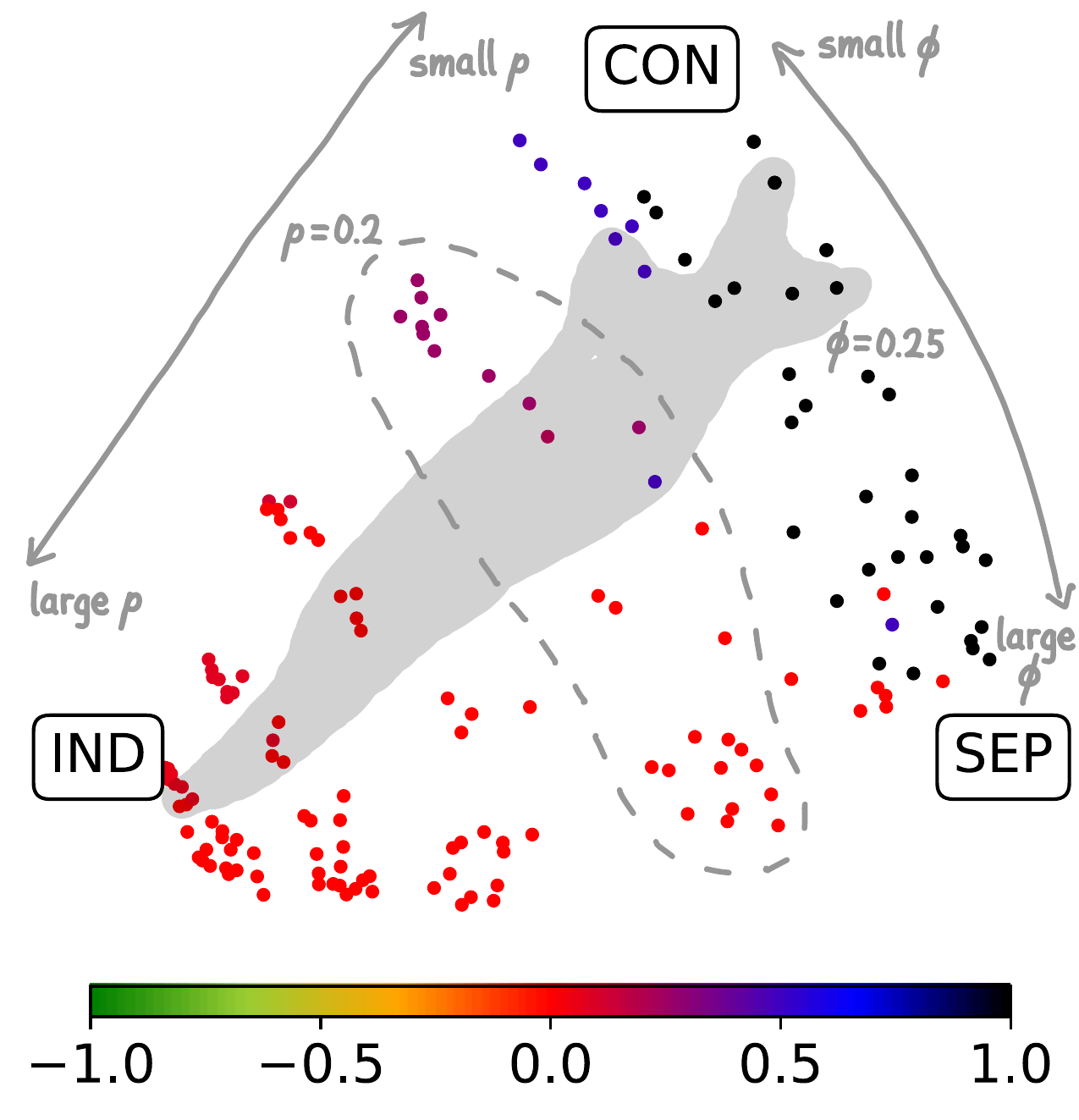}
\caption{Distance-embedding map for the demand distance, containing $10 \times 20$ instances drawn from the resampling model. Colors indicate minimax envy. Hand-written annotations show placement of resampling instances for different $p, \phi$.}
\label{fig:resamplingparameters10x20}
\end{figure*}

\subsection{Larger Instances}\label{app:large-instances-datasets}

Finally, to check scalability of the demand distance, we also conducted our
experiments on data with larger numbers of goods and agents.
We computed the maps for maximum instance sizes of $n=m=300$ on a standard
laptop computer\footnote{Specifically, we used two machines for the experiments
	we present:
	Apple\textregistered{} MacBook Pro\textregistered{}
	(Apple M2 Max, 64GB RAM) and HP\textregistered{} ProBook 650 G8 (11th Gen
	Intel\textregistered{} Core\texttrademark{}
i5-1135G7 @ 2.40GHz, 32GB RAM). We computed the embeddings on the first one and
the feature values on the second one.
In the course of the full research project we also used 
an Ubuntu 22.04.4 LTS server (Intel\textregistered{} Xeon\textregistered{}
Silver 4310 CPU @ 2.10GHz with 12 physical cores, 128GB RAM). We, however, do
not present any results obtained using this machine.}. For example, computing the demand
distance between a pair of them took, on average, $\nicefrac{1}{11}$ seconds and
$1.2$~seconds for, respectively instance sizes~$n=m=100$ and~$n=m=300$. This
gives, respectively, around $20$~minutes and slightly above $3.5$~hours to
compute maps of the datasets analogous to the above-described $10 \times 20$
one. We emphasize that the running time of our algorithm for computing the
demanding distance depends quadratically on the number of instances in the
dataset and nonlinearly only on the \emph{minimum} of~$n$ and~$m$. Hence, the
running times grow roughly linearly for larger instance sizes, as long
as~$\min(n,m) = 300$ (we verified it for instances with the number of either
agents or goods being~$1000$, which increased the running time only at most
twice). The given running times do not include computing the values of
features\,---\,studying effective and scalable ways of computing them is not a
goal of our study. 

The results were qualitatively the same, so not to overload the paper with
more pictures, we omit the computed maps.

\section{Deferred Details from
Section~\ref{sec:prelim}}\label{app:sec:prelimdetails}

\subsection*{Valuation Distance}

Consider a set~$[n]$~of agents, a set~$[m]$~of resources and some allocation
instances~$U \in \mathbb{R}_{\geq 0}^{n \times m}$ and~$U' \in \mathbb{R}_{\geq
0}^{n \times m}$.
Denote by~\allPermAgents{} and~\allPermgoods{} the sets of, respectively, all
permutations $[n] \to [n]$ and $[m] \to [m]$. Then, consider
some $\agentsperm \in \allPermAgents$ and some~$\resperm \in \allPermgoods$,
which we call, respectively, an \emph{agent matching} and a~\emph{good
matching}.
For some distance~$\delta$ on nonnegative real numbers, we let
$$
D_\delta(\task, \task', \agentsperm, \resperm) \coloneqq \sum_{i \in
[n]}\sum_{j \in [m]} \delta(u_{i,j},
u_{\agentsperm(i),\resperm(j)})
$$
and refer to~$D_\delta(\task, \task', \agentsperm, \resperm)$ as the
$\delta$-distance between~$\task$ and~$\task'$ \emph{witnessed by~$\agentsperm$
and~$\resperm$}. The $\delta$-distance between~$\task$ and~$\task'$, denoted
by~$d_{\delta}(\task, \task')$, is then the minimal $\delta$-distance
between~$\task$ and~$\task'$ witnessed over all pair of matchings; formally:
$$
d_\delta(\task, \task')
\coloneqq \min_{\agentsperm' \in \allPermAgents, \resperm' \in
\allPermgoods} D_\sigma(\task, \task', \agentsperm', \resperm'). 
$$

We are now ready to formally define our \emph{valuation distance} that,
intuitively, is the smallest sum of difference in agents valuations of the goods
achievable over all possible matchings of agents and goods.
\begin{definition}
	Given two allocation instances~$\task \in \mathbb{R}_{\geq 0}^{n \times m}$
	and~$\task' \in \mathbb{R}_{\geq 0}^{n \times m}$ with $n$~agents and $m$~goods,
	its \emph{valuation distance} $\idealdist(\task, \task')$ is:
	\begin{align*}
		\idealdist(\task, \task')& \coloneqq d_{\ell_1}(\task, \task') \coloneqq\\
														 & \min_{\agentsperm, \resperm} \sum_{i \in [n]}
														 \sum_{j \in [m]}
														 \left| u_{i,j} - u'_{\agentsperm(i),\resperm(j))}\right|. 
		\end{align*}
\end{definition}

It is easy to see that that the valuation distance is isomorphic.  Naturally, if
the tasks are isomorphic, then there exists some pair of agent and good
matchings that witness distance~$0$. On the other hand, if there is no such
pair, there is no possibility that the valuation distance is~$0$. The property
of being an isomorphic distance, however, comes at a cost of computational
intractability.

\subsection{Valuation Distance Hardness}\label{sec:valuation-distance}

\begin{theorem}
	Given two task allocations~$\task \in \mathbb{R}_{\geq 0}^{n \times m}$
	and~$\task' \in \mathbb{R}_{\geq 0}^{n \times m}$ and an real number~$d$,
	deciding whether~$\idealdist(\task, \task') \leq d$ is~\nphard{}.
 \label{thm:valuationdistnphard}
\end{theorem}
\begin{proof}
	We give a polynomial-time
	many-one reduction from an \nphard{} problem
	$d_{\textrm{Spear}}$\textsc{-Isomorphic Distance}. In this problem we
	are given two ordinal elections~$E = (C, V)$ and~$E'=(C', V')$ such
	that $|C| = |C'|$ and~$|V| = |V'|$ and an integer~$k$. 
    Assuming that
	for some voter~$a$ and candidate~$b$, where both $a$ and~$b$ are part
	of the same election~$E$, we denote by~$\textrm{pos}^{E}_{a}(b)$ the position
	of candidate~$b$ according to the ranking of~$a$, we ask whether there
	exist two permutations~$\rho \colon C \rightarrow C'$ and~$\phi \colon
	V \rightarrow V'$ such that 
	$$
	D(\rho, \phi) \coloneqq  \sum_{v \in V}
	\sum_{c \in C} \left| \textrm{pos}^{E}_v(c) -
	\textrm{pos}^{E'}_{\phi(v)}(\rho(c)) \right| \leq k. 
	$$

	Given the instance $I$ of $d_{\textrm{Spear}}$\textsc{-Isomorphic
	Distance} as described above, our reduction constructs an instance~$I'$
	of our problem as follows.  We first construct allocation instance~$\task$ using
	election~$E$ from the original instance. Allocation instance~$\task$ consists
	of $n \coloneq |V|$~agents $[n]$ representing voters and $m \coloneq
	|C|$~goods~$[m]$. Thus, $\task$ is a matrix of dimension $n \times m$.
	Taking a normalizing factor~$F = 1 + 2 + \ldots + |C| = \binom {|C|}
	{2}$, for each voter~$v_i \in V$ and candidate~$c_j \in C$, we set the
	corresponding agent $i$'s utility for good~$j$ to be $u_{i,j} =
	\nicefrac{\textrm{pos}^{E}_{v_i}(c_j)}{F}$. It can be easily verified that the
	values of the utility function of each agent in~$\task$ (that is, the values
	of each row of~$\task$) sum to~$1$. We obtain instance~$I'$, by analogously
	constructing allocation instance~$\task'$ using election~$E'$ and setting the
	distance~$d$ in question (regarding instance~$I'$) to~$d \coloneq
	\nicefrac{k}{F}$.

	We show that for each pair of permutations~$\rho \colon C \rightarrow
	C'$ and~$\phi \colon V \rightarrow V'$ such that $D(\rho, \phi)  \leq k$,
	there are two permutations \agentsperm{} and~\resperm{} such that witness that
	$\idealdist(\task, \task') \leq d$. Since we also show that the opposite
	direction is true, we obtain the reduction's correctness. 

	Suppose that we have $\rho$ and~$\phi$ that meet the above assumption.
	Consider the following~$\agentsperm$ and~$\resperm$. For each voter~$v_i
	\in V$ and candidate~$c_j \in C$, let~$\agentsperm(i) = i'$ such that
	$\phi(v_i)= v'_{i'}$
	and~$\resperm(j) = j'$ such that $\phi(c_j)=c'_{j'}$. In words, permutation~\agentsperm{}
	maps agents exactly as permutation~$\phi$ maps their respective voters,
	and so does permutation~\resperm{} with respect to goods and
	candidates. Now, in the series of algebraic transformations, let us
	analyze the relation of $D(\rho, \phi)$ and $d$:
	\begin{align*}
	\nicefrac{D(\rho, \phi)}{F} = \frac{\sum_{v_i \in V} \sum_{c_j \in
		C} \left| \textrm{pos}^{E}_{v_i}(c_j) - \textrm{pos}^{E'}_{\phi(v_i)}(\rho(c_j))
	\right|}{F} = \\
	\sum_{v_i \in V} \sum_{c_j \in C} \left|
	\frac{\textrm{pos}^{E}_{v_i}(c_j)}{F} - \frac{\textrm{pos}^{E'}_{\phi(v_i)}(\rho(c_j))}{F}
	\right| = \\
	\sum_{v_i \in V} \sum_{c_j \in C} \left| u_{i,j} -
	u'_{\agentsperm(i),\resperm(j)} \right| = \\
	\sum_{i \in [n]} \sum_{j \in [m]}
	|u_{i,j} - u'_{\agentsperm(i),\resperm(j)}| = d.
  \end{align*} 
So,
clearly, if $D(\rho, \phi)  \leq k$, then $\idealdist(\task, \task')$ witnessed
by~$\agentsperm$ and~$\resperm$ is smaller than~$\nicefrac{k}{F} = d$. On the
other hand, if there exist~$\agentsperm$ and~$\resperm$ that witness
$\idealdist(\task, \task') \leq d$, then one can construct $\rho$ and~$\phi$
for which $D(\rho, \phi)  \leq dF = k$.
\end{proof}

\begin{figure}
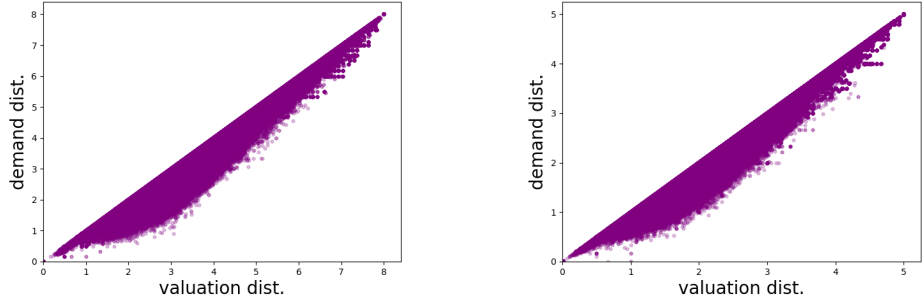

  \includegraphics[width=.48\linewidth]{img/correlation\_5x5\_large\_1\_ijcai.png}
  \includegraphics[width=.48\linewidth]{img/correlation\_3x6\_large\_2\_ijcai.png}
	\caption{Correlation between our distances for the $5 \times 5$ dataset (left)
		and the $3 \times 6$ dataset (right).}
  \label{fig:distances-correlation}
\end{figure}

The computational hardness of the task of computing the valuation distance
comes from the fact that one seeks an optimal value depending on two matchings
simultaneously. It turns out that this intuitive understanding can be well
supported by a formal claim. We show that for a given either the agent matching
or the good matching, the optimal value of the distance as witnessed by the
given matching can be computed in polynomial-time. 

\begin{theorem}
	Given two allocation instances~$\task{}$, $\task'{}$, a (fixed) agent
	matching~$\agentsperm$, a real number~$d$, deciding whether~$\idealdist(\task,
	\task')$ as witnessed by~$\agentsperm$ is at most $d$ is polynomial-time
	solvable.  The same holds for the case of a given good matching.
\end{theorem}
\begin{proof}
	Let us fix numbers~$n$ and~$m$ of, respectively, agents and goods.  For two
	allocation instances~$\task \in \mathbb{R}_{\geq 0}^{n \times m}$ and $\task'
	\in \mathbb{R}_{\geq 0}^{n \times m}$
	and an agent matching~$\agentsperm\colon [n] \to [n]$, we give a
	polynomial-time algorithm that computes a good matching~$\resperm\colon [m]
	\to [m]$ minimizing
	$$
	D(\resperm) \coloneqq  \sum_{i \in [n]} \sum_{j \in [m]} \left|
	u_{i,j} - u_{\agentsperm(i),\resperm(j)}\right|.
	$$
	In words, the algorithm computes the minimal achievable distance as witnessed
	by the given agent matching~$\agentsperm$.

	The algorithm constructs a complete bipartite weighted graph~$G$ consisting
	of vertices~$x_1, x_2, \ldots, x_{m}$ of one partition and
	consisting of vertices~$y_1, y_2, \ldots, y_{m}$ of the
	other partition. For each pair~$i \in [m]$, $j \in
	[m]$, the weight~$w(\{x_i,y_j\})$ of edge~$\{x_i, x_j\}$
	is equal to~$\sum_{\ell \in [n]}\left|u_{\ell,i} -
	u_{\agentsperm(\ell),j}\right|$. Finally, the algorithm looks for a
	minimum weight perfect matching (which always exists) in~$G$.

	Let~$M$ be some perfect matching in~$G$. Clearly, this perfect matching
	corresponds to exactly one good matching~$\resperm' \colon [n]\to[n]$. Let us
	now compute the weight~$w(M)$ of~$M$:
	\begin{align*} 
		w(M) =& \sum_{\{x_i,y_j\} \in M}\sum_{\ell \in [n]}\left|u_{\ell, i} -
		u_{\agentsperm(\ell),j}\right| = \\
					&\sum_{\ell \in [n]}\sum_{\ell' \in [m]}\left|u_{\ell, \ell'} -
					u_{\agentsperm(\ell), \resperm'(\ell')}\right| = D(\resperm').
	\end{align*}
	Since our algorithm finds the minimum-weight matching, the correctness
	follows.

	The algorithm runs in polynomial time because finding a minimum-weight
	matching is well-known polynomial-time solvable task and building	the
	bipartite graph is quadratic with respect to the number of goods (which is
	polynomially bounded in the input size). The proof for the case of a given
	good matching is analogous.
\end{proof}

\begin{figure*}
	\begin{minipage}{.33\textwidth}%
    \begin{align*}
    \IND_m & \coloneqq {\scriptsize\begin{bmatrix}
        \nicefrac{1}{m} & \nicefrac{1}{m} & \cdots & \nicefrac{1}{m}\\
        \nicefrac{1}{m} & \nicefrac{1}{m} & \cdots & \nicefrac{1}{m}\\
        \vdots & \vdots & \ddots & \vdots\\
        \nicefrac{1}{m} & \nicefrac{1}{m} & \cdots & \nicefrac{1}{m}
        \end{bmatrix}}
   \end{align*}%
	\end{minipage}\hfill%
	\begin{minipage}{.3\textwidth}
    \begin{align*}
    \SEP_m & \coloneqq {\scriptsize\begin{bmatrix}
    1 & 0 & \cdots & 0\\
    0 & 1 & \cdots & 0\\
    \vdots & \vdots & \ddots & \vdots\\
    0 & 0 & \cdots & 1
    \end{bmatrix}}
   \end{align*}
	\end{minipage}\hfill%
	\begin{minipage}{.3\textwidth}
    \begin{align*}
    \CON_m & \coloneqq {\scriptsize\begin{bmatrix}
    1 & 0 & \cdots & 0\\
    1 & 0 & \cdots & 0\\
    \vdots & \vdots & \ddots & \vdots\\
    1 & 0 & \cdots & 0
    \end{bmatrix}}
   \end{align*}
	\end{minipage}
	\caption{Matrices representing characteristic instances: Indifference (IND),
	Separability (SEP), and Contention (CON).}
  \label{fig:characteristic-instances}
\end{figure*}

\subsection{Demand Distance and Distances Properties}\label{sec:demand-distance}
In this section, we formally introduce the~\emph{demand
distance} and we show the properties of both the demand and the valuation
distances.

\subsubsection*{Demand Distances}
Here, for each good of both instances, we build a \emph{demand
vector} containing the utility values that the good receives from different
agents, sorted in decreasing order.
We then find a mapping of vectors from one instance to the other that minimizes
the sum of $\ell_1$ distances of the mapped pairs. Hence, we obtain the
following formal definition.

\begin{definition}\label{def:demand-distance}
	Let $U^1$ and~$U^2$ be two allocation instances with $n$~agents and
	$m$~goods. The \emph{demand vector}~$\sordem_{U}(j)$ of good~$j \in
	[m]$ of instance~$U$ is the vector $(u_{1,j}, u_{2, j}, \ldots, u_{n, j})$
	sorted in descending order. Denoting by~$\Pi_\mathit{goods}$ all permutations of
	$[m]$, the \emph{demand distance} $\idealdist(U^1, U^2)$ of~$U^1$ and~$U^2$ is
	\begin{equation*}
		\min_{\resperm \in \Pi_\mathit{goods}} \sum_{j
		\in [m]} ||\sordem_{U^1}(j) - \sordem_{U^2}(\resperm(j))||_{1}.
  \end{equation*}
\end{definition}
Due to the fact that this definition optimizes over only a single permutation, 
the demand distance can be computed in polynomial-time by finding a minimum
matching in a weighted bipartite graph representing the contributions to the
distance from matchings between each pair of agents.

\begin{theorem}\label{thm:demand-poly}
	Given two task allocations~$\task{}$ and~$\task'{}$ and a real
	number~$d$, deciding whether it holds that~$\demdist(\task, \task') \leq d$
	is polynomial-time solvable.
\end{theorem}
\begin{proof}
	We give an algorithm that first constructs a weighted bipartite graph
	(representing the task of computing the demand distance) and then computes its
	minimum weight perfect matching, which represents the optimal good
	matching.

	To be specific, our algorithm proceeds as follows. For each $r \in \goods
	\cup \goods'$, it first computes \sordem(r). Then, it constructs a
	bipartite graph~$G$ with one partition consisting of the
	goods~$\goods$ and the other one of the goods~$\goods'$. For
	every pair~$(r, r') \in \goods \times \goods'$, the algorithm adds an
	edge~$\{r, r'\}$ to~$G$ and sets its weight~$w(r, r') \coloneqq
	d_{\ell_1}(\sordem(r), \sordem(r'))$. Finally, the algorithm finds a minimum
	weight perfect matching, say $M$, of~$G$. 

	Since~$M$ is a perfect matching (and $|\goods| = |\goods'|$), it is
	clear that~$M$ represents a good matching~$\resperm$ such that for
	each~$\{r, r'\} \in M$, $\resperm(r) = \resperm(r')$. Hence, the total
	weight~$w(M)$ of~$M$ can be expressed as
	\begin{align*}
		w(M) \coloneqq& \sum_{\{r, r'\} \in M} d_{\ell_1}(\sordem(r),
		\sordem(r')) = \\
									&\sum_{r \in \goods} d_{\ell_1}(\sordem(r),
		\sordem(\resperm(r))). 
	\end{align*}
	As a result, a minimum weight perfect matching in graph~$G$ yields a good
	matching that witnesses the demand distance and the weight of this matching
	is exactly the requested demand distance.

	Computing a minimum weight perfect matching is polynomial-time solvable. Thus,
	our algorithm also runs in polynomial time.
\end{proof}

Note that the procedure described in the proof of \Cref{thm:demand-poly} is
constructive and, in fact, solves the optimization variant of the problem of
computing the demand distance between two allocation tasks.

The improvement in the running time of computing the demand distance in
comparison to that of the valuation distance comes at a cost. The demand
distance ignores information about the identity of agents, which can lead two
non-isomorphic instances to be at
distance $0$ from each other (see~\cref{fig:dist-vs-pseudodist}).

\subsubsection*{Maximal Values of Valuation and Demand Distances}
We show the upper-bound on the values of both distances in the following
technical proposition.

\newcommand{\distsupperbound}{\ensuremath{2n-\frac{2n}{m}}}
\begin{proposition}
\label{prop:distupperbound}
  Let $U$ and~$U'$ be two allocations instances
  with $n$~agents and	$m$~goods.
  Then, the valuation distance $\idealdist(U, U')$
  and the demand distance $\demdist(U, U')$
	are at most $\distsupperbound$. 
\end{proposition}

For the proof of this proposition, we will first require the following lemma:
\begin{lemma}[Rearrangement Inequality]
\label{lem:rearrangement-inequality}
Let~$\vec{x}$ and~$\vec{y}$ be two size~$n$ vectors whose entries are
sorted non-increasingly. Then, for every permutation~$\sigma$ of~$[n]$,
it holds that
$$\sum_{i \in [n]} \min (x_i,y_i) \geq \sum_{i \in [n]} \min (x_i,y_{\sigma(i)}).$$
\end{lemma}
\begin{proof}
 The proof works recursively as follows:
 Assume that the smallest entry of~$\vec{y}$
 would not be in the last position~$n$, but in the position~$i^*$.
 Consider two cases.
 
 First, $x_n\le y_{i^*}$.
 Since be definition $y_n\ge y_{i^*}$,
 swapping~$y_{i^*}$ and~$y_{n}$ cannot decrease the overall sum.
 
 Second, $x_n>y_{i^*}$.
 Then, also $x_{i^*}>y_{i^*}$ (by~$\vec{x}$ being non-increasing).
 That is, swapping~$y_{i^*}$ and~$y_{n}$
 leads the minimum in position~$n$ to become~$y_{i^*}$
 (the previous minimum of position~${i^*}$).
 Yet, the minimum in position~${i^*}$ is now
 at least~$\min (y_n,x_{i^*})$, which must
 be at least $\min (y_{i^*},x_n)$, since~$y_n \ge y_{i^*}$
 and~$x_{i^*} \ge x_n$.

 The last position is correct, that is, remove these entries and recurse.
\end{proof}
\begin{proof}[Proof of Proposition~\ref{prop:distupperbound}]
 We first show the proof for the valuation distance.
 It is an adaption of the proof of Lemma~2 in the full
 version~(arXiv:2205.00492 [cs.GT]) of~\citet{boe-fal-nie-szu-was:c:metrics}.

 Assume towards a contradiction that we have two allocations instances~$U$ and~$U'$
 with $\idealdist(U,U') > \distsupperbound$.
 Assume w.l.o.g.\ that
 $\idealdist(U,U') = \sum_{i \in [n],j \in [m]}
		\left| u_{i,j} - u'_{i,j} \right|$.
 otherwise we could permute rows and column (relabel the goods and agents)
 of one of the allocation instances.

 Observe that
 \begin{align}
 \label{eq:distupperbound:1}
    \idealdist(U,U') 
        &= \sum_{i \in [n],j \in [m]} \left| u_{i,j} - u'_{i,j} \right| \notag\\
        &= \sum_{i \in [n],j \in [m]} \big(\max(u_{i,j},u'_{i,j}) - \min(u_{i,j},u'_{i,j})\big) \notag\\
        &= \sum_{i \in [n],j \in [m]} (u_{i,j} + u'_{i,j}) - 2 \sum_{i \in [n],j \in [m]} \min (u_{i,j},u'_{i,j}) \notag\\
        &= \sum_{i \in [n]} 2 - 2 \sum_{i \in [n],j \in [m]]} \min (u_{i,j},u'_{i,j}) \notag\\
        &= 2n - 2 \sum_{i \in [n],j \in [m]} \min (u_{i,j},u'_{i,j}).
 \end{align}

If $\idealdist(U,U') > \distsupperbound$, then it must hold that:
\begin{equation}
\label{eq:distupperbound:2}
    \sum_{i \in [n],j \in [m]]} \min (u_{i,j},u'_{i,j}) < n/m.
\end{equation}

For each permutation~$\sigma$ of~$[m]$,
we have $\idealdist(U,U') \le \sum_{i \in [n],j \in [m]} \left| u_{i,j} - u'_{i,\sigma(j)} \right|$
(this being incorrect would violate our assumption that
$\idealdist(U,U') = \sum_{i \in [n],j \in [m]}
		\left| u_{i,j} - u'_{i,j} \right|$.
Consequently, for every permutation~$\sigma$ of~$[m]$, and reasoning
analogously to Eq.~\eqref{eq:distupperbound:1} and Eq.~\eqref{eq:distupperbound:2}, we get
\[
    \sum_{i \in [n],j \in [m]} \min (u_{i,j},u'_{i,\sigma(j)}) < n/m.
\]

Observe that if $x, y \in [0,1]$,
then it holds that $x \cdot y \le \min(x,y)$.
Since for each~$i \in [n]$ and~$,j \in [m]$,
we have $u_{i,j},u'_{i,j} \in [0,1]$,
for each each permutation~$\sigma$ of~$[m]$,
it holds that:
\begin{equation}
\label{eq:distupperbound:3}
    \sum_{i \in [n],j \in [m]} u_{i,j} \cdot u'_{i,\sigma(j)} < n/m.
\end{equation}

We define a family~$\Psi:=\{\sigma^{(k)} \mid k \in [m]\}$ of permutations
using functions~$\sigma^{(k)}(j):={(j + k - 1 \mod m)+1}$.
By summing up Eq.~\eqref{eq:distupperbound:3} (on both sides)
for each permutation from~$\Psi$ we obtain:
\[
    \sum_{\sigma \in \Psi} \sum_{i \in [n],j \in [m]} u_{i,j} \cdot u'_{i,\sigma(j)}
    < n/m \cdot |\Psi|,
\]
which can be rearranged to
\begin{equation}
\label{eq:distupperbound:4}
    \sum_{i \in [n],j \in [m]} u_{i,j} \cdot \sum_{\sigma \in \Psi}(u'_{i,\sigma(j)})
    < n/m \cdot |\Psi|.
\end{equation}

Since $\sum_{j \in [m]} u_{i,j} = \sum_{j \in [m]} u'_{i,j} = 1, \forall i \in [n]$, it holds that:
\[
 \sum_{\sigma \in \Psi}(u'_{i,\sigma(j)}) = \sum_{k \in [m]}(u'_{i,\sigma^{(k)}(j)})
  = \sum_{\ell \in [m]} u'_{i,\ell} = 1.
\]
Hence (and with $|\Psi|=m$), from Eq.~\eqref{eq:distupperbound:4} we get
\[
    \sum_{i \in [n],j \in [m]} u_{i,j} < n,
\]
which contradicts the fact that $\sum_{j \in [m]} u_{i,j} = \sum_{j \in [m]} u'_{i,j} = 1, \forall i \in [n]$.
Hence, we have $\idealdist(U,U') \le \distsupperbound$.

To see that also $\demdist(U,U') \le \distsupperbound$, we upper-bound
$\demdist(U,U')\leq\idealdist(U,U')$.

Assume towards a contradiction that the demand distance between
two instances would be greater than the valuation distance.
Let matrices $V=\sordem_{U}(1)\cdots\sordem_{U}(m)$
and $V'=\sordem_{U'}(\sigma(1))\cdots\sordem_{U'}(\sigma(m))$
be the matrices resulting from the column-wise concatenation
of the demand vectors of~$U$ and~$U'$, respectively, using some permutation~$\sigma$ of the columns.

Recall Eq.~\eqref{eq:distupperbound:1}.
Since the demand distance is upper-bounded by the
entry-wise sum of $\ell^1$-distances between~$V$ and~$V'$,
the same reasoning as above holds.
Thus,
 \begin{align}
 \label{eq:distupperbound:5}
    \sum_{i \in [n],j \in [m]} \min (v_{i,j},v'_{i,j}) <
    \sum_{i \in [n],j \in [m]} \min (u_{i,j},u'_{i,j}),
 \end{align}
 Note that, Eq.~\ref{eq:distupperbound:5} holds for any column permutation
 used to define~$V'$, so in particular also when we use the same
 which we used for the valuation distance.
 In other words, we can assume that, up to permutation of the entries,
 $U$~and~$V$ as well as $U'$~and~$V'$ have the same column vectors.

 For Eq.~\ref{eq:distupperbound:5} to hold, it would need to hold that 
 \begin{align}
 \label{eq:distupperbound:6}
    \sum_{i \in [n]} \min (v_{i,j^*},v'_{i,j^*}) <
    \sum_{i \in [n]} \min (u_{i,j^*},u'_{i,j^*})
 \end{align}
 for some column~$j^* \in m$.
 Due to the rearrangement inequality
 (Lemma~\ref{lem:rearrangement-inequality}), however,
 we know that
 $\sum_{i \in [n]} \min (v_{i,j^*},v'_{i,j^*}) \geq
  \sum_{i \in [n]} \min (u_{i,j^*},u'_{i,j^*})$;
 a contradiction to Eq.\ref{eq:distupperbound:6}.
\end{proof}

\begin{figure}
	\begin{minipage}[t]{.4\linewidth}
  	\begin{tabular}{r|cccc}
			\multicolumn{5}{c}{allocation instance~\task}\\\toprule
  		$u_1$ & 2 & 4 & 6 & 8 \\
  		$u_2$ & 3 & 3 & 6 & 8 \\
  		$u_3$ & 6 & 8 & 6 & 0 \\
  		$u_4$ & 8 & 6 & 0 & 6 
  	\end{tabular}

		\bigskip

  	\begin{tabular}{r|cccc}
			\multicolumn{5}{c}{Demand vectors of~\task}\\\toprule
  		\phantom{$a_1$} & 8 & 8 & 6 & 8 \\
  		& 6 & 6 & 6 & 8 \\
  		& 3 & 4 & 6 & 6 \\
  		& 2 & 3 & 0 & 0 
  	\end{tabular}
  \end{minipage}\hfill%
	\begin{minipage}[t]{.4\linewidth}
  	\begin{tabular}{r|cccc}
			\multicolumn{5}{c}{allocation instance~$\task'$}\\\toprule
  		$u'_1$ & 2 & 4 & 6 & 8 \\
  		$u'_2$ & 3 & 3 & 6 & 8 \\
  		$u'_3$ & 6 & 8 & 0 & 6 \\
  		$u'_4$ & 8 & 6 & 6 & 0 
  	\end{tabular}

		\bigskip

  	\begin{tabular}{r|cccc}
			\multicolumn{5}{c}{Demand vectors of~$\task'$}\\\toprule
  		\phantom{$a'_1$} & 8 & 8 & 6 & 8 \\
  		& 6 & 6 & 6 & 8 \\
  		& 3 & 4 & 6 & 6 \\
  		& 2 & 3 & 0 & 0 
  	\end{tabular}
  \end{minipage}
	\caption{Allocation instances demonstrating a zero demand distance (note the
	demand vectors and apply the identity matching) but a non-negative valuation
	distance (verify via an exhaustive check).}	
  \label{fig:dist-vs-pseudodist}
\end{figure}

\section{Deferred Details from Section~\ref{sec:explicit}}
\label{app:sec:explicit}
\begin{figure}\centering%
	\includegraphics[width=0.4\linewidth]{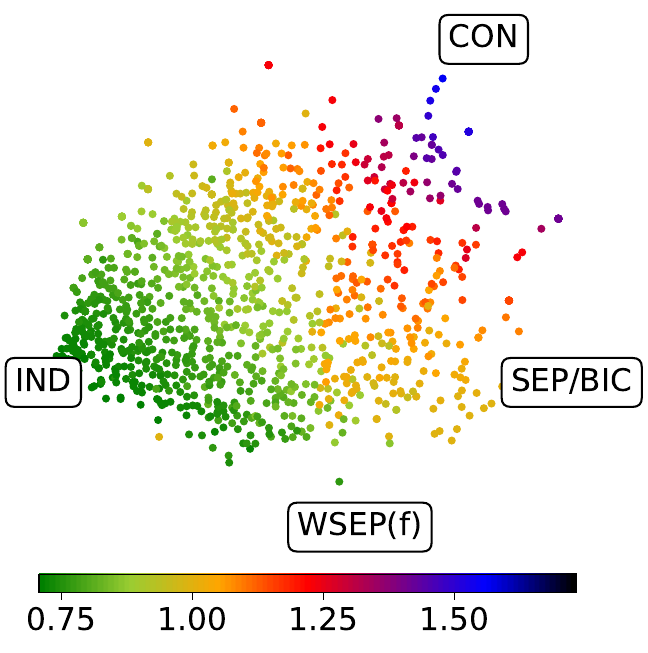}
	\includegraphics[width=0.4\linewidth]{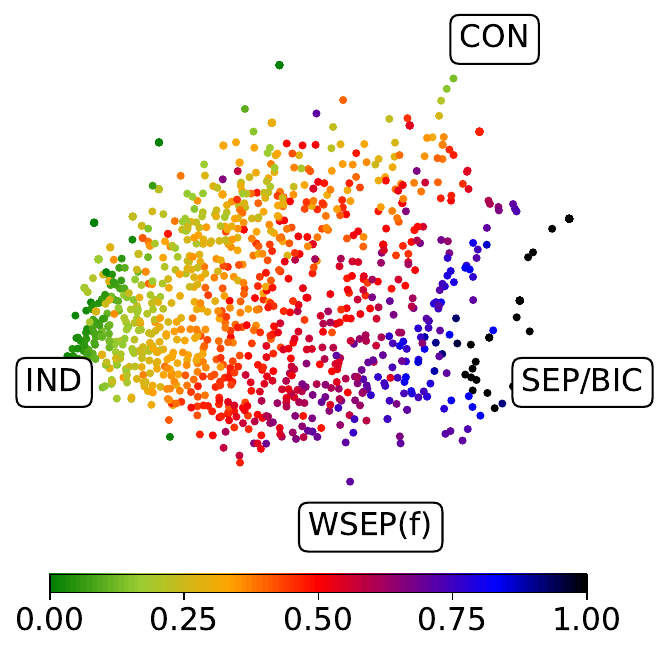}
	\caption{Distributions of the $\sigma_1$ (left) and $\sigma_2$ (right) values on our distance-embedding map using the
	valuation distance for the $3\times 6$ instances.}
	\label{app:fig:explicit:embedding:singvalues}
\end{figure}
An analogous map to the map presented in \cref{fig:explicit:embedding} can be seen in Figure \ref{app:fig:explicit:embedding:singvalues}, in which the distribution of $\sigma_1$ and $\sigma_2$ on our distance-embedding map is shown for the $3\times 6$ instances.
\subsection{Deferred Details from Section~\ref{sec:explicit:theory}}
\label{app:explicit:theory}

\subsubsection{Wide Separability}
\label{app:explicit:theory:wsep}
We propose two ways of generalizing this instance to $m$ not divisible by $n$, setting $\ell \coloneqq \lfloor m/n\rfloor$:
In the first variant, $\WSEP$, each agent values $\ell$ goods at $1/\ell$, and thus $m \bmod n$ goods have no value for any agent.
In the second variant, $\WSEPf$, each agent values the $\ell$ goods that no other agent values at $n/m$, and the final $m \bmod n$ goods have a value of $\frac{1-\ell\,n/m}{m \bmod n}$ for every agent.
$\WSEP$ always lies on the East and $\WSEPf$ on the South border; if $n \mathrel{|} m$, they coincide, meeting in the South-East corner.

\subsection{Deferred Proofs}
\label{app:explicit:theory:proofs}
\thmexplicitw*
\begin{proof}
Singular values are always nonnegative real numbers.
It is well-known that the second singular value is 0 if and only if the matrix has rank 1 (or zero), i.e., if all rows are linearly dependent.
Since all row sums are 1, this is equivalent to all rows being identical.
Since IND and CON each have only identical rows, they have $\sigma_2=0$.
The same property is inherited by their convex combinations, which, by the continuity of the singular values, must trace the entire boundary between IND and CON.
\end{proof}

\thmexplicits*
\begin{proof}
By \cref{eq:sigma1rayleigh},
\begin{align*}
\sigma_1 &= \max_{\vec{v}_1 \in \mathbb{R}^m, \norm{\vec{v}_1}=1} \norm{U \, \vec{v}_1} \\
&\geq \left\lVert U \, \begin{psmallmatrix}1/\sqrt{m}\\1/\sqrt{m}\\\vdots\\1/\sqrt{m}\end{psmallmatrix} \right\rVert \\
&= \left\lVert \begin{psmallmatrix}1/\sqrt{m}\\1/\sqrt{m}\\\vdots\\1/\sqrt{m}\end{psmallmatrix} \right\rVert \tag*{(by row stochasticity)} \\
&= \sqrt{n/m}.
\end{align*}
Note that the vector being multiplied with $U$ in the second row has dimension $m$ (thus norm $1$), but the vector in the third row has dimension $n$.

We now show that, whenever this inequality is tight, all column sums must be $n/m$.
This extends a widely known proof (presentation adapted
from~\citet{mathstackstochastic}) showing that, among square matrices, a
row-stochastic matrix has $\sigma_1 = 1$ if and only if it is doubly stochastic, i.e., its
column sums are also equal to one.
In the following, we set $\mathbf{1}_t$ to denote the vector $\begin{psmallmatrix}
    1\\\vdots\\1
\end{psmallmatrix} \in \mathbb{R}^t$, and denote the vector dot product by $\langle \cdot, \cdot \rangle$.
\begin{align*}
m &= \frac{m}{n} \langle \mathbf{1}_n, \mathbf{1}_n \rangle \\
&= \frac{m}{n} \langle \mathbf{1}_n, U \, \mathbf{e}_m \rangle \tag*{(by row stochasticity)} \\
&= \left\langle U^T \left(\frac{m}{n} \mathbf{1}_n\right), \mathbf{1}_m \right\rangle \\
&\leq \left\lVert U^T \left(\frac{m}{n} \mathbf{1}_n\right) \right\rVert \cdot \norm{\mathbf{1}_m} \tag*{(Cauchy-Schwartz)} \\
&\leq \sigma_1(U^T) \, \frac{m}{n} \, \norm{\mathbf{1}_n} \, \norm{\mathbf{1}_m} \tag*{(property of Operator Norm)} \\
&= \sigma_1(U) \, \frac{m}{n} \, \norm{\mathbf{1}_n} \, \norm{\mathbf{1}_m} \tag*{\text{($\sigma_1(A) = \sigma_1(A^T)$)}}\\
&= \sqrt{\frac{n}{m}} \, \frac{m}{n} \, \norm{\mathbf{1}_n} \, \norm{\mathbf{1}_m} \tag*{(by assumption)}\\
&= \sqrt{\frac{n}{m}} \, \frac{m}{n} \, \sqrt{n} \, \sqrt{m} \\
&= m.
\end{align*}
Since both ends of the inequality chain are equal, all terms along the chain must be equal. Since the Cauchy-Schwartz step was an equality, we know that $U^T \, (m/n \, \mathbf{1}_n)$ is a scalar multiple of $\mathbf{1}_m$; since $\norm{U^T \, (m/n \, \mathbf{1}_n)} \, \norm{1_m} = m$, we now that $\norm{U^T \, (m/n \, \mathbf{1}_n)} = \sqrt{m}$; finally, we now that $U^T \, (m/n \, \mathbf{1}_n)$ is nonnegative.
Taking these facts together, we conclude that $U^T \, (m/n \, \mathbf{1}_n) = \mathbf{1}_n$, which means that all column sums of $U$ are equal to $n/m$.

It is easy to see that $\IND$, $\WSEPf$, and their linear interpolations all have column sums of $n/m$, which concludes the claim.
\end{proof}

\thmexplicitn*
\begin{proof}
Setting $\sigma_1, \sigma_2, \dots, \sigma_r$ for the singular values of some matrix, it is well known that
$\sum_{t=1}^r \sigma_t^2$ equals the square of the matrix' Frobenius norm, i.e., equals the sum of squares across the entries of the matrix. That is, for a matrix $U$,
\[ \sum_{t=1}^r \sigma_t^2 = \sum_{i=1}^{n} \sum_{j=1}^m u_{i,j}^2. \]
Since all entries of our matrix are between 0 and 1, the $i$th row's
contribution to the right-hand side is $\sum_{j=1}^m u_{i,j}^2 \leq \sum_{j=1}^m
u_{i,j} = 1$, and this inequality is tight if and only if agent $i$ values one item at 1
and all others at 0.
It follows that
\[ \sigma_1^2 + \sigma_2^2 \leq \sum_{t=1}^r \sigma_t^2 \leq n, \]
where the inequality is tight exactly if and only if (a) all agents single-mindedly value a single good (which makes the second inequality tight) and (b) all $\sigma_t$ for $t \geq 3$ are zero.
Part (b) is the case if and only if $U$ has rank at most 2. Assuming part (a), $U$'s rank is exactly the number of distinct goods which some agent values single-mindedly.
Taking square, we obtain the desired inequality with its necessary-and-sufficient condition.

Clearly, this condition is satisfied by $\IND$ and by $\BIC$ if $n$ is even.
Note that there is a natural way to interpolate between these two, where one good is single-mindedly valued by $t$ agents and a second good by $n-t$ agents.
Clearly, each of these points lies on the boundary. If one linearly interpolates between successive values of $t$, the interpolations have one agent who values two items, which removes this interpolation point from the boundary.
But of course this operation does not move the instance far from the boundary.
\Cref{fig:explicit:outlines} shows this interpolation as the blue line following the upper boundary of the map.
\end{proof}

\thmexplicite*
\begin{proof}
$\sigma_2 \leq \sigma_1$ holds by definition of the singular values.

Let $U$ be a utility matrix with the block matrix structure from the theorem
statement (since the singular values are invariant to row and column
permutations, it suffices to consider such matrices directly).
One important property of singular values we have not used yet is that the singular values of a matrix $U$ are the square roots of the eigenvalues of the Gram matrix $U^T U$ (or, equivalently, for $U U^T$).

Given the block matrix structure,
\[ \begin{pmatrix}
A & 0 & 0\\
0 & A & 0\\
0 & 0 & B
\end{pmatrix} \, \begin{pmatrix}
A & 0 & 0\\
0 & A & 0\\
0 & 0 & B
\end{pmatrix}^T = \begin{pmatrix}
A A^T & 0 & 0\\
0 & A A^T & 0\\
0 & 0 & B B^T
\end{pmatrix}.\]
It is well known that the eigenvalues of such a block diagonal matrix are simply the eigenvalues of the blocks $A A^T$, $A A^T$, and $B B^T$ combined (with multiplicity preserved), which means that the singular values of $U$ are just the singular values of $A$, $A$, and $B$ combined. In particular, all singular values of $A$ will appear in $U$ at least twice. By $\sigma_1(A) \geq \sigma_1(B)$, the largest singular value is one of these duplicated singular values, which implies that $\sigma_1(U) = \sigma_2(U) = \sigma_1(A)$.

After reshuffling the columns, $\WSEP$ looks as follows (setting again $\ell \coloneqq \lfloor m / n \rfloor$):\\
\resizebox{\linewidth}{!}{$\left(\begin{array}{ccc|ccc|ccccccc}
1/\ell & \cdots & 1/\ell & 0 & 0 & 0 & 0 & 0 & 0 & 0& 0 & 0 & 0\\ \hline
0 & 0 & 0 & 1/\ell & \cdots & 1/\ell & 0 & 0 & 0 & 0 &0 &0 & 0 \\ \hline
0 & 0 & 0 & 0 & 0 & 0 & 1/\ell & \cdots & 1/\ell & 0 & 0 & \cdots & 0  \\ 
 &  &  &  &  &  &  & &  &  \ddots &   \\ 
0 & 0 & 0 & 0 & 0 & 0 & 0 & 0 & 0 & 0 & 1/\ell & \cdots & 1/\ell \\ 
\end{array}\right)$}
where the lines indicate the division of the blocks.
The symmetry of the instance ensures that (unless $B$ has height 0) $\sigma_1(B) = \sigma_1(A)$, so $\WSEP$ lies on this boundary.

For $\BIC$, the situation is even simpler:
\[\left(\begin{array}{c|c|cccc}
1&0&0&0&\cdots&0 \\
\vdots&0&0&0&\cdots&0 \\
1&0&0&0&\cdots&0 \\ \hline
0&1&0&0&\cdots&0 \\
0&\vdots&0&0&\cdots&0 \\
0&1&0&0&\cdots&0 \\ \hline
0&0&1& 0 & \cdots & 0
\end{array}\right)\]
In this case, the singular values are easy to calculate: $\sigma_1(A) = \sqrt{\lfloor n/2 \rfloor}$ which is at least $\sigma_1(B) = 1$.

Interpolating between both matrices without leaving the boundary is not straightforward.
For this, we first linearly interpolate from $\WSEP$ to $\SEP$.
For some $0 < \theta < 1$, this means that each agent approves one good at $\theta + (1-\theta)/\ell$ and $\ell-1$ goods at $(1-\theta)/\ell$.
For this interpolation, the block structure remains the same as the one discussed for $\WSEP$ and preserves the same symmetry, which is why the interpolation stays on the boundary.

The interpolation from $\SEP$ to $\BIC$ proceeds in discrete steps as follows: For $1 \leq r \leq \lfloor n / 2 \rfloor$, $r$ many agents only value the first good, $r$ agents value only the second good, and the remaining agents value each a separate good.
One verifies that this recovers $\SEP$ for $r=1$ and $\BIC$ for $r= \lfloor n/2 \rfloor$, and that each of these stages can be represented in the block matrix shape, where $A$ is a column of $r$ ones, and $B$ is the identity matrix with possibly zero columns attached to its right.
Then, $\sigma_1(A) = \sqrt{r}$ which is at least $\sigma_1(B) = 1$, which means that this interpolation step lies on the boundary. By linearly interpolating between successive steps, one obtains (after reordering) matrices of the shape
\[\left(\begin{array}{cc|cc|cccc}
1&0&0&0&0&0&\cdots&0 \\
\vdots&0&0&0&0&0&\cdots&0 \\
1&0&0&0&0&0&\cdots&0 \\ 
\theta & 1-\theta&0&0&0&0&\cdots&0 \\ \hline
0&0&1&0&0&0&\cdots&0 \\
0&0&\vdots&0&0&0&\cdots&0 \\
0&0&1&0&0&0&\cdots&0 \\ 
0&0&\theta&1-\theta&0&0&\cdots&0 \\ \hline
0&0&0&0&1& 0 & \cdots & 0\\
0&0&0&0&0&1 & \cdots & 0 \\
0&0&0&0&0& 0 & \ddots & 0
\end{array}\right),\]
which one verifies also lie on the boundary. As a result, we can continuously\footnote{While we reordered the matrix in between, this is just for exposition.} interpolate from $\WSEP$ to $\BIC$ while staying on the right boundary. It follows that this interpolation traces the entire right boundary between $\WSEP$ and $\BIC$.
\end{proof}

\section{Additional Experimental Results}
\label{app:moreexperiments}
In this section, we expand upon \cref{subsec:stude,sec:explicit:empirical} by
(i) presenting the remaining comparisons of the different maps for the different features, including the demand distance defined in \cref{sec:demand-distance},
(ii) introducing new features, and
(iii) disaggregating plots of the distributions of instance sources.

A comparison of the distance-embedding maps and the explicit map regarding the
minimax envy, the maximum Nash welfare, and the maximum utilitarian welfare can be seen in
\cref{fig:all:mmenvy,fig:all:nash,fig:all:welf}, respectively;
these show that the distance-embedding and explicit maps provide similar
information regarding these features.

The same observation holds for some other features:
\cref{fig:all:exef} shows whether an instance permits an envy-free
allocation---an information that can also be derived from the minimax envy---,
while \cref{fig:all:exefp} shows whether an instance allows an envy-free and Pareto-efficient allocation.
Interestingly, the maps of the two features look identical---indeed they are for $5\times 5$ and $10\times 20$, and only differ for 15 instances for $3\times 6$.
We have also investigated whether an instance fulfills the maximin share (MMS) criterion, which requires each agent to receive a bundle with a utility no less than the maximum, over all allocations, of the utility of the bundle with the lowest utility for the agent.
We omit the corresponding maps, as each instance of our two instance set satisfies this criterion.
Additionally, we investigate the fraction of the proportional share that
can be guaranteed, i.e. the largest $\alpha$ so that $u_i(S_i) \geq \alpha\cdot \frac{u_i([m])}{n}$ for each $i \in [n]$, where $S_i$ denotes the bundle of goods given to agent~$i$.
Across our maps shown in Figure \cref{fig:all:propshare}, this feature follows a similar pattern
as the maximum Nash welfare, which is intuitive since it is (up to scaling) the maximum
achievable egalitarian welfare.
Furthermore, we consider the sum, over all agents, of the maximal envies, i.e. $\sum_{i \in [n]}\max_{i' \neq i} u_i(S_{i'}) - u_i(S_i)$, which can be seen in \cref{fig:all:sumabs} and which shows a similar color gradient to utilitarian welfare, but reversed: the sum of the maximal envies (smoothly) decreases if an instance is closer to separability.

While the previous features are based on allocations, we also consider features that can be computed solely from the utility matrix.
\cref{fig:all:maxd} and \cref{fig:all:prefdiv} show the ``maximum demand'' and ``preference diversity'', respectively, which are introduced in Section~\ref{sec:explicit:demystifying}: These results support the correlation between the features and the singular values mentioned in the latter section.
\cref{fig:all:fracsingleminded} shows the fraction of agents who are single-minded, i.e., who value only one item: More than half of the map is covered by instances in which at most $20\%$ of the agents are single-minded.

\begin{figure*}\centering
	\def\firsttwo{.29}
	\begin{subfigure}[t]{\firsttwo\linewidth}
		\centering
		\includegraphics[width=\linewidth]{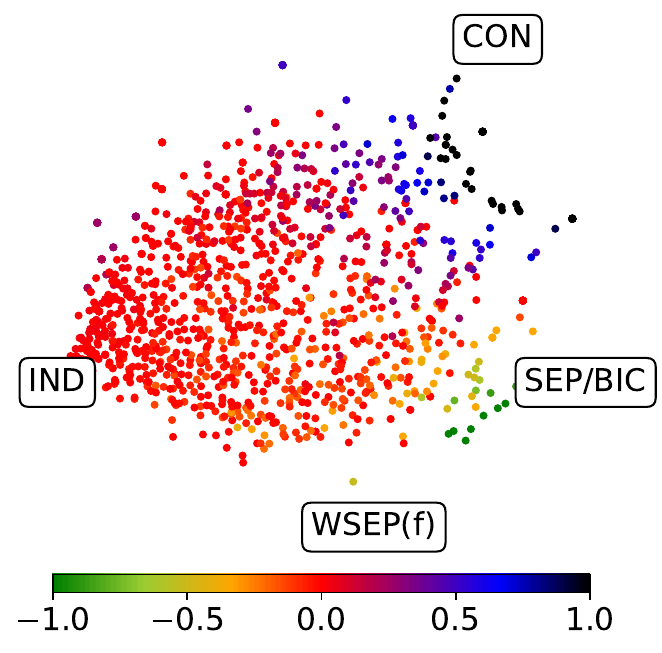}
		\caption{$3\times 6$, valuation distance}
	\end{subfigure}\hfill%
	\begin{subfigure}[t]{\firsttwo\linewidth}
		\centering
		\includegraphics[width=\linewidth]{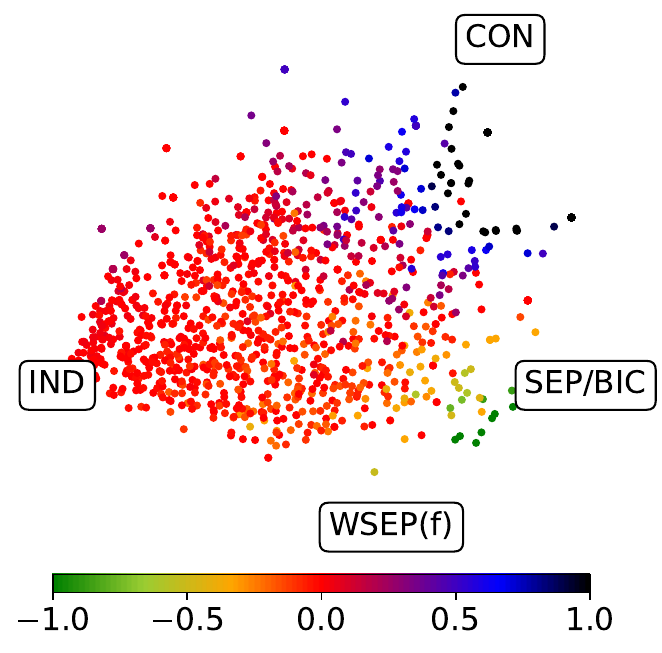}
		\caption{$3\times 6$, demand distance}
	\end{subfigure}\hfill%
	\begin{subfigure}[t]{\fpeval{(0.96 - \firsttwo - \firsttwo)*\linewidth}pt}
		\centering
		\includegraphics[width=\linewidth]{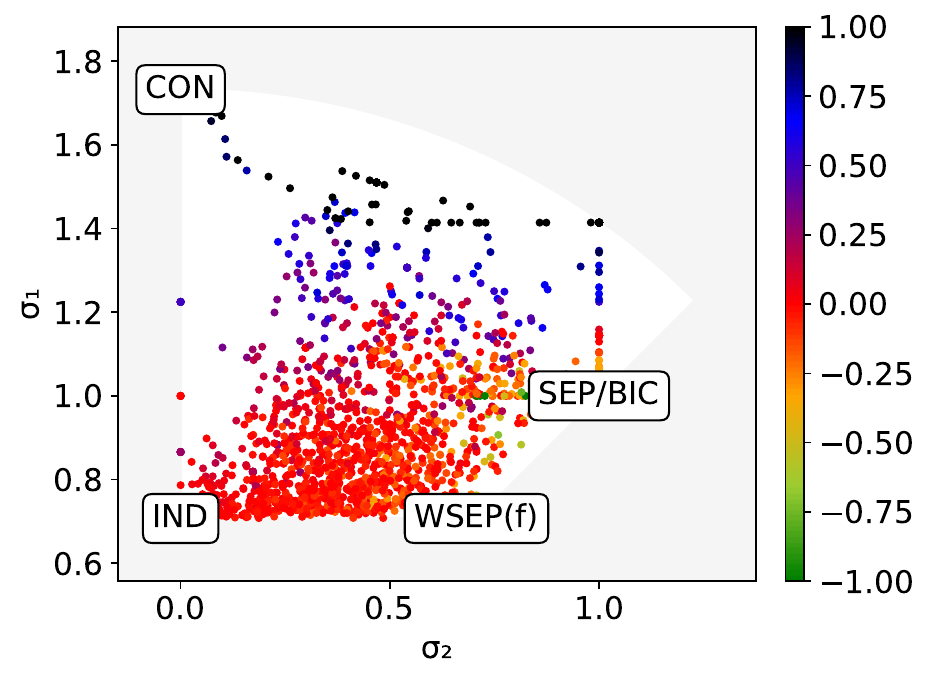}
		\caption{$3\times 6$}
	\end{subfigure}\\
	\begin{subfigure}[t]{\firsttwo\linewidth}
		\centering
		\includegraphics[width=\linewidth]{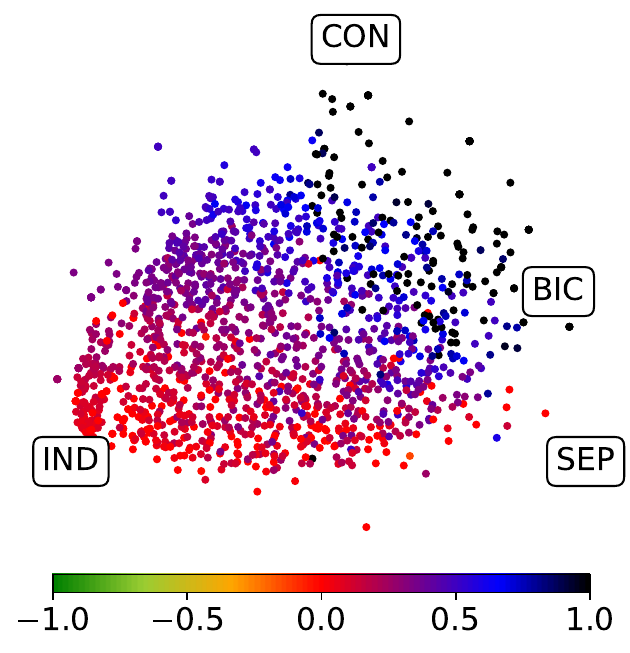}
		\caption{$5\times 5$, valuation distance}
	\end{subfigure}\hfill%
	\begin{subfigure}[t]{\firsttwo\linewidth}
		\centering
		\includegraphics[width=\linewidth]{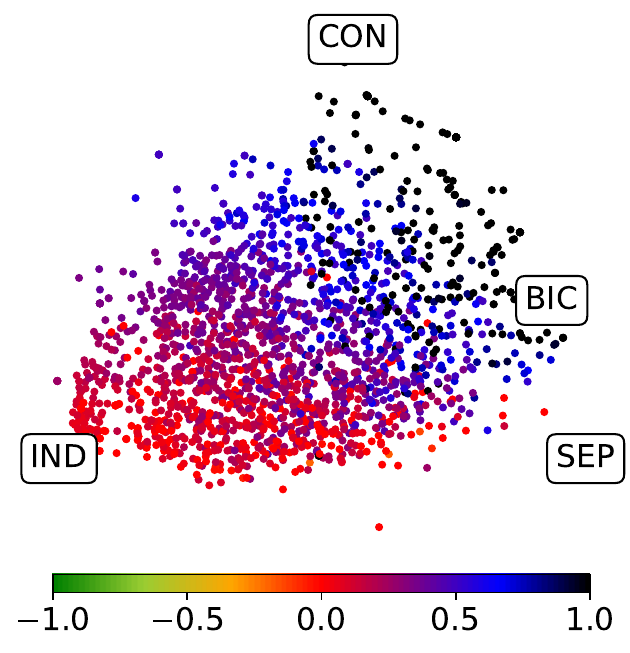}
		\caption{$5\times 5$, demand distance}
	\end{subfigure}\hfill%
	\begin{subfigure}[t]{\fpeval{(0.96 - \firsttwo - \firsttwo)*\linewidth}pt}
		\centering
		\includegraphics[width=\linewidth]{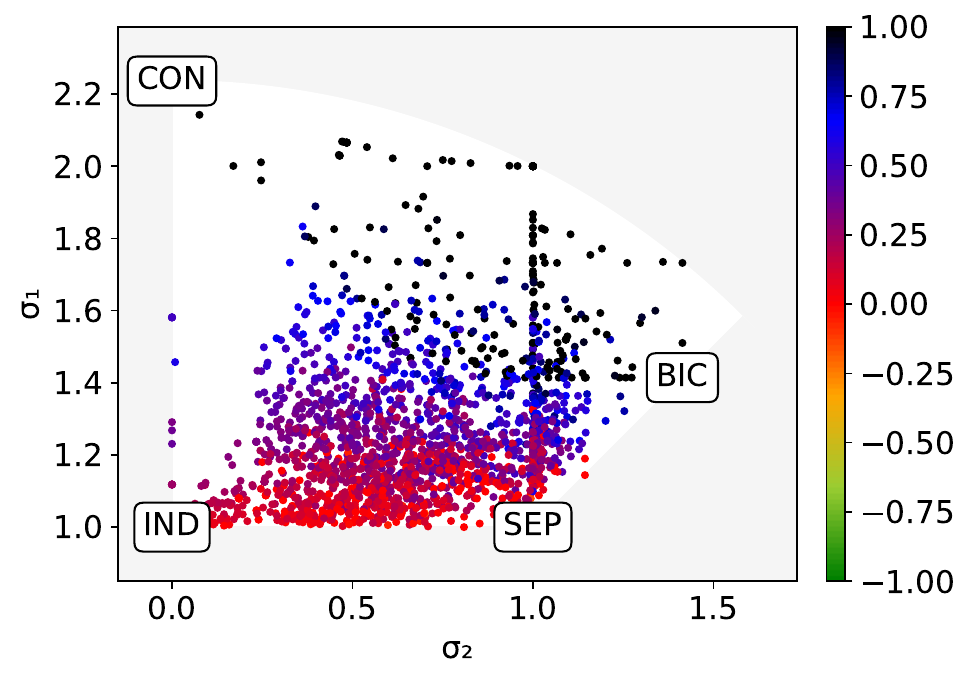}
		\caption{$5\times 5$}
	\end{subfigure}\\
	\begin{subfigure}[t]{\firsttwo\linewidth}
		\quad
	\end{subfigure}\hfill%
	\begin{subfigure}[t]{\firsttwo\linewidth}
		\centering
		\includegraphics[width=\linewidth]{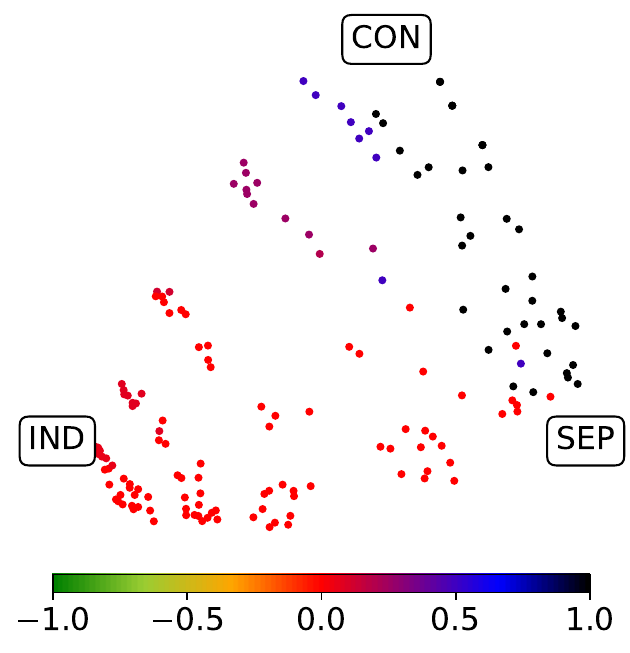}
		\caption{$10\times 20$, demand distance}
	\end{subfigure}\hfill%
	\begin{subfigure}[t]{\fpeval{(0.96 - \firsttwo - \firsttwo)*\linewidth}pt}
		\centering
		\includegraphics[width=\linewidth]{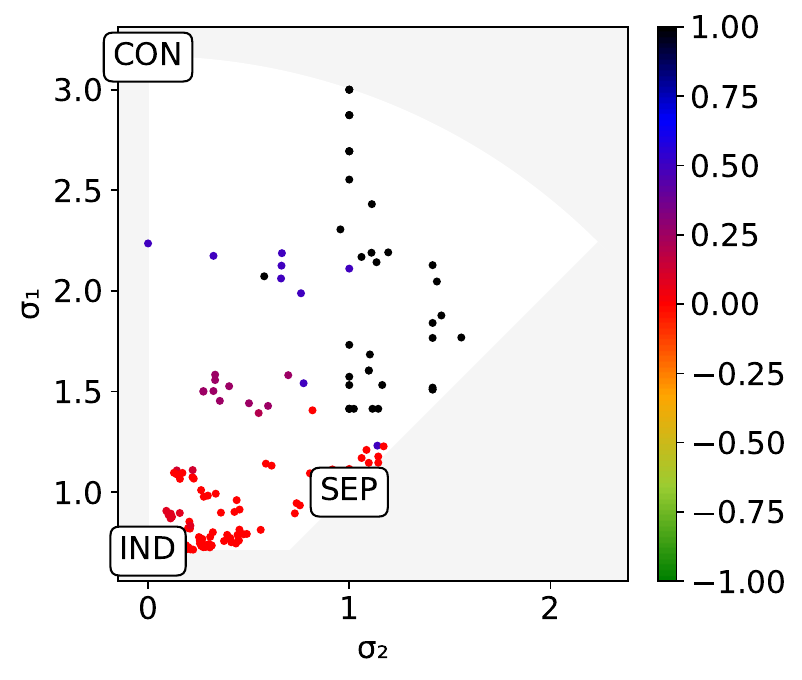}
		\caption{$10\times 20$}
	\end{subfigure}\\
	\caption{Distribution of the minimax envy on our distance-embedding map using
	the valuation distance (first column), using the demand distance (second
  column), and on our explicit map (third column). Computing the valuation
  distance for instances of the~$10 \times 20$~dataset was computationally too
  demanding, hence the blank space in the third row.}
  \label{fig:all:mmenvy}
	\let\firsttwo\undefined
\end{figure*}

\begin{figure*}\centering
	\def\firsttwo{.29}
	\begin{subfigure}[t]{\firsttwo\linewidth}
		\centering
		\includegraphics[width=\linewidth]{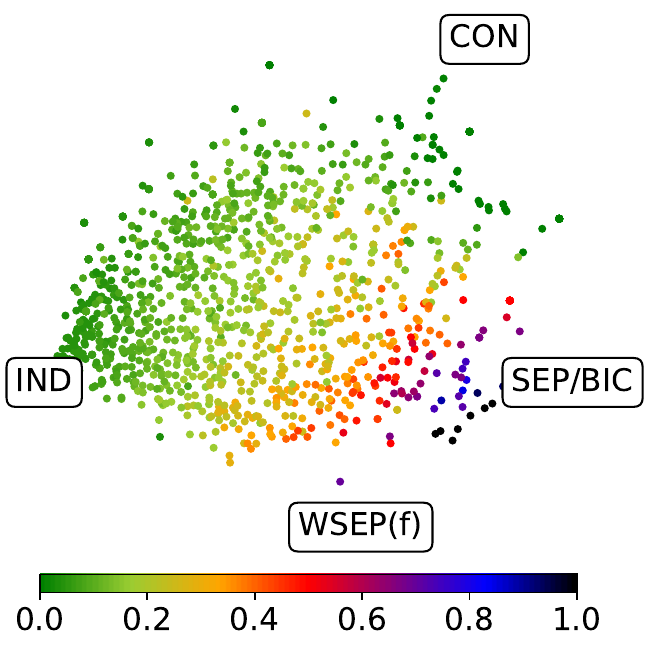}
		\caption{$3\times 6$, valuation distance}
	\end{subfigure}\hfill%
	\begin{subfigure}[t]{\firsttwo\linewidth}
		\centering
		\includegraphics[width=\linewidth]{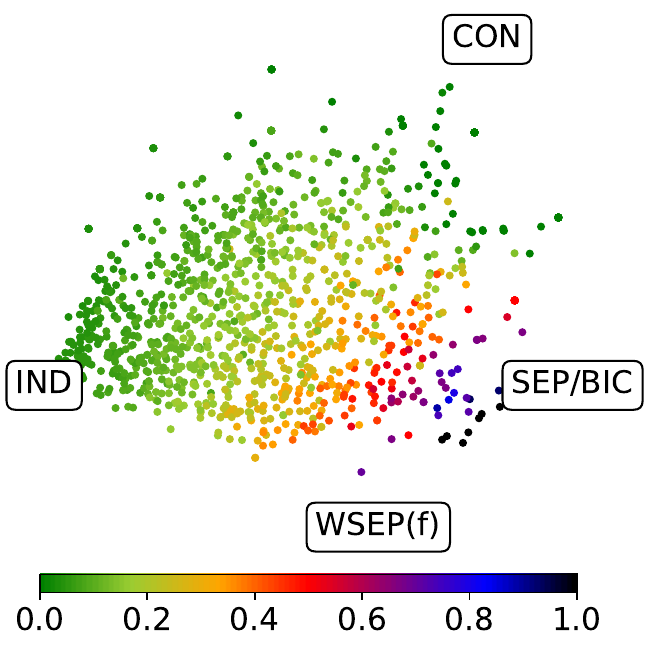}
		\caption{$3\times 6$, demand distance}
	\end{subfigure}\hfill%
	\begin{subfigure}[t]{\fpeval{(0.96 - \firsttwo - \firsttwo)*\linewidth}pt}
		\centering
		\includegraphics[width=\linewidth]{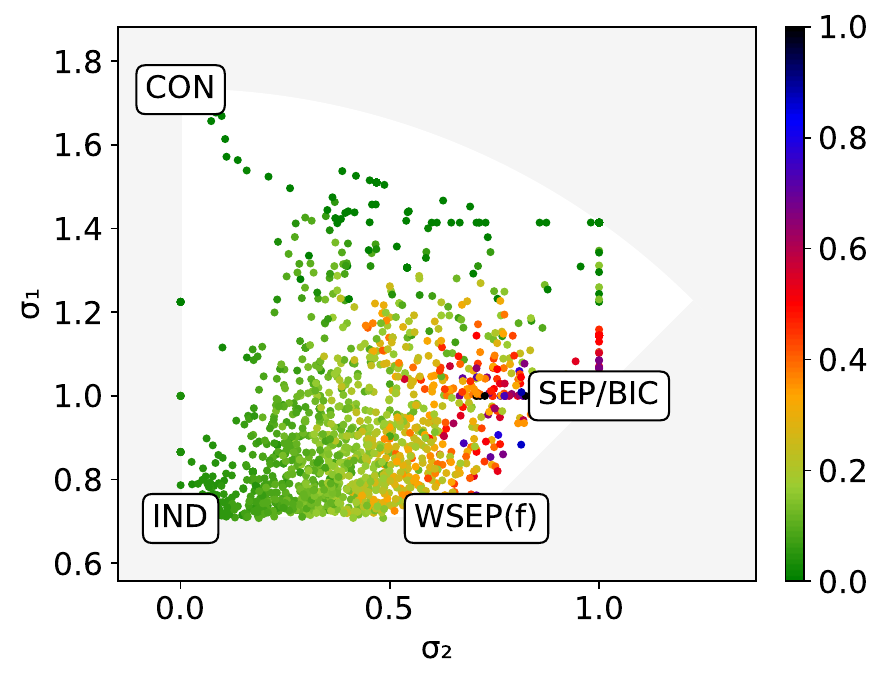}
		\caption{$3\times 6$}
	\end{subfigure}\\
	\begin{subfigure}[t]{\firsttwo\linewidth}
		\centering
		\includegraphics[width=\linewidth]{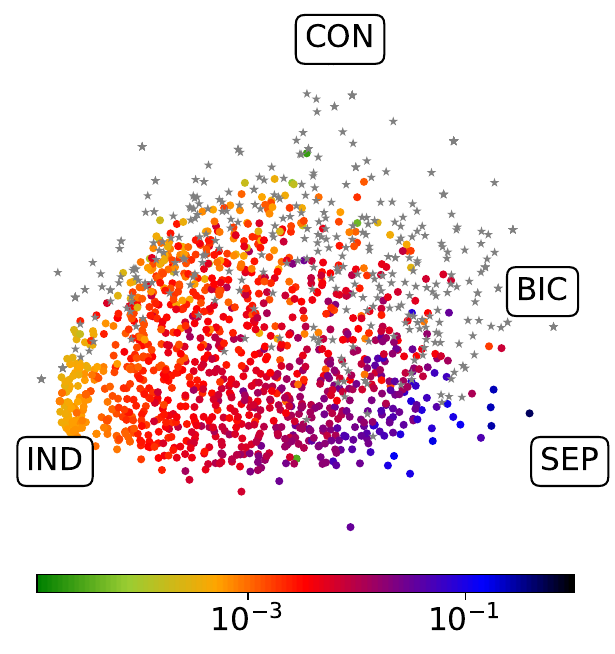}
		\caption{$5\times 5$, valuation distance}
	\end{subfigure}\hfill%
	\begin{subfigure}[t]{\firsttwo\linewidth}
		\centering
		\includegraphics[width=\linewidth]{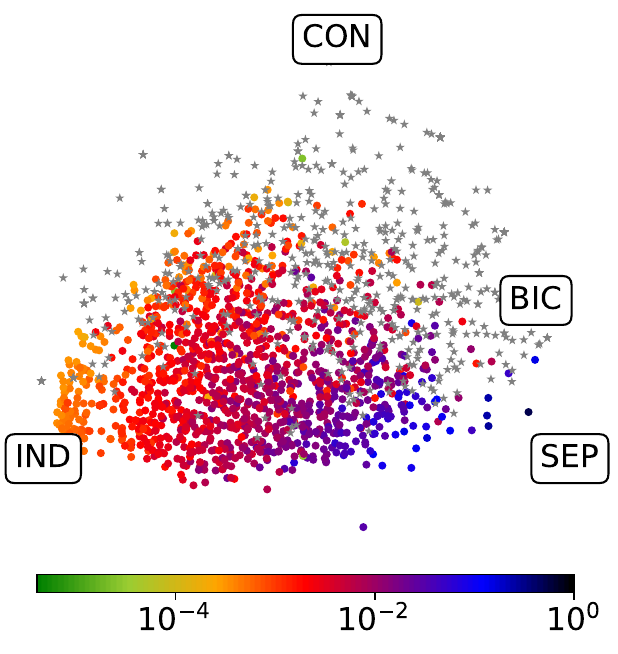}
		\caption{$5\times 5$, demand distance}
	\end{subfigure}\hfill%
	\begin{subfigure}[t]{\fpeval{(0.96 - \firsttwo - \firsttwo)*\linewidth}pt}
		\centering
		\includegraphics[width=\linewidth]{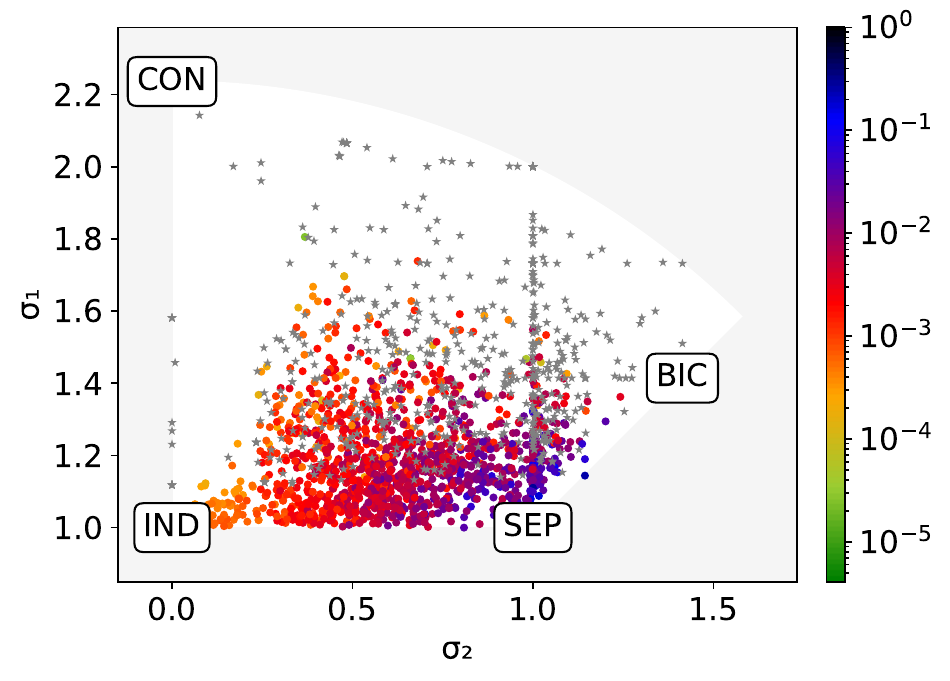}
		\caption{$5\times 5$}
	\end{subfigure}\\
	\begin{subfigure}[t]{\firsttwo\linewidth}
		\quad
	\end{subfigure}\hfill%
	\begin{subfigure}[t]{\firsttwo\linewidth}
		\centering
		\includegraphics[width=\linewidth]{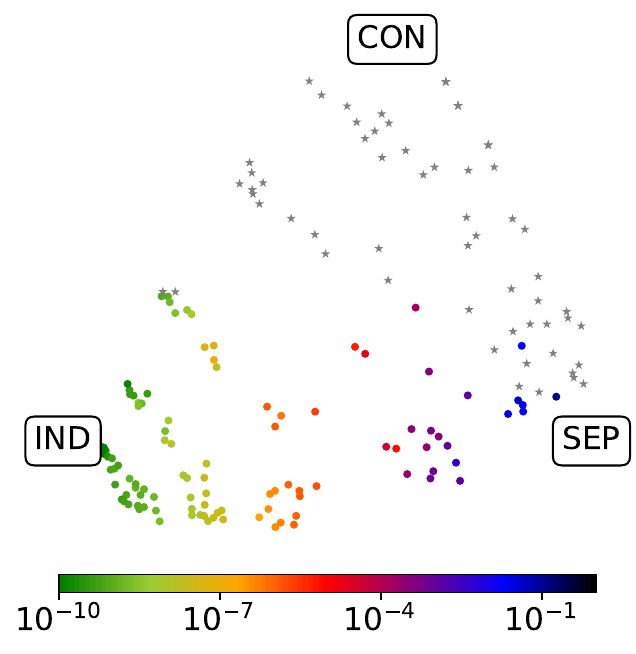}
		\caption{$10\times 20$, demand distance}
	\end{subfigure}\hfill%
	\begin{subfigure}[t]{\fpeval{(0.96 - \firsttwo - \firsttwo)*\linewidth}pt}
		\centering
		\includegraphics[width=\linewidth]{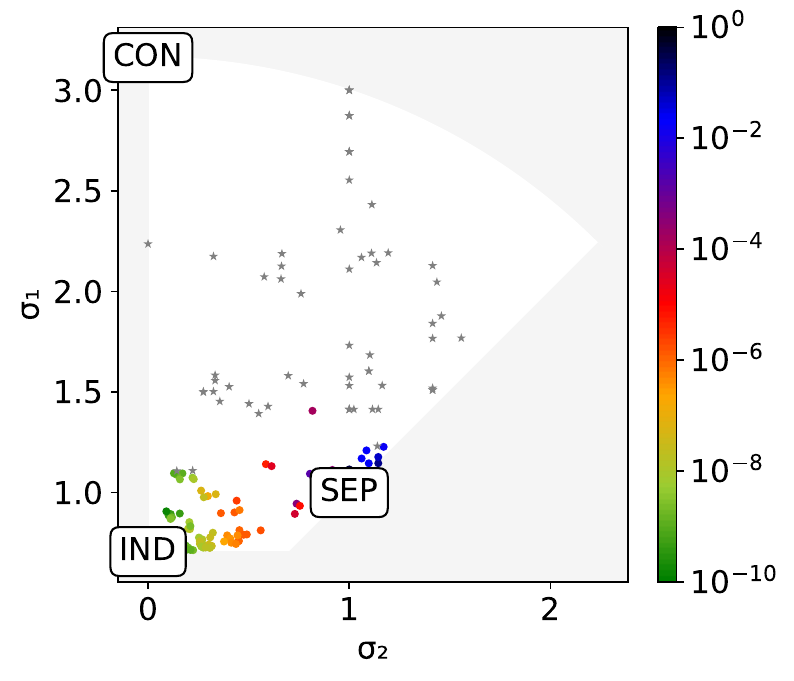}
		\caption{$10\times 20$}
	\end{subfigure}\\
	\caption{Distribution of the maximal Nash welfare on our distance-embedding
	map using the valuation distance (first column), using the demand distance
(second column), and on our explicit map (third column). Grey stars represent
the value $0$ in the $5\times 5$ and $10\times 20$ maps. For the $10\times 20$
instances, we multiply the utilities by $1000$ and round them naturally while
computing the allocation to avoid excessive runtimes. Computing the valuation
distance for instances of the~$10 \times 20$~dataset was computationally too
demanding, hence the blank space in the third row.}
  \label{fig:all:nash} 
	\let\firsttwo\undefined
\end{figure*}

\begin{figure*}\centering
	\def\firsttwo{.29}
	\begin{subfigure}[t]{\firsttwo\linewidth}
		\centering
		\includegraphics[width=\linewidth]{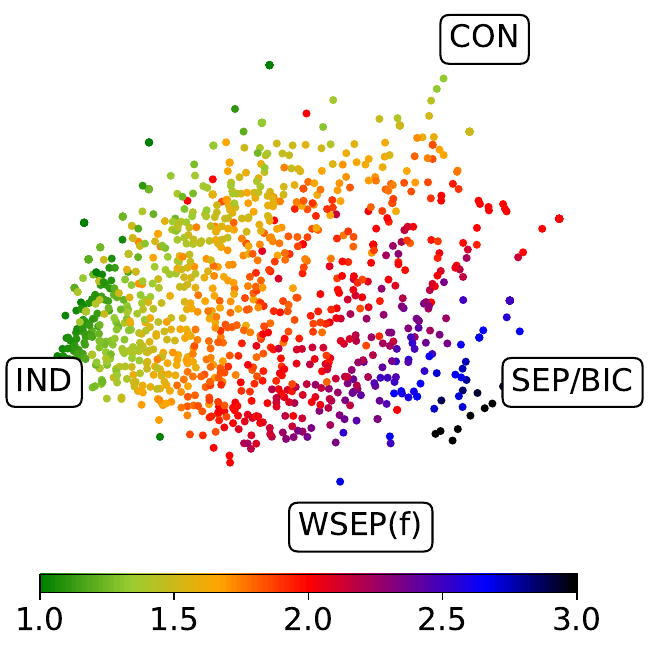}
		\caption{$3\times 6$, valuation distance}
	\end{subfigure}\hfill%
	\begin{subfigure}[t]{\firsttwo\linewidth}
		\centering
		\includegraphics[width=\linewidth]{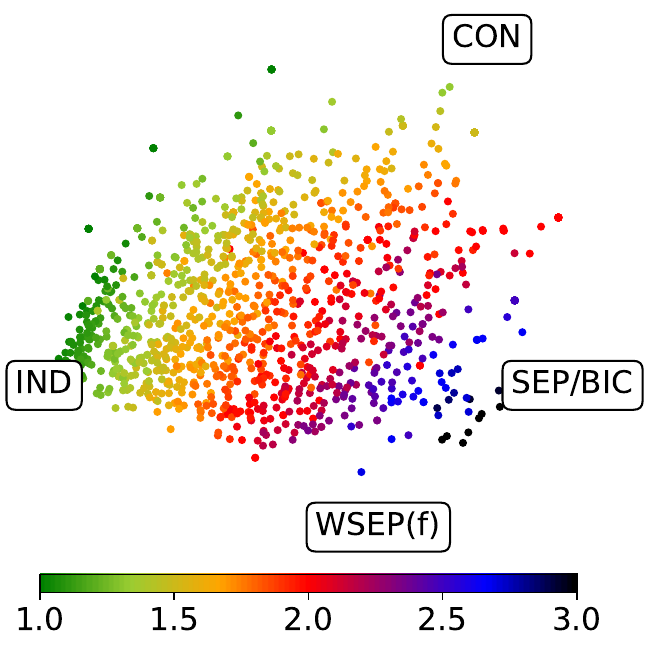}
		\caption{$3\times 6$, demand distance}
	\end{subfigure}\hfill%
	\begin{subfigure}[t]{\fpeval{(0.96 - \firsttwo - \firsttwo)*\linewidth}pt}
		\centering
		\includegraphics[width=\linewidth]{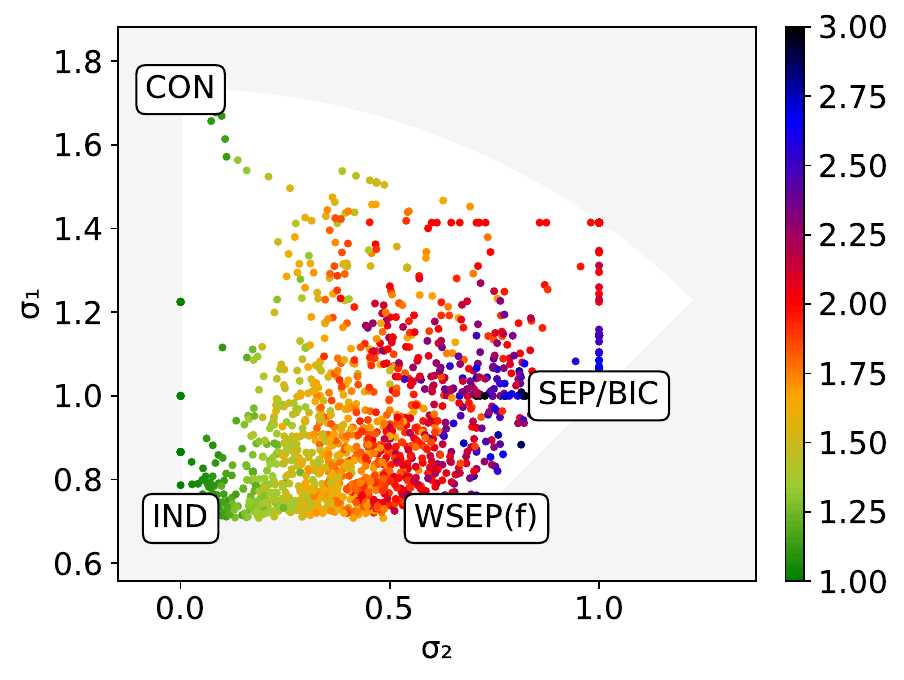}
		\caption{$3\times 6$}
	\end{subfigure}\\
	\begin{subfigure}[t]{\firsttwo\linewidth}
		\centering
		\includegraphics[width=\linewidth]{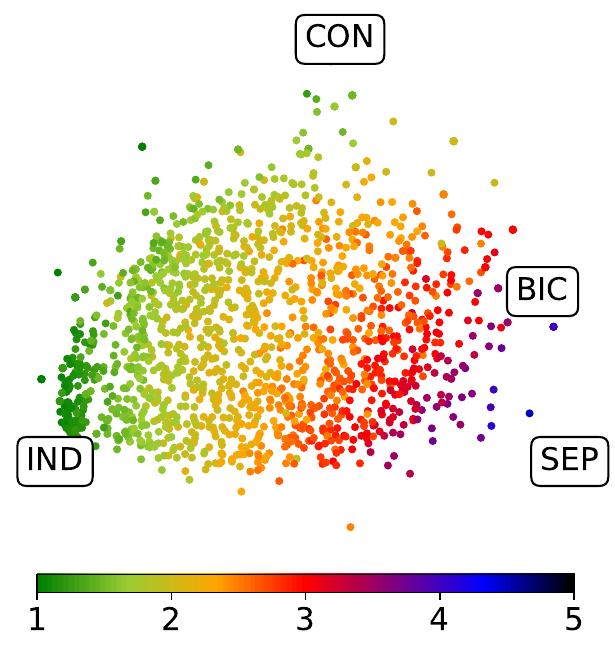}
		\caption{$5\times 5$, valuation distance}
	\end{subfigure}\hfill%
	\begin{subfigure}[t]{\firsttwo\linewidth}
		\centering
		\includegraphics[width=\linewidth]{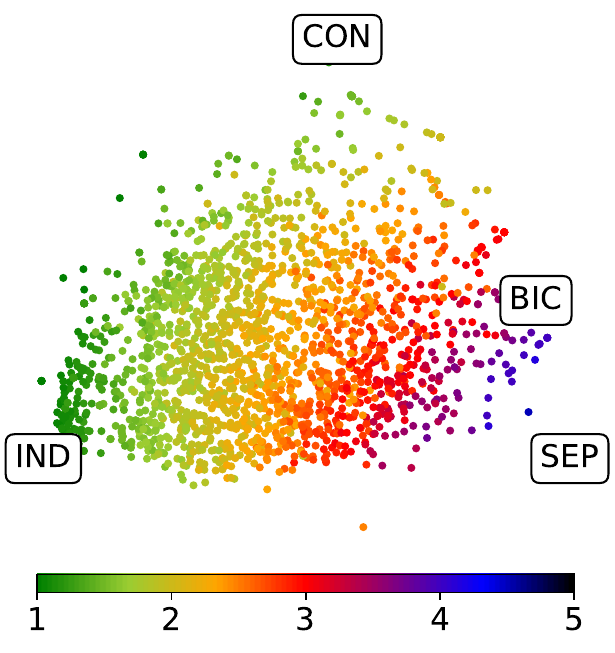}
		\caption{$5\times 5$, demand distance}
	\end{subfigure}\hfill%
	\begin{subfigure}[t]{\fpeval{(0.96 - \firsttwo - \firsttwo)*\linewidth}pt}
		\centering
		\includegraphics[width=\linewidth]{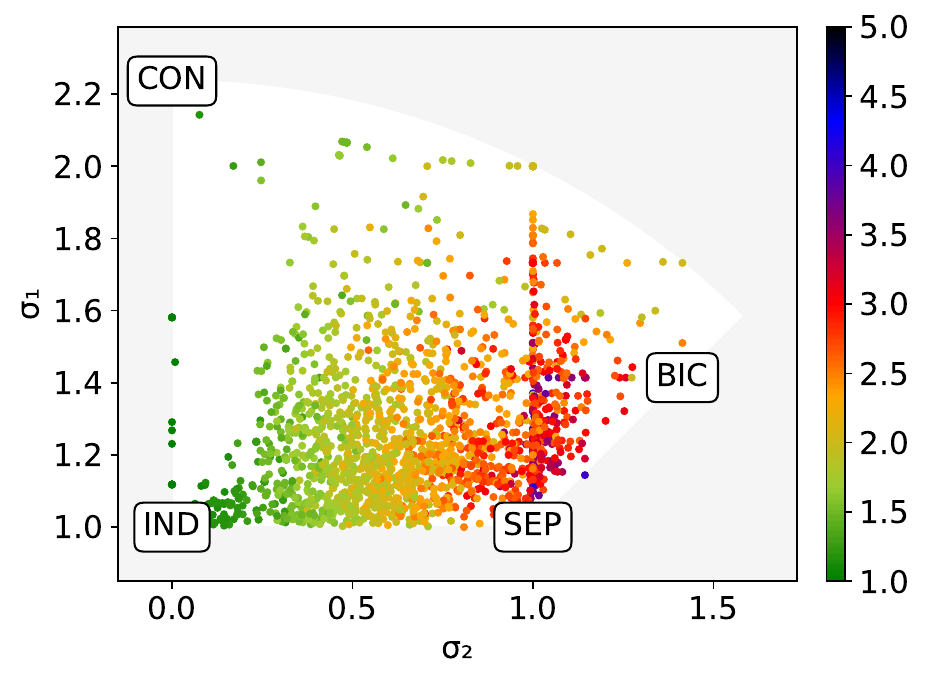}
		\caption{$5\times 5$}
	\end{subfigure}\\
	\begin{subfigure}[t]{\firsttwo\linewidth}
		\quad
	\end{subfigure}\hfill%
	\begin{subfigure}[t]{\firsttwo\linewidth}
		\centering
		\includegraphics[width=\linewidth]{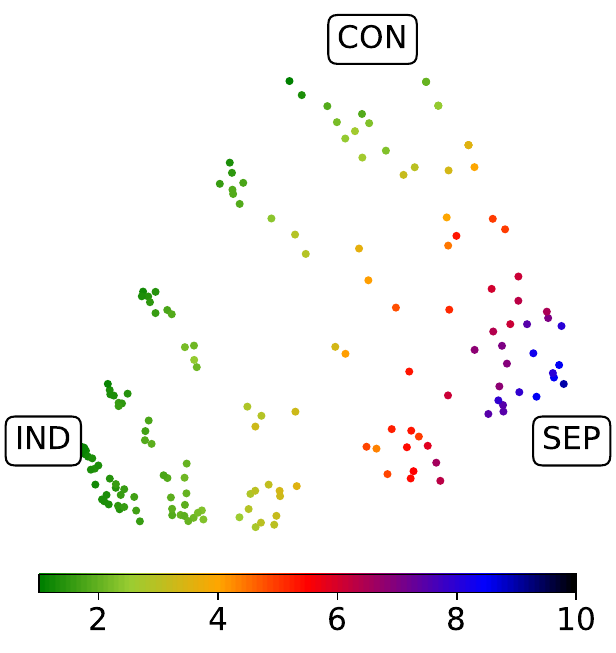}
		\caption{$10\times 20$, demand distance}
	\end{subfigure}\hfill%
	\begin{subfigure}[t]{.32\linewidth}
		\centering
		\includegraphics[width=\linewidth]{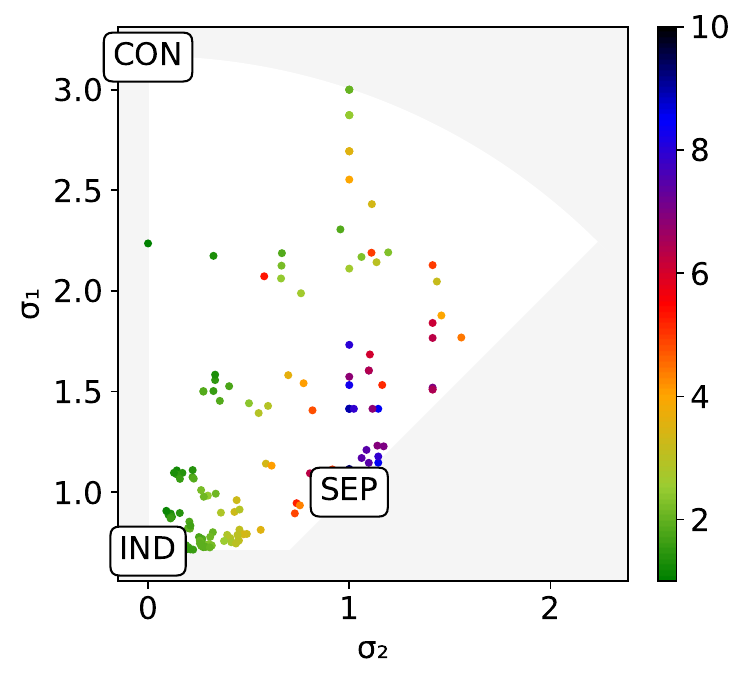}
		\caption{$10\times 20$}
	\end{subfigure}\\
	\caption{Distribution of the maximum utilitarian welfare on our
	distance-embedding map using the valuation distance (first column), using the
  demand distance (second column), and on our explicit map (third column).
  Computing the valuation distance for instances of the~$10 \times 20$~dataset was
  computationally too demanding, hence the blank space in the third row.}
  \label{fig:all:welf} 
	\let\firsttwo\undefined
\end{figure*}

\begin{figure*}\centering
	\def\firsttwo{.29}
	\begin{subfigure}[t]{\firsttwo\linewidth}
		\centering
		\includegraphics[width=\linewidth]{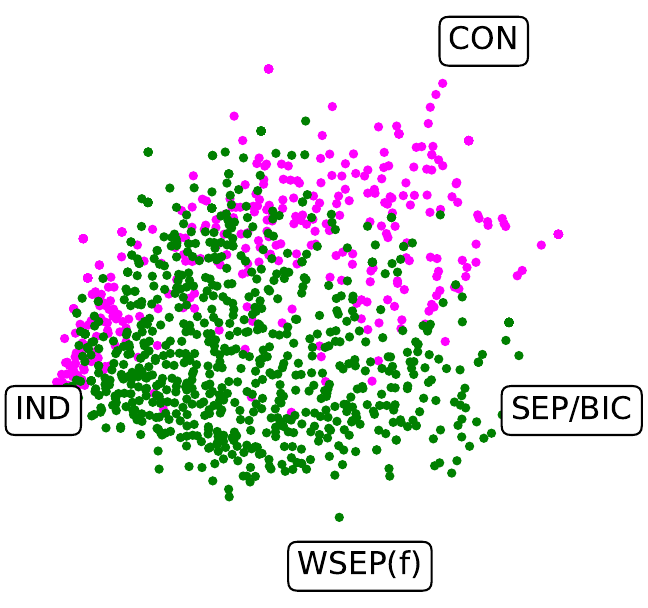}
		\caption{$3\times 6$, valuation distance}
	\end{subfigure}\hfill%
	\begin{subfigure}[t]{\firsttwo\linewidth}
		\centering
		\includegraphics[width=\linewidth]{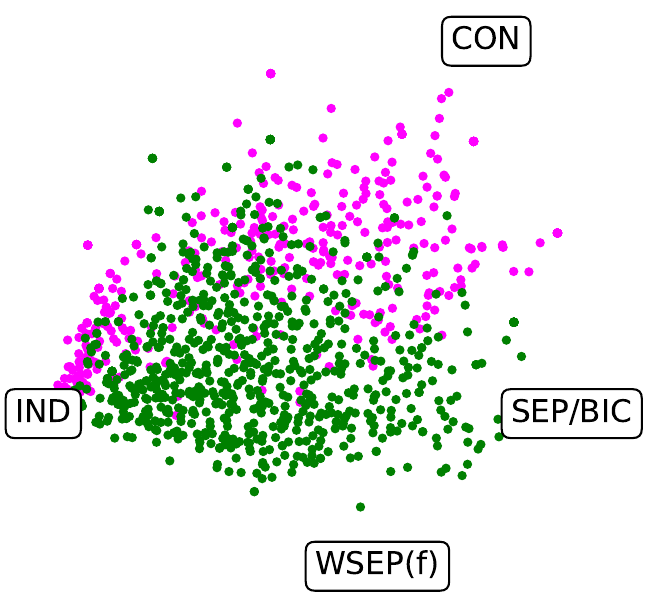}
		\caption{$3\times 6$, demand distance}
	\end{subfigure}\hfill%
	\begin{subfigure}[t]{\fpeval{(0.96 - \firsttwo - \firsttwo)*\linewidth}pt}
		\centering
		\includegraphics[width=\linewidth]{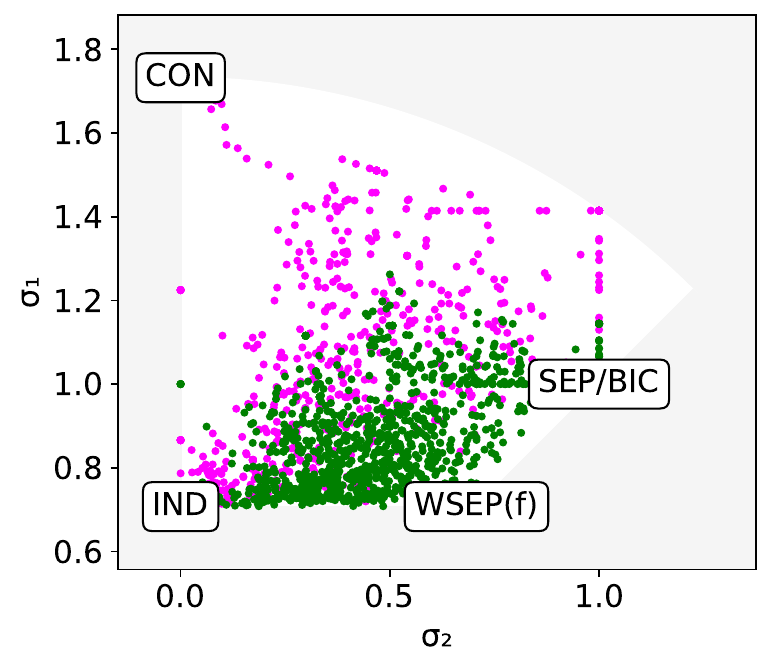}
		\caption{$3\times 6$}
	\end{subfigure}\\
	\begin{subfigure}[t]{\firsttwo\linewidth}
		\centering
		\includegraphics[width=\linewidth]{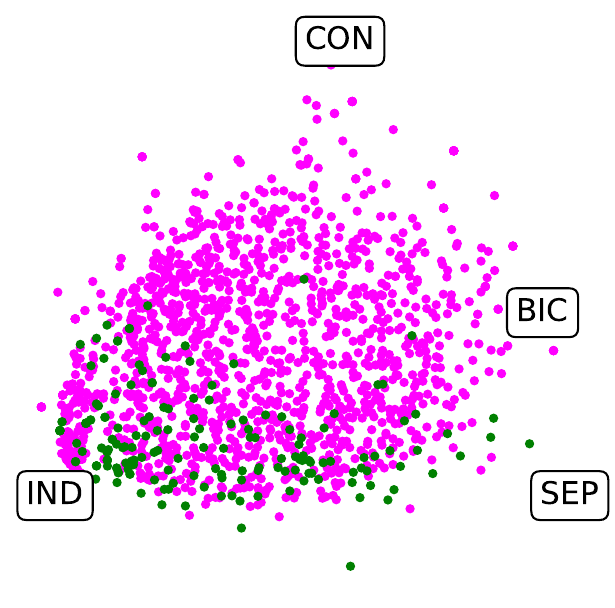}
		\caption{$5\times 5$, valuation distance}
	\end{subfigure}\hfill%
	\begin{subfigure}[t]{\firsttwo\linewidth}
		\centering
		\includegraphics[width=\linewidth]{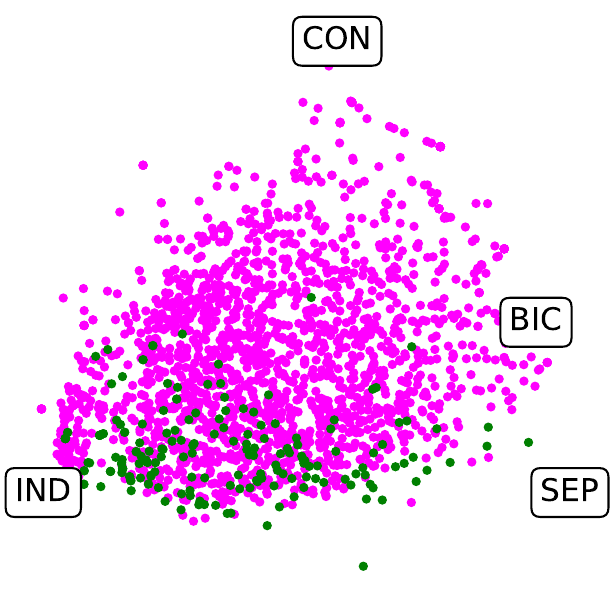}
		\caption{$5\times 5$, demand distance}
	\end{subfigure}\hfill%
	\begin{subfigure}[t]{\fpeval{(0.96 - \firsttwo - \firsttwo)*\linewidth}pt}
		\centering
		\includegraphics[width=\linewidth]{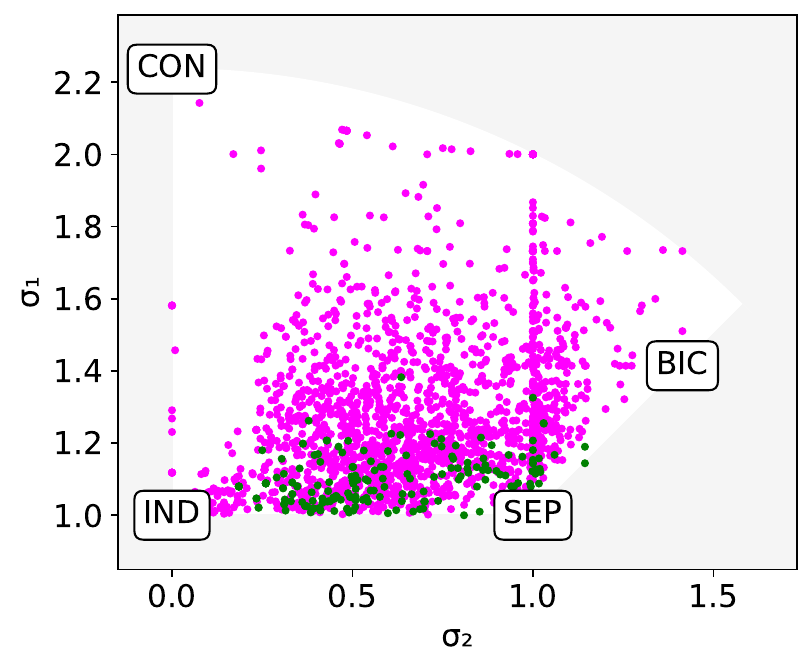}
		\caption{$5\times 5$}
	\end{subfigure}\\
	\begin{subfigure}[t]{\firsttwo\linewidth}
		\quad
	\end{subfigure}\hfill%
	\begin{subfigure}[t]{\firsttwo\linewidth}
		\centering
		\caption{$10\times 20$, demand distance}
	\end{subfigure}\hfill%
	\begin{subfigure}[t]{\fpeval{(0.96 - \firsttwo - \firsttwo)*\linewidth}pt}
		\centering
		\includegraphics[width=\linewidth]{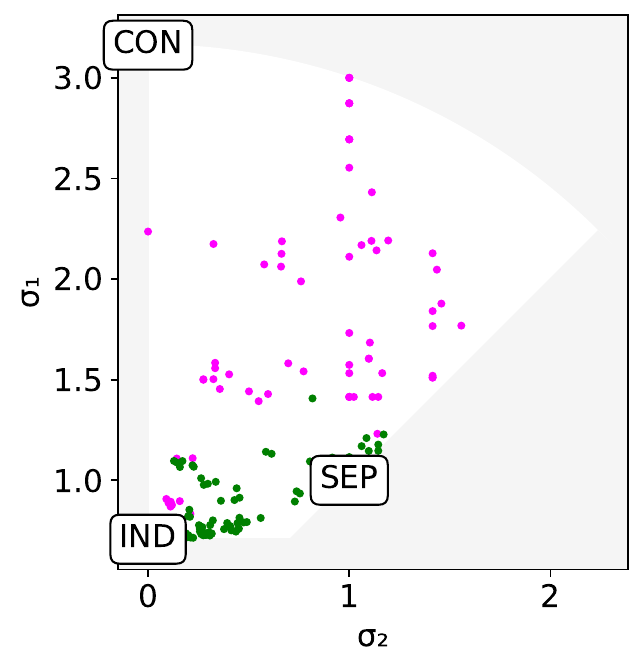}
		\caption{$10\times 20$}
	\end{subfigure}\\
	\caption{Distribution of the existence of an envy-free allocation on our
	distance-embedding map using the valuation distance (first column), using the
demand distance (second column), and on our explicit map (third column). Green
points indicate that an envy-free allocation exists. Computing the valuation
distance for instances of the~$10 \times 20$~dataset was computationally too
demanding, hence the blank space in the third row.}
  \label{fig:all:exef} 
	\let\firsttwo\undefined
\end{figure*}

\begin{figure*}\centering
	\def\firsttwo{.29}
	\begin{subfigure}[t]{\firsttwo\linewidth}
		\centering
		\includegraphics[width=\linewidth]{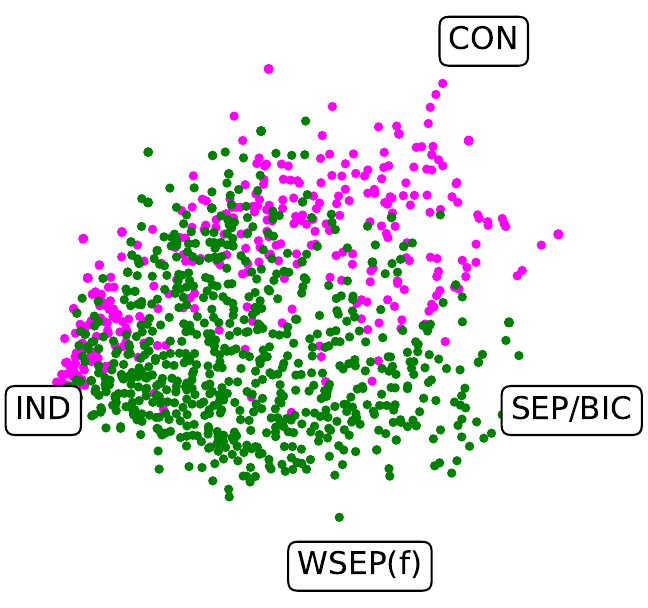}
		\caption{$3\times 6$, valuation distance}
	\end{subfigure}\hfill%
	\begin{subfigure}[t]{\firsttwo\linewidth}
		\centering
		\includegraphics[width=\linewidth]{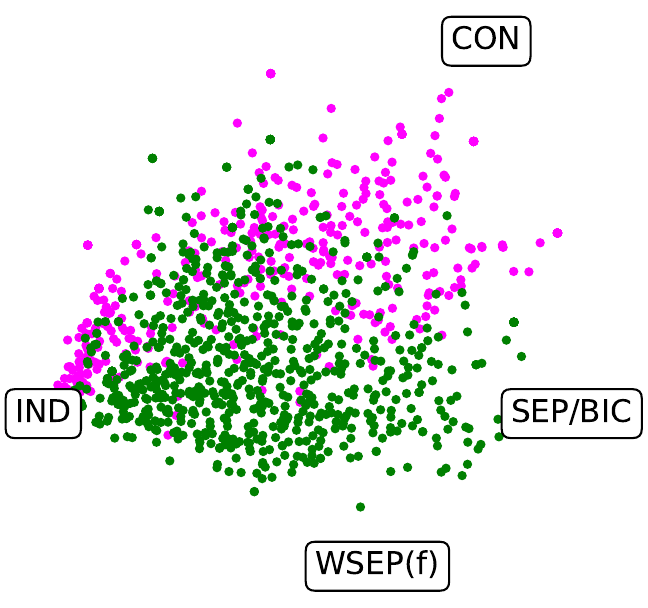}
		\caption{$3\times 6$, demand distance}
	\end{subfigure}\hfill%
	\begin{subfigure}[t]{\fpeval{(0.96 - \firsttwo - \firsttwo)*\linewidth}pt}
		\centering
		\includegraphics[width=\linewidth]{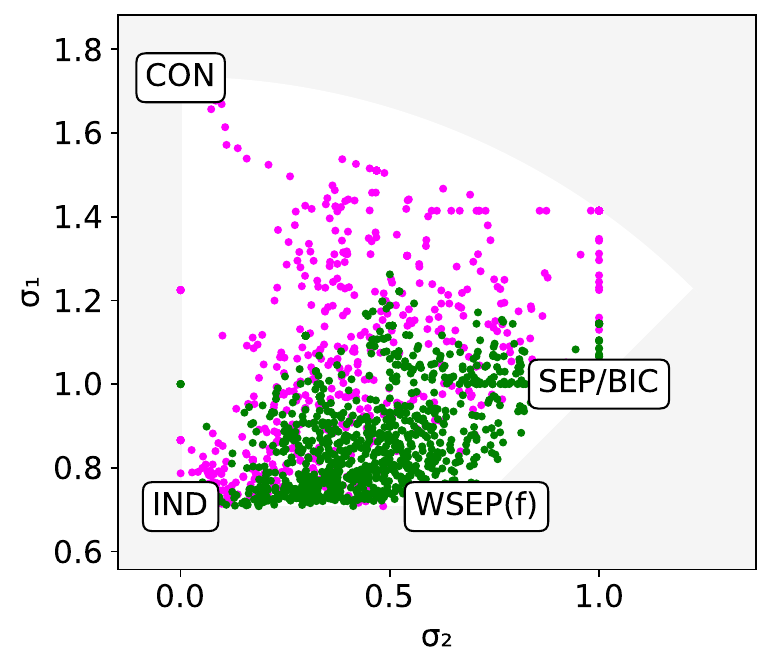}
		\caption{$3\times 6$}
	\end{subfigure}\\
	\begin{subfigure}[t]{\firsttwo\linewidth}
		\centering
		\includegraphics[width=\linewidth]{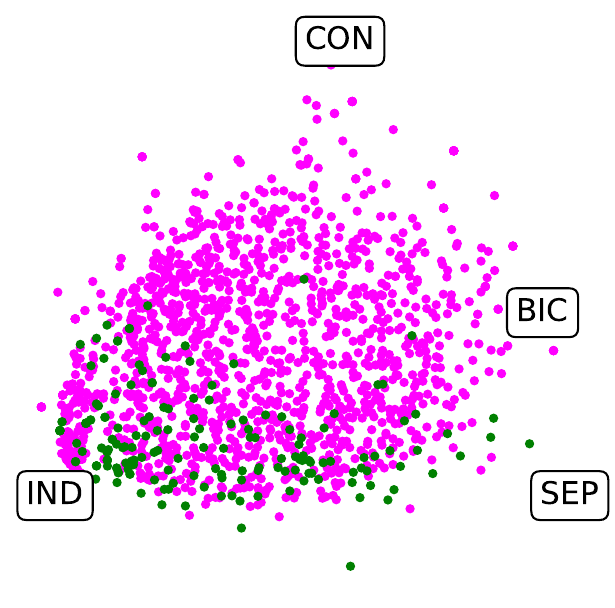}
		\caption{$5\times 5$, valuation distance}
	\end{subfigure}\hfill%
	\begin{subfigure}[t]{\firsttwo\linewidth}
		\centering
		\includegraphics[width=\linewidth]{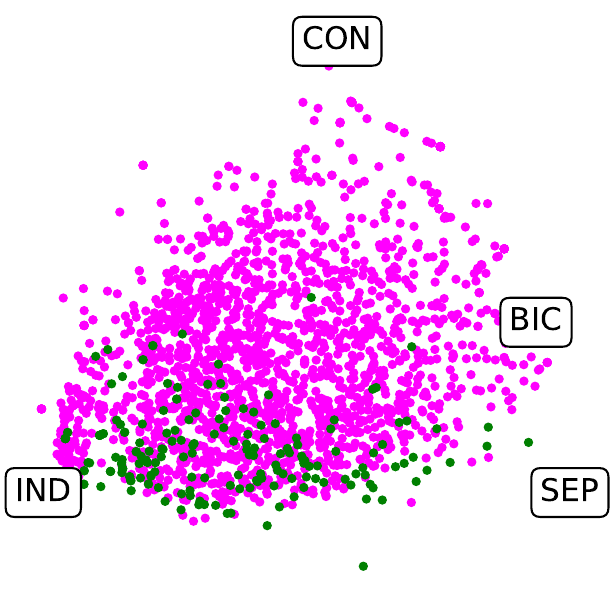}
		\caption{$5\times 5$, demand distance}
	\end{subfigure}\hfill%
	\begin{subfigure}[t]{\fpeval{(0.96 - \firsttwo - \firsttwo)*\linewidth}pt}
		\centering
		\includegraphics[width=\linewidth]{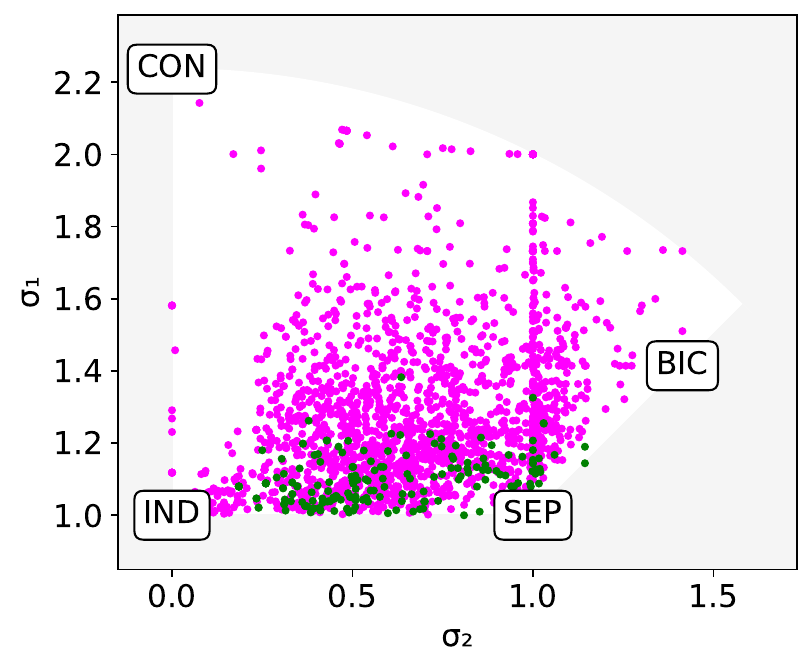}
		\caption{$5\times 5$}
	\end{subfigure}\\
	\begin{subfigure}[t]{\firsttwo\linewidth}
		\quad
	\end{subfigure}\hfill%
	\begin{subfigure}[t]{\firsttwo\linewidth}
		\centering
		\includegraphics[width=\linewidth]{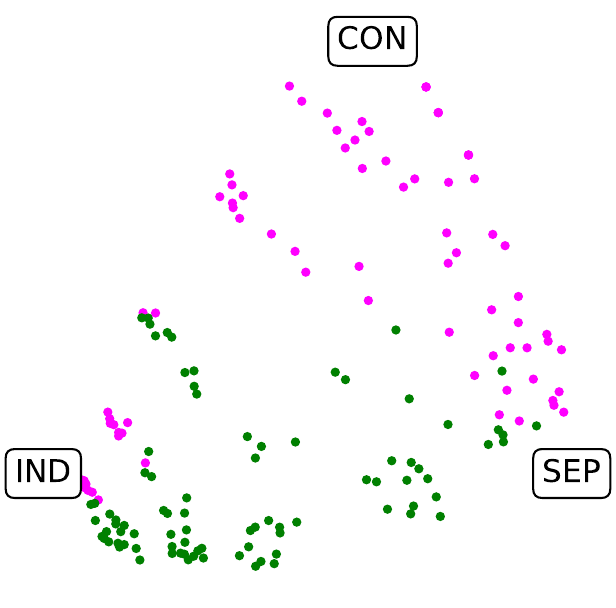}
		\caption{$10\times 20$, demand distance}
	\end{subfigure}\hfill%
	\begin{subfigure}[t]{\fpeval{(0.96 - \firsttwo - \firsttwo)*\linewidth}pt}
		\centering
		\includegraphics[width=\linewidth]{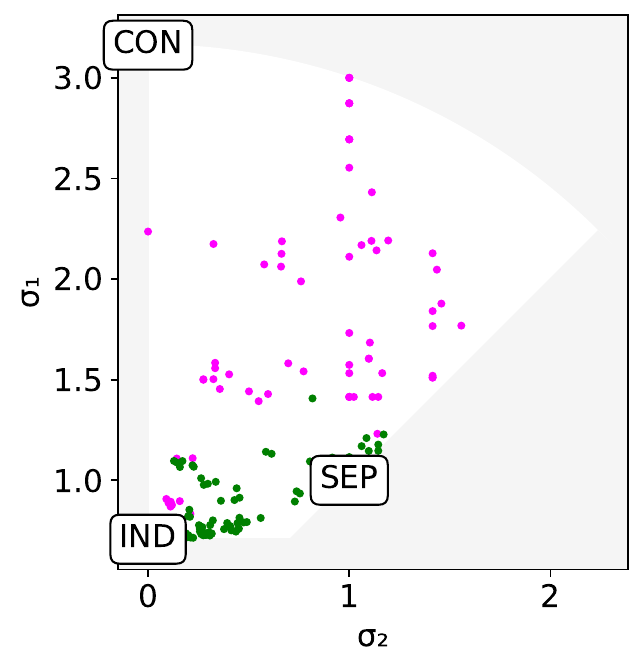}
		\caption{$10\times 20$}
	\end{subfigure}\\
	\caption{Distribution of the existence of an envy-free and Pareto-efficient
	allocation on our distance-embedding map using the valuation distance (first
column), using the demand distance (second column), and on our explicit map
(third column). Green points indicate that an envy-free and Pareto-efficient
allocation exists. Computing the valuation distance for instances of the~$10
\times 20$~dataset was computationally too demanding, hence the blank space in
the third row.}
  \label{fig:all:exefp} 
	\let\firsttwo\undefined
\end{figure*}

\begin{figure*}\centering
	\def\firsttwo{.29}
	\begin{subfigure}[t]{\firsttwo\linewidth}
		\centering
		\includegraphics[width=\linewidth]{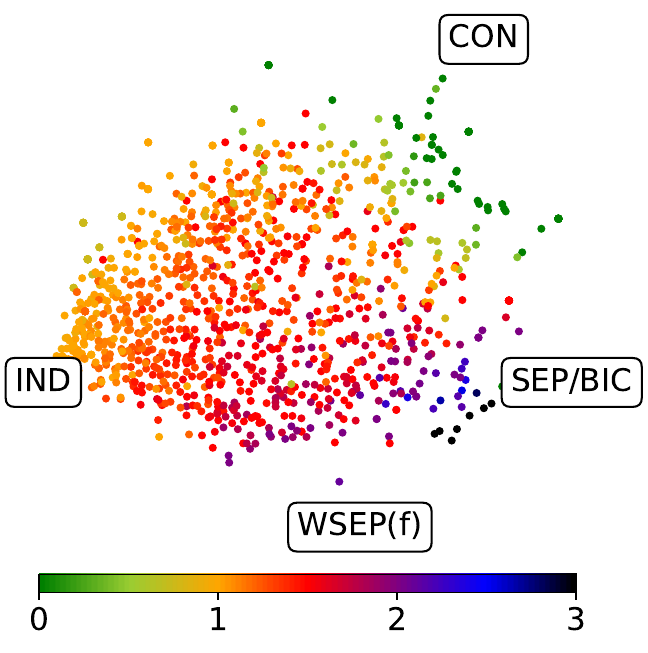}
		\caption{$3\times 6$, valuation distance}
	\end{subfigure}\hfill%
	\begin{subfigure}[t]{\firsttwo\linewidth}
		\centering
		\includegraphics[width=\linewidth]{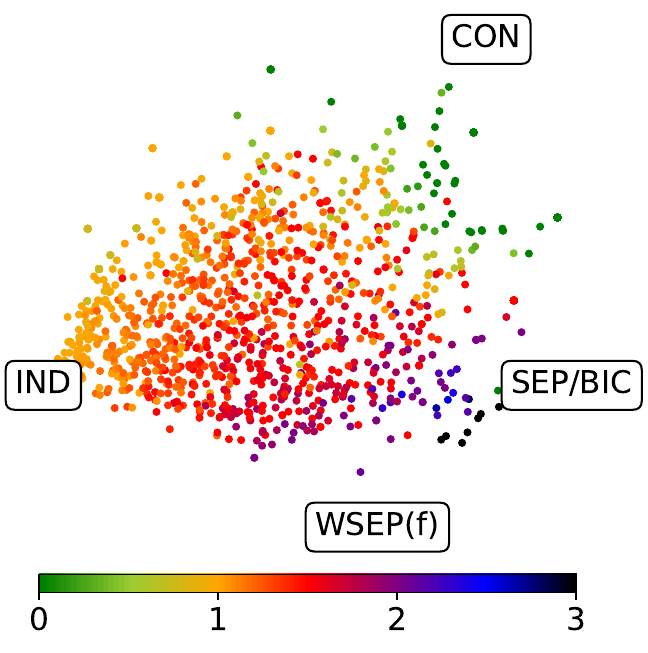}
		\caption{$3\times 6$, demand distance}
	\end{subfigure}\hfill%
	\begin{subfigure}[t]{\fpeval{(0.96 - \firsttwo - \firsttwo)*\linewidth}pt}
		\centering
		\includegraphics[width=\linewidth]{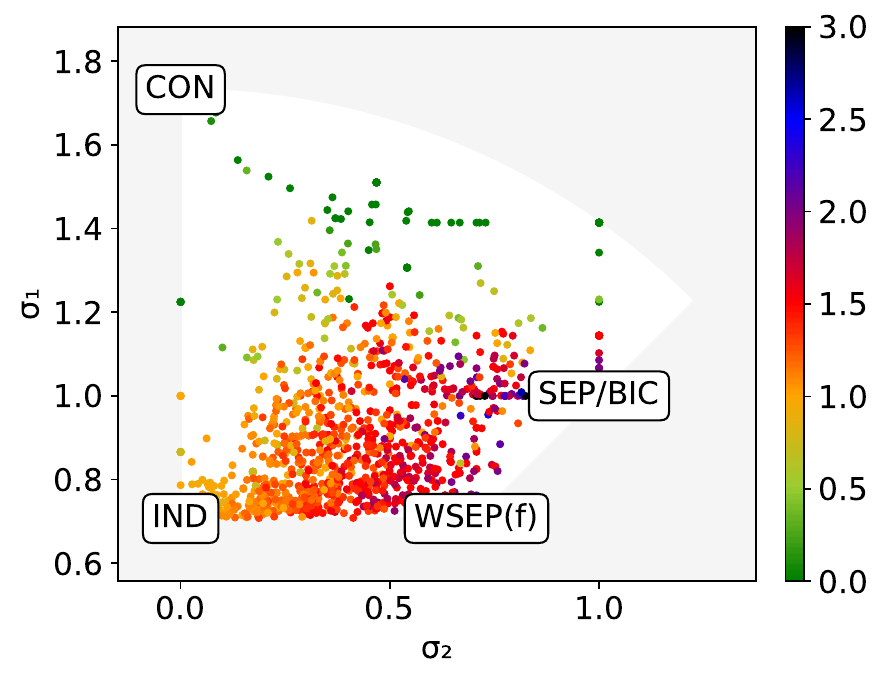}
		\caption{$3\times 6$}
	\end{subfigure}\\
	\begin{subfigure}[t]{\firsttwo\linewidth}
		\centering
		\includegraphics[width=\linewidth]{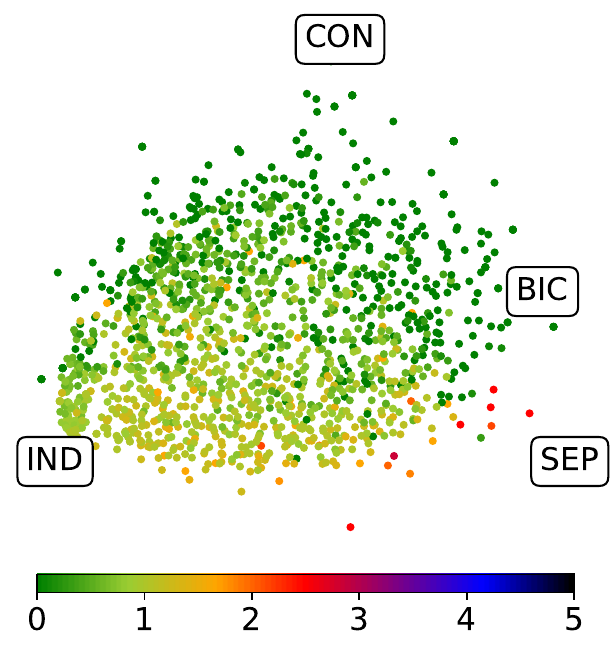}
		\caption{$5\times 5$, valuation distance}
	\end{subfigure}\hfill%
	\begin{subfigure}[t]{\firsttwo\linewidth}
		\centering
		\includegraphics[width=\linewidth]{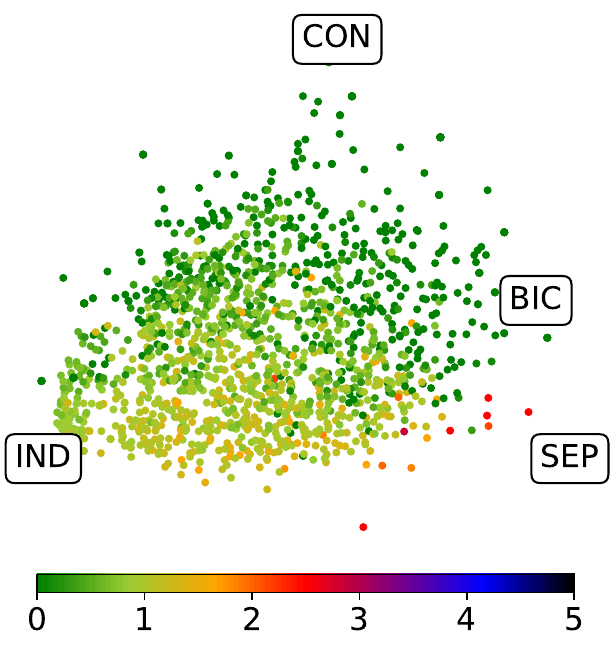}
		\caption{$5\times 5$, demand distance}
	\end{subfigure}\hfill%
	\begin{subfigure}[t]{\fpeval{(0.96 - \firsttwo - \firsttwo)*\linewidth}pt}
		\centering
		\includegraphics[width=\linewidth]{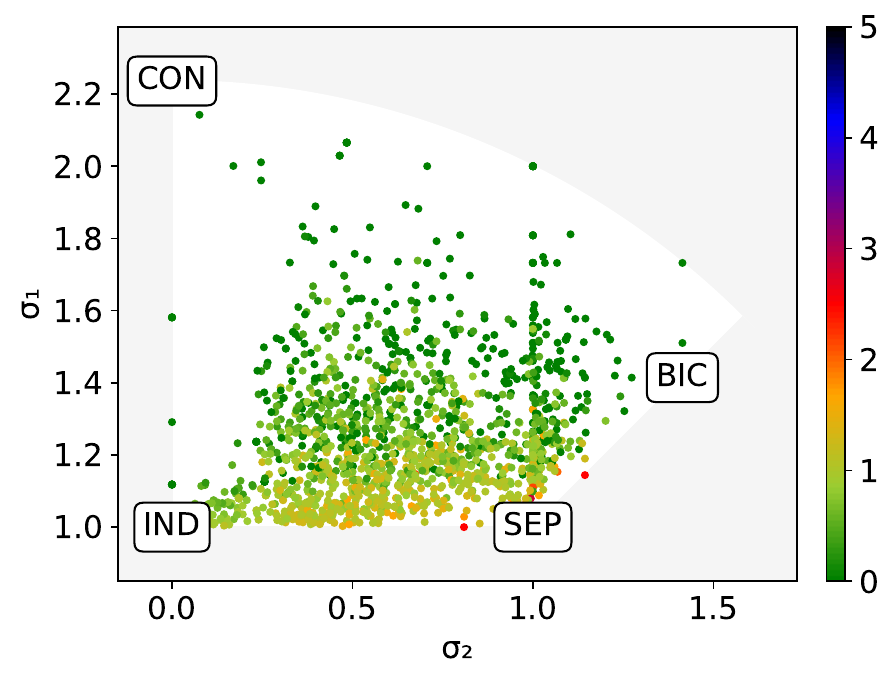}
		\caption{$5\times 5$}
	\end{subfigure}\\
	\begin{subfigure}[t]{\firsttwo\linewidth}
		\quad
	\end{subfigure}\hfill%
	\begin{subfigure}[t]{\firsttwo\linewidth}
		\centering
		\includegraphics[width=\linewidth]{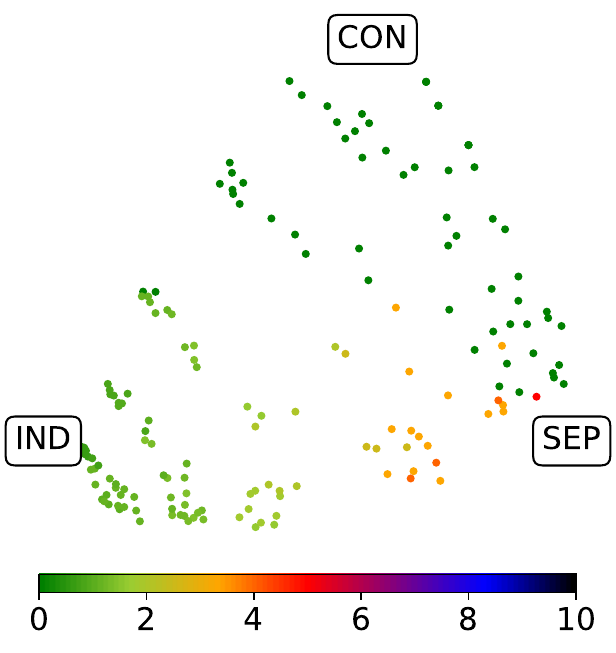}
		\caption{$10\times 20$, demand distance}
	\end{subfigure}\hfill%
	\begin{subfigure}[t]{\fpeval{(0.96 - \firsttwo - \firsttwo)*\linewidth}pt}
		\centering
		\includegraphics[width=\linewidth]{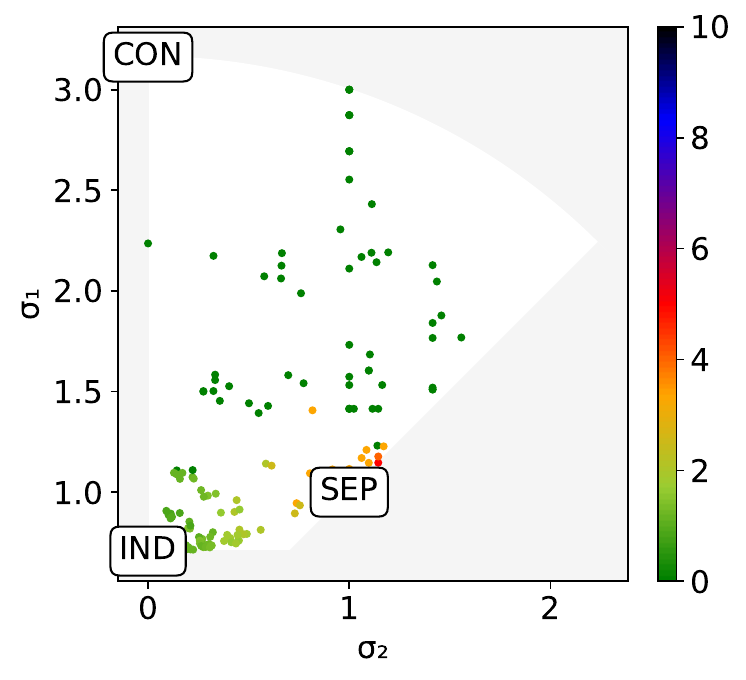}
		\caption{$10\times 20$}
	\end{subfigure}\\
  \caption{Distribution of fraction of the proportional share that
	can be guaranteed at each instance on our distance-embedding map using the
valuation distance (first column), using the demand distance (second column),
and on our explicit map (third column). Computing the valuation distance for
instances of the~$10 \times 20$~dataset was computationally too demanding, hence
the blank space in the third row.}
  \label{fig:all:propshare} 
	\let\firsttwo\undefined
\end{figure*}

\begin{figure*}\centering
	\def\firsttwo{.29}
	\begin{subfigure}[t]{\firsttwo\linewidth}
		\centering
		\includegraphics[width=\linewidth]{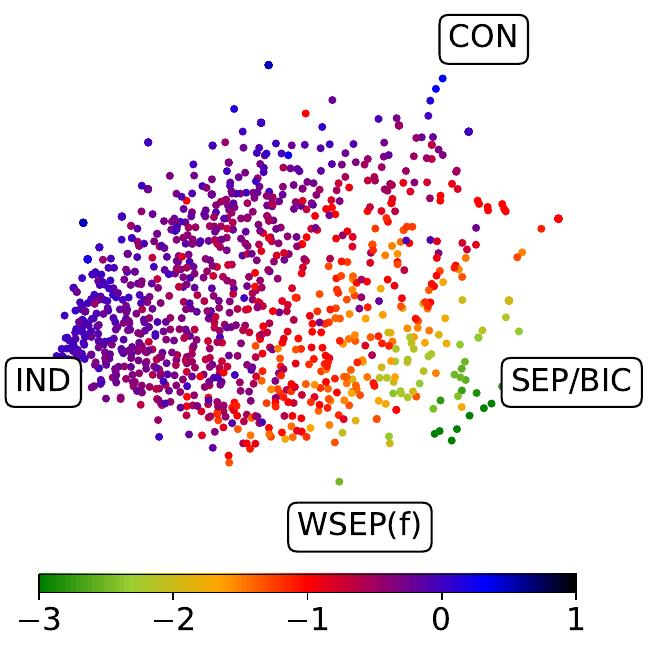}
		\caption{$3\times 6$, valuation distance}
	\end{subfigure}\hfill%
	\begin{subfigure}[t]{\firsttwo\linewidth}
		\centering
		\includegraphics[width=\linewidth]{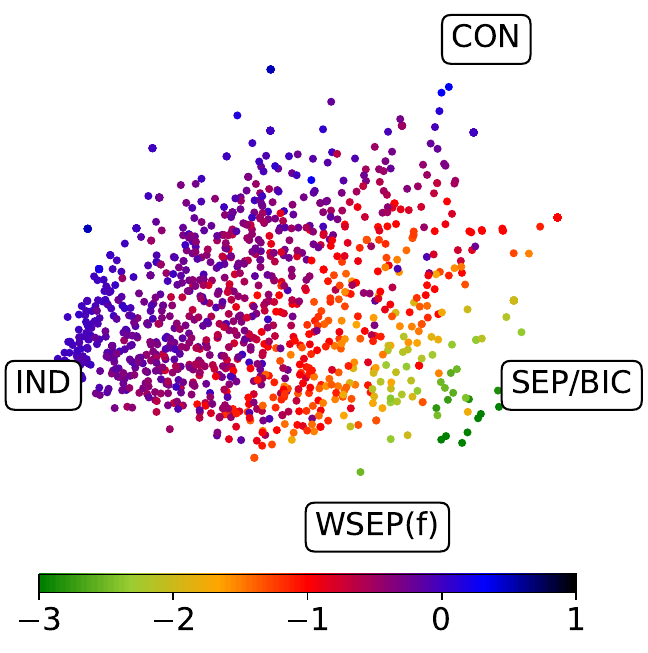}
		\caption{$3\times 6$, demand distance}
	\end{subfigure}\hfill%
	\begin{subfigure}[t]{\fpeval{(0.96 - \firsttwo - \firsttwo)*\linewidth}pt}
		\centering
		\includegraphics[width=\linewidth]{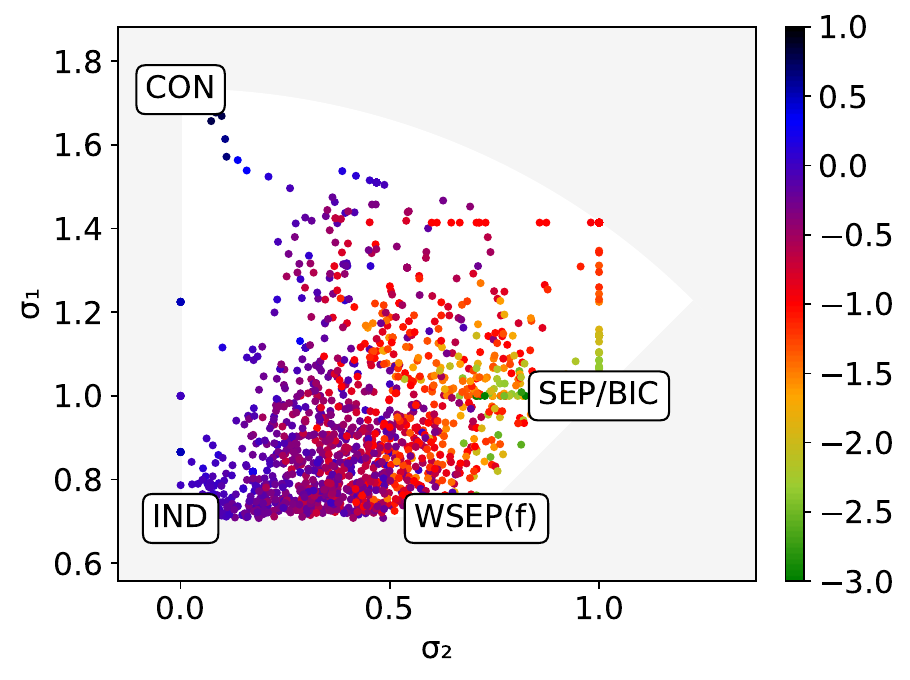}
		\caption{$3\times 6$}
	\end{subfigure}\\
	\begin{subfigure}[t]{\firsttwo\linewidth}
		\centering
		\includegraphics[width=\linewidth]{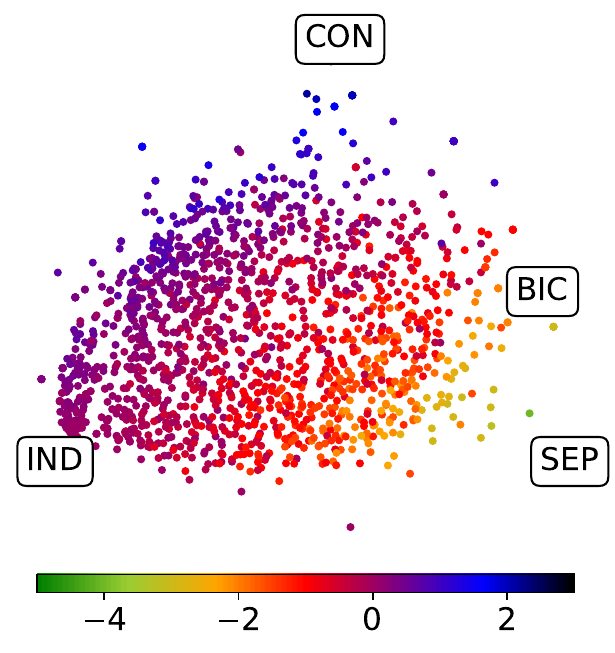}
		\caption{$5\times 5$, valuation distance}
	\end{subfigure}\hfill%
	\begin{subfigure}[t]{\firsttwo\linewidth}
		\centering
		\includegraphics[width=\linewidth]{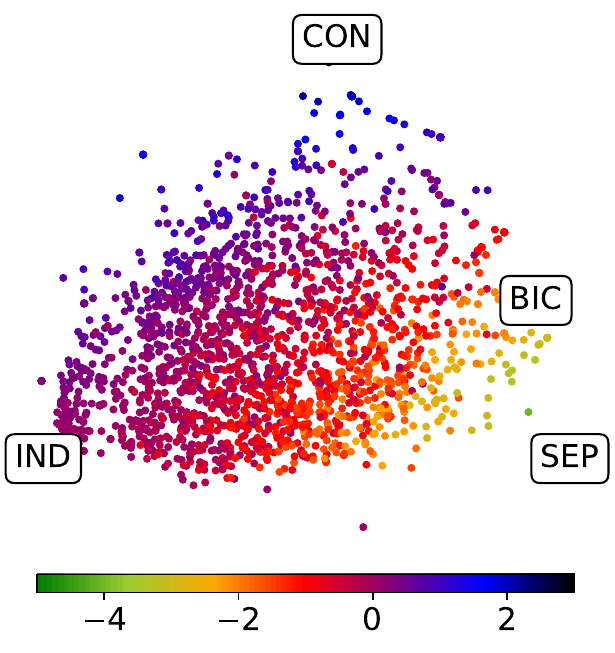}
		\caption{$5\times 5$, demand distance}
	\end{subfigure}\hfill%
	\begin{subfigure}[t]{\fpeval{(0.96 - \firsttwo - \firsttwo)*\linewidth}pt}
		\centering
		\includegraphics[width=\linewidth]{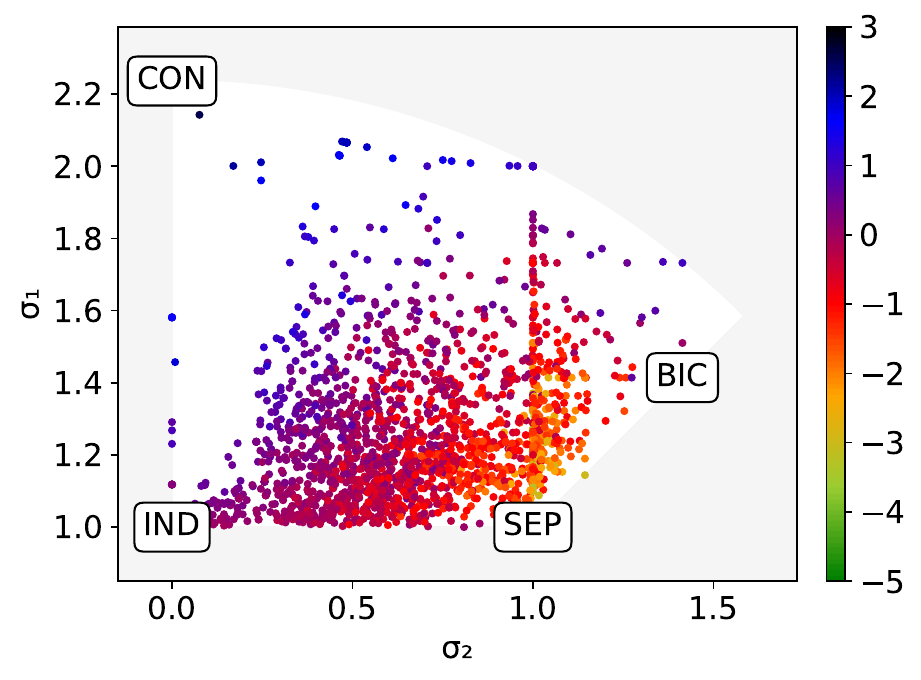}
		\caption{$5\times 5$}
	\end{subfigure}\\
	\begin{subfigure}[t]{\firsttwo\linewidth}
		\quad
	\end{subfigure}\hfill%
	\begin{subfigure}[t]{\firsttwo\linewidth}
		\centering
		\includegraphics[width=\linewidth]{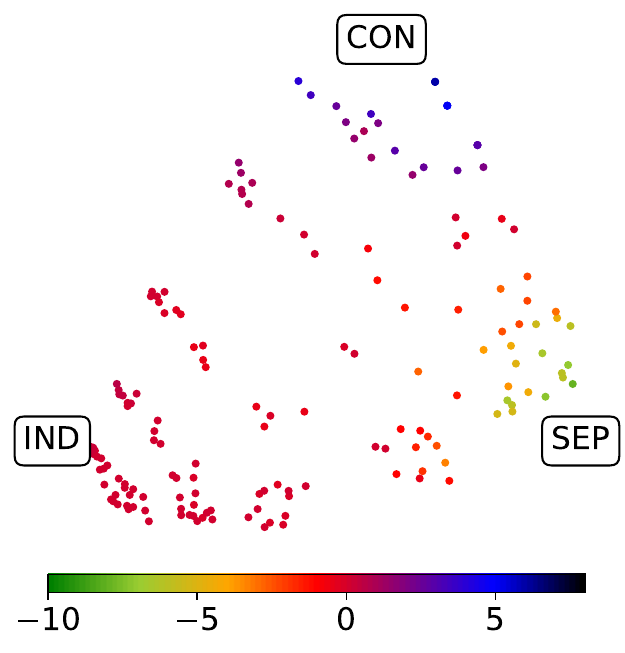}
		\caption{$10\times 20$, demand distance}
	\end{subfigure}\hfill%
	\begin{subfigure}[t]{\fpeval{(0.96 - \firsttwo - \firsttwo)*\linewidth}pt}
		\centering
		\includegraphics[width=\linewidth]{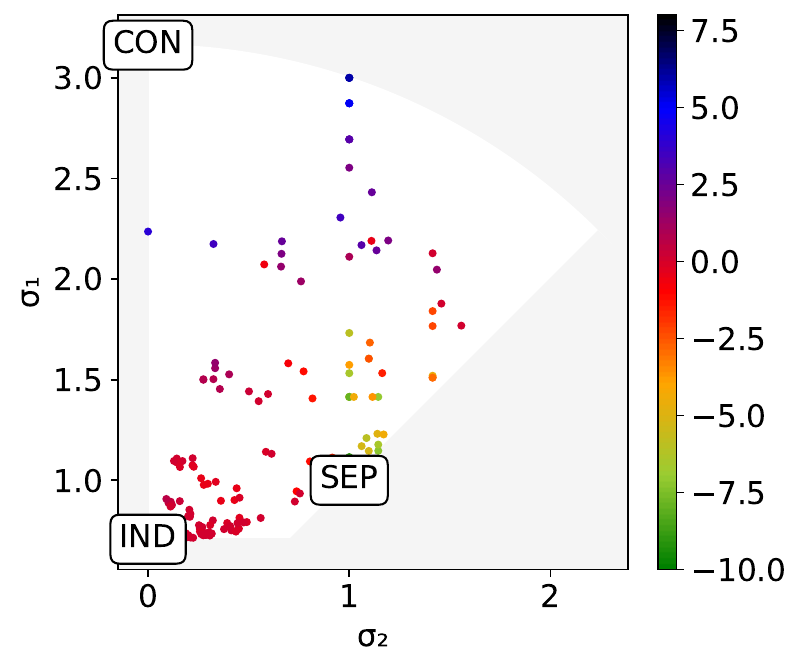}
		\caption{$10\times 20$}
	\end{subfigure}\\
	\caption{Distribution of the sum of the maximal envies on our
	distance-embedding map using the valuation distance (first column), using the
demand distance (second column), and on our explicit map (third column).
Computing the valuation distance for instances of the~$10 \times 20$~dataset was
computationally too demanding, hence the blank space in the third row.}
  \label{fig:all:sumabs} 
	\let\firsttwo\undefined
\end{figure*}

\begin{figure*}\centering
	\def\firsttwo{.29}
	\begin{subfigure}[t]{\firsttwo\linewidth}
		\centering
		\includegraphics[width=\linewidth]{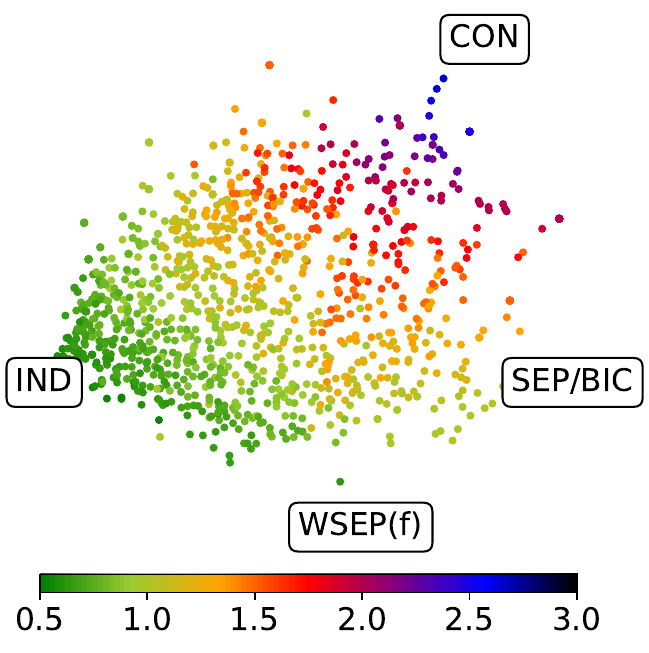}
		\caption{$3\times 6$, valuation distance}
	\end{subfigure}\hfill%
	\begin{subfigure}[t]{\firsttwo\linewidth}
		\centering
		\includegraphics[width=\linewidth]{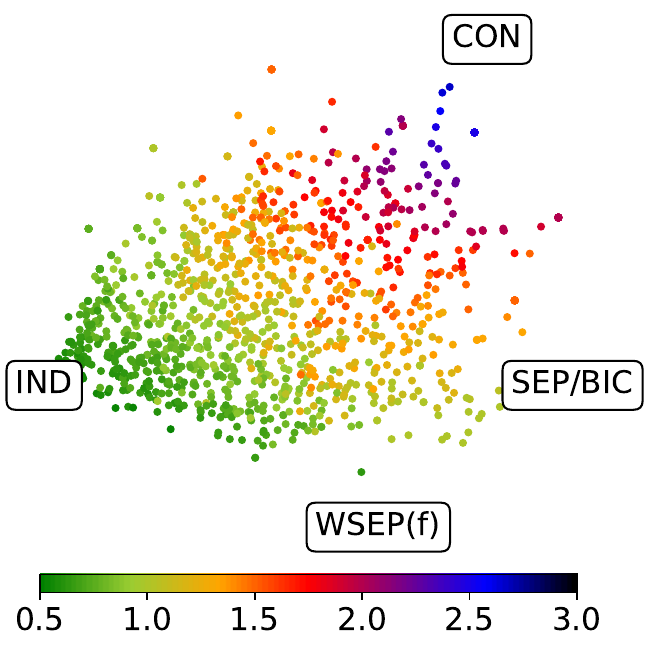}
		\caption{$3\times 6$, demand distance}
	\end{subfigure}\hfill%
	\begin{subfigure}[t]{\fpeval{(0.96 - \firsttwo - \firsttwo)*\linewidth}pt}
		\centering
		\includegraphics[width=\linewidth]{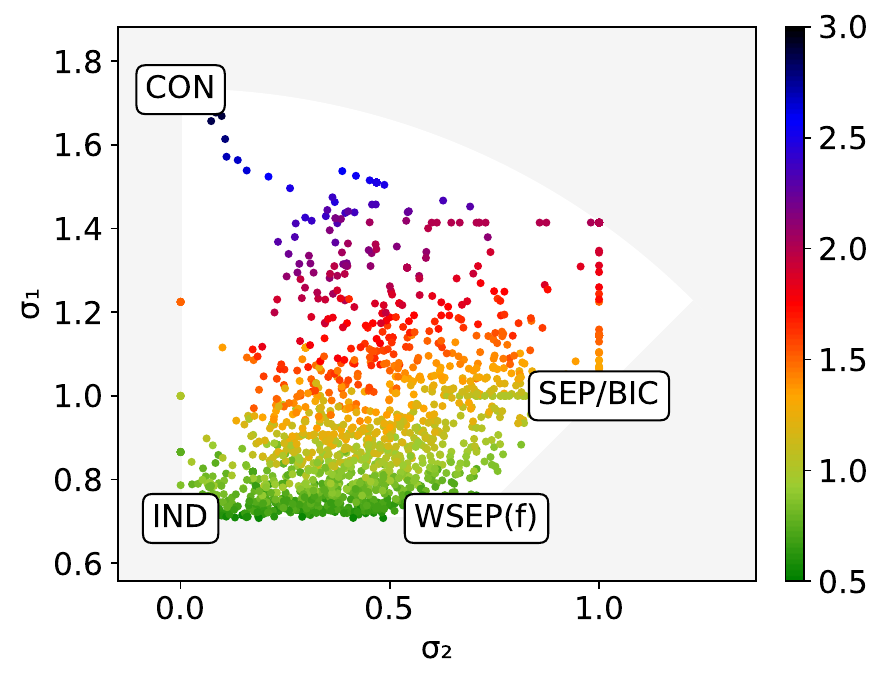}
		\caption{$3\times 6$}
	\end{subfigure}\\
	\begin{subfigure}[t]{\firsttwo\linewidth}
		\centering
		\includegraphics[width=\linewidth]{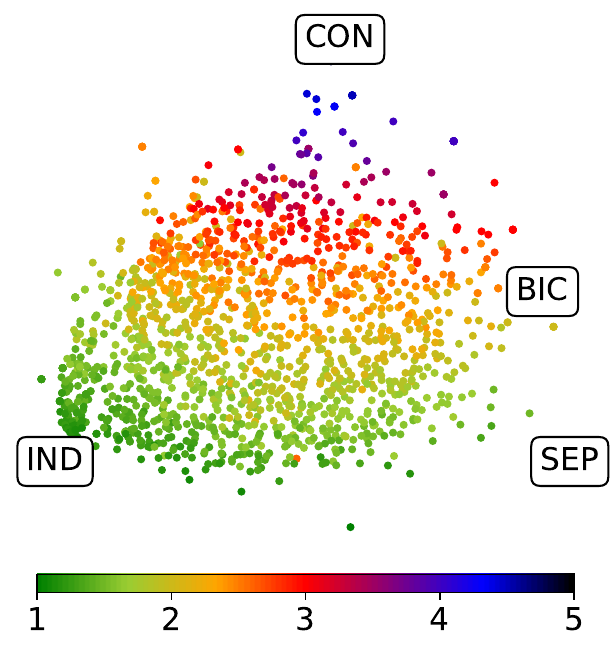}
		\caption{$5\times 5$, valuation distance}
	\end{subfigure}\hfill%
	\begin{subfigure}[t]{\firsttwo\linewidth}
		\centering
		\includegraphics[width=\linewidth]{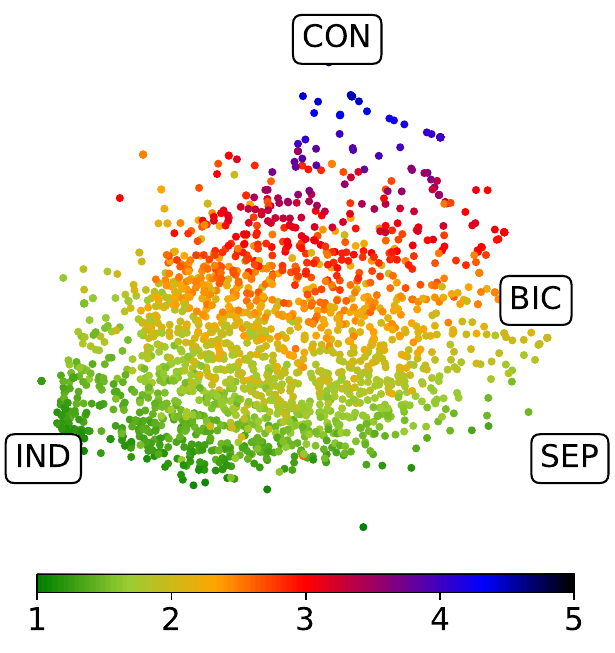}
		\caption{$5\times 5$, demand distance}
	\end{subfigure}\hfill%
	\begin{subfigure}[t]{\fpeval{(0.96 - \firsttwo - \firsttwo)*\linewidth}pt}
		\centering
		\includegraphics[width=\linewidth]{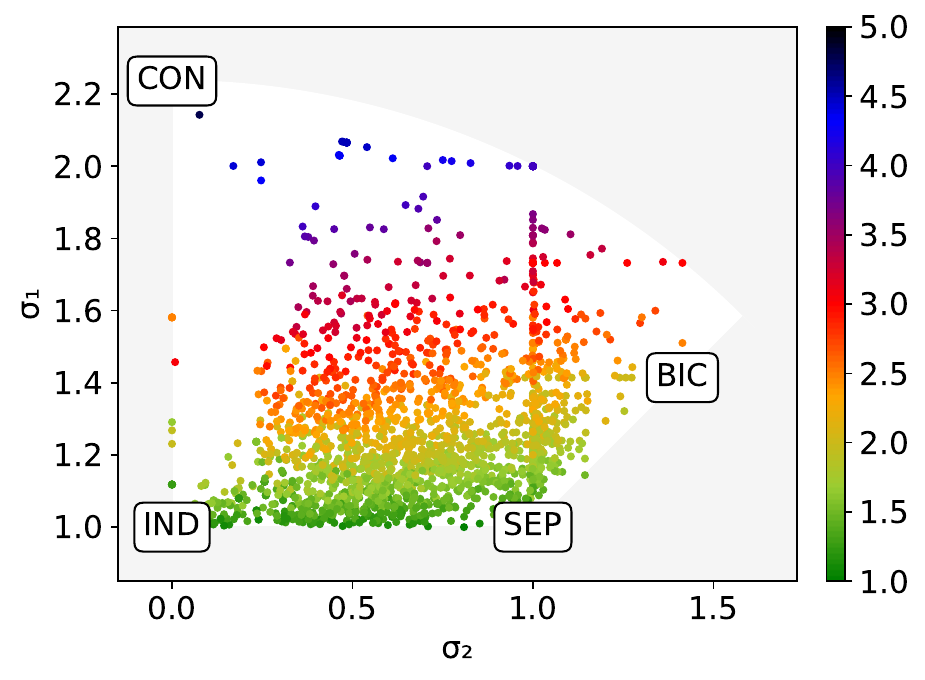}
		\caption{$5\times 5$}
	\end{subfigure}\\
	\begin{subfigure}[t]{\firsttwo\linewidth}
		\quad
	\end{subfigure}\hfill%
	\begin{subfigure}[t]{\firsttwo\linewidth}
		\centering
		\includegraphics[width=\linewidth]{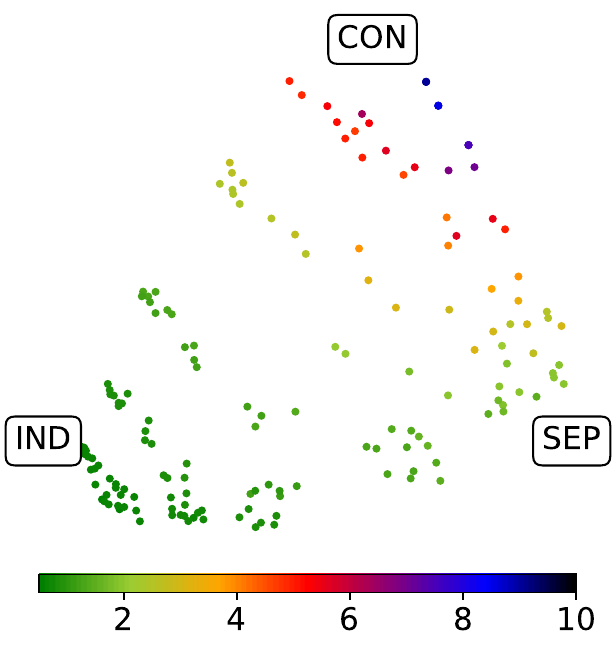}
		\caption{$10\times 20$, demand distance}
	\end{subfigure}\hfill%
	\begin{subfigure}[t]{\fpeval{(0.96 - \firsttwo - \firsttwo)*\linewidth}pt}
		\centering
		\includegraphics[width=\linewidth]{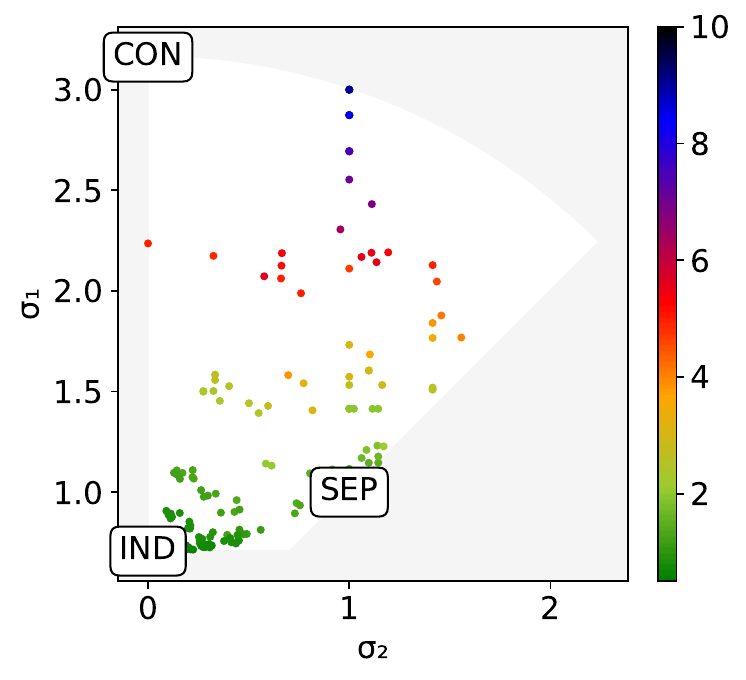}
		\caption{$10\times 20$}
	\end{subfigure}\\
	\caption{Distribution of the maximum demand on our distance-embedding map
	using the valuation distance (first column), using the demand distance (second
column), and on our explicit map (third column). Computing the valuation
distance for instances of the~$10 \times 20$~dataset was computationally too
demanding, hence the blank space in the third row.}
  \label{fig:all:maxd} 
	\let\firsttwo\undefined
\end{figure*}

\begin{figure*}\centering
	\def\firsttwo{.29}
	\begin{subfigure}[t]{\firsttwo\linewidth}
		\centering
		\includegraphics[width=\linewidth]{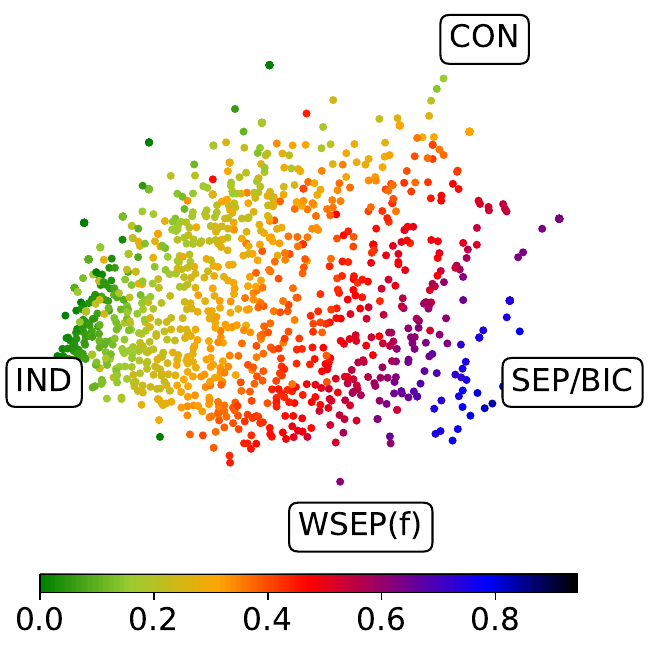}
		\caption{$3\times 6$, valuation distance}
	\end{subfigure}\hfill%
	\begin{subfigure}[t]{\firsttwo\linewidth}
		\centering
		\includegraphics[width=\linewidth]{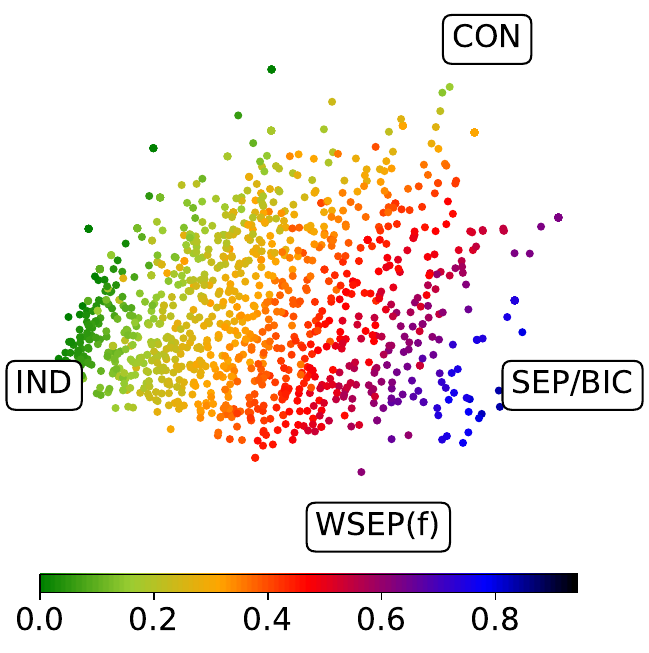}
		\caption{$3\times 6$, demand distance}
	\end{subfigure}\hfill%
	\begin{subfigure}[t]{\fpeval{(0.96 - \firsttwo - \firsttwo)*\linewidth}pt}
		\centering
		\includegraphics[width=\linewidth]{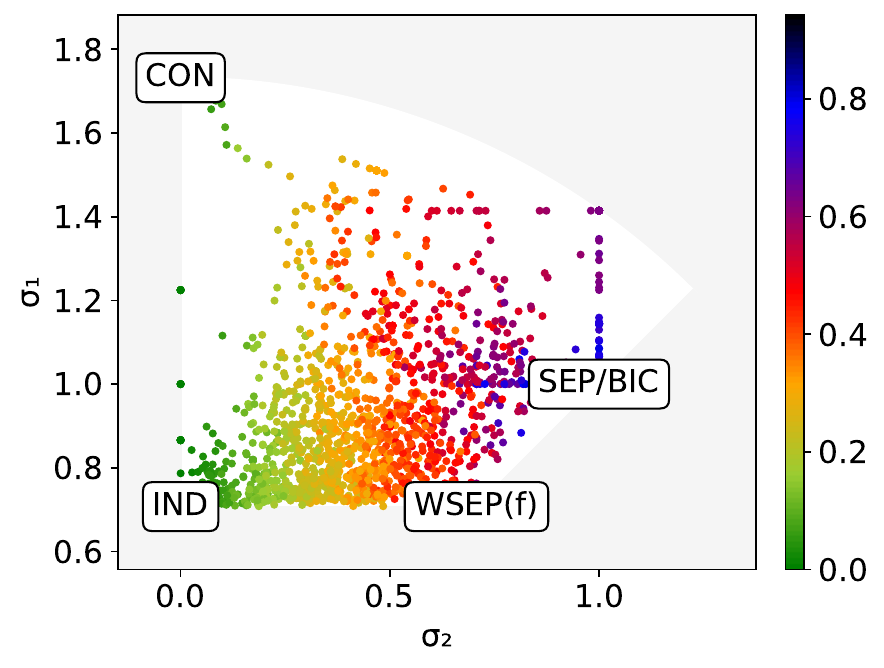}
		\caption{$3\times 6$}
	\end{subfigure}\\
	\begin{subfigure}[t]{\firsttwo\linewidth}
		\centering
		\includegraphics[width=\linewidth]{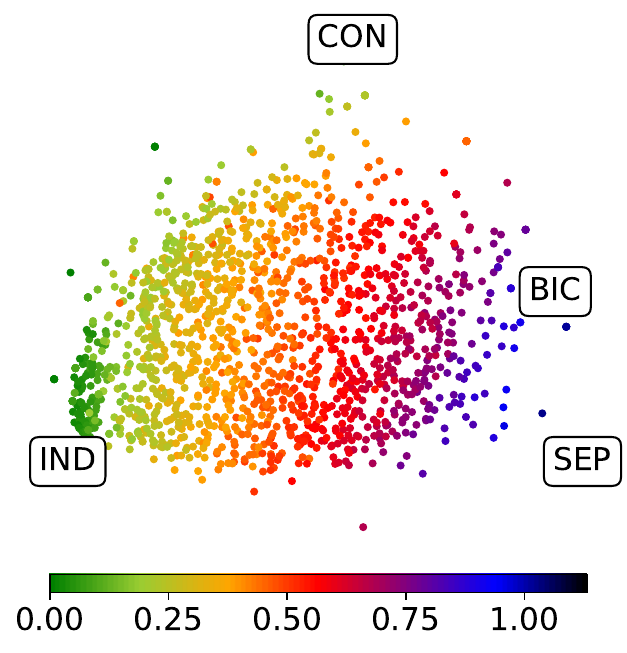}
		\caption{$5\times 5$, valuation distance}
	\end{subfigure}\hfill%
	\begin{subfigure}[t]{\firsttwo\linewidth}
		\centering
		\includegraphics[width=\linewidth]{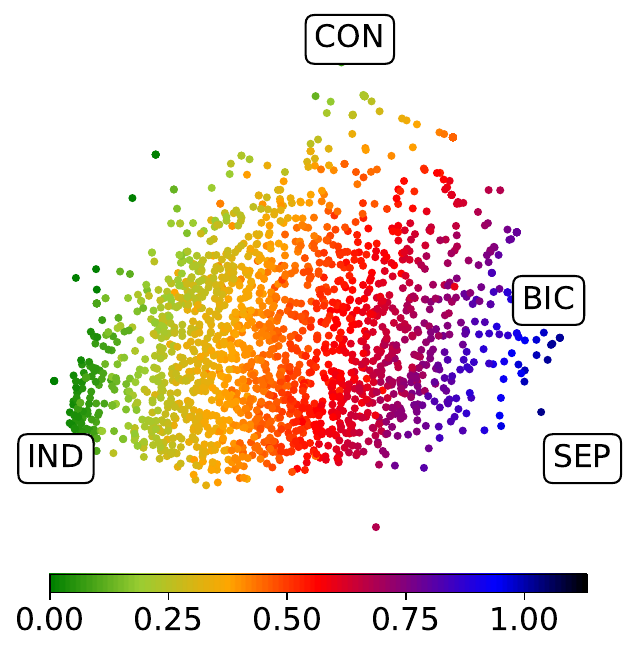}
		\caption{$5\times 5$, demand distance}
	\end{subfigure}\hfill%
	\begin{subfigure}[t]{\fpeval{(0.96 - \firsttwo - \firsttwo)*\linewidth}pt}
		\centering
		\includegraphics[width=\linewidth]{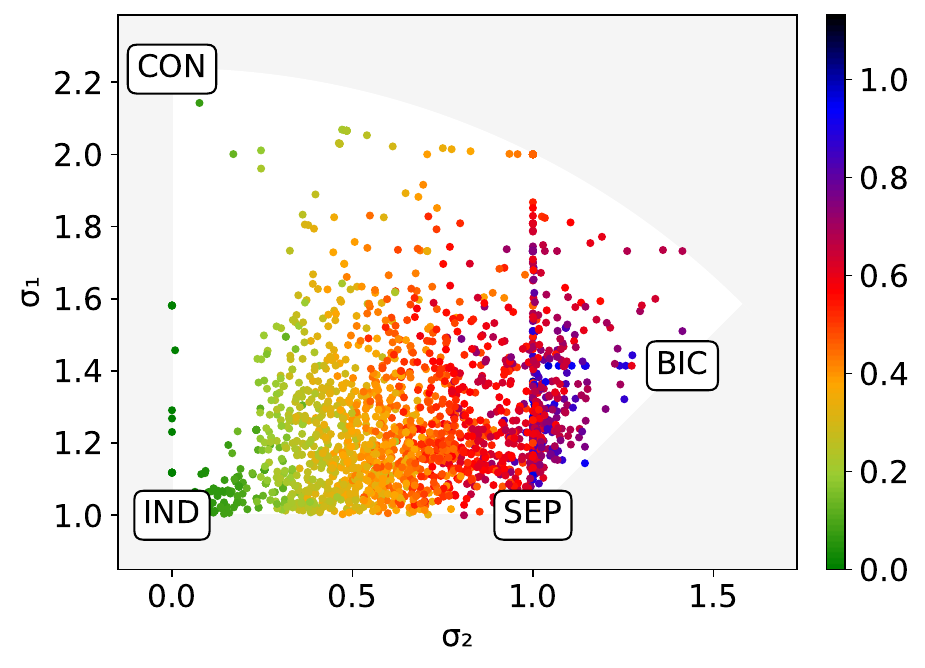}
		\caption{$5\times 5$}
	\end{subfigure}\\
	\begin{subfigure}[t]{\firsttwo\linewidth}
		\quad
	\end{subfigure}\hfill%
	\begin{subfigure}[t]{\firsttwo\linewidth}
		\centering
		\includegraphics[width=\linewidth]{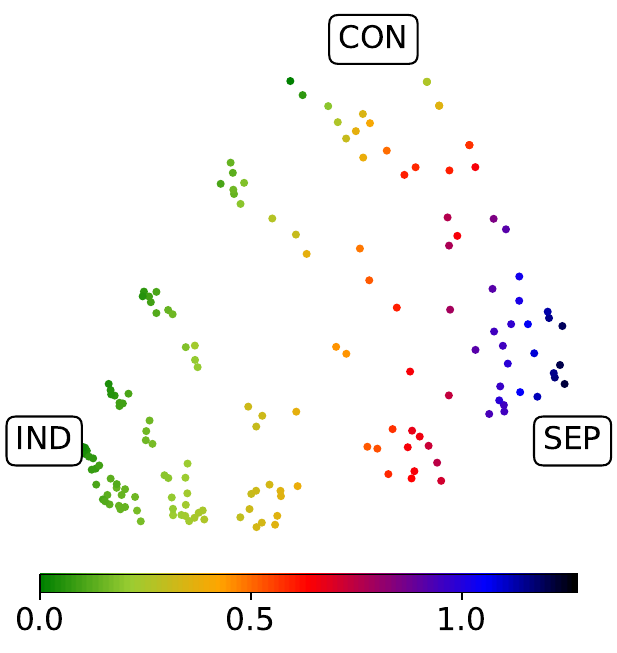}
		\caption{$10\times 20$, demand distance}
	\end{subfigure}\hfill%
	\begin{subfigure}[t]{\fpeval{(0.96 - \firsttwo - \firsttwo)*\linewidth}pt}
		\centering
		\includegraphics[width=\linewidth]{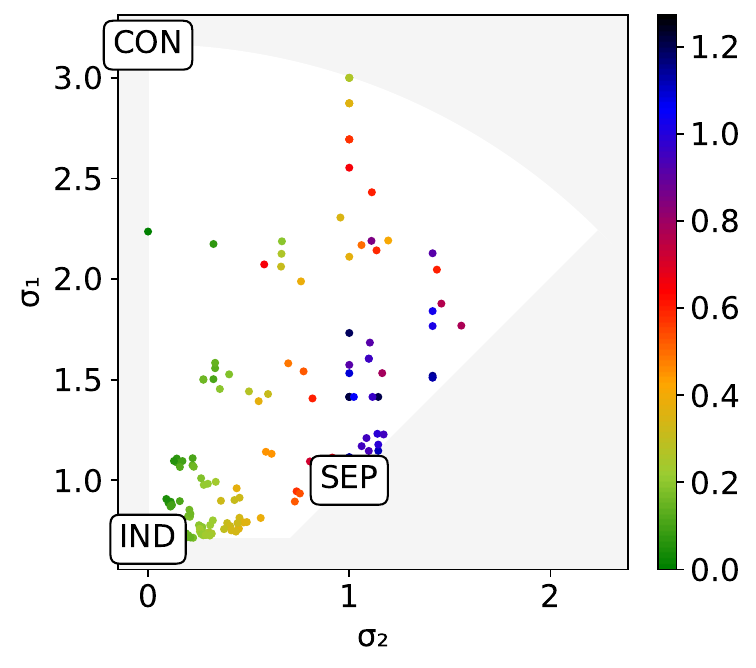}
		\caption{$10\times 20$}
	\end{subfigure}\\
	\caption{Distribution of the preference diversity on our distance-embedding
	map using the valuation distance (first column), using the demand distance
(second column), and on our explicit map (third column). Computing the valuation
distance for instances of the~$10 \times 20$~dataset was computationally too
demanding, hence the blank space in the third row.}
  \label{fig:all:prefdiv} 
	\let\firsttwo\undefined
\end{figure*}

In addition, we introduce the following measures:
To measure the \emph{diversity of demand}, we create a vector of all total demands and compute a Gini coefficient of this vector, where the demand for good $j$ is defined as $\sum_{i \in [n]} u_{i,j}$.
On the other hand, we compute a Gini coefficient of each vote and use the average
over all votes to measure the \emph{pickiness}.
The maps showing these two features can be seen in
\cref{fig:all:divdemand} and \ref{fig:all:minuspicki} (which shows
one minus pickiness, as this makes the value one for one extreme point and
zero for the other extreme points, which is also the case for diversity of
demand and preference diversity), which show that these measures also vary
smoothly over the map.

\begin{figure*}\centering
	\def\firsttwo{.29}
	\begin{subfigure}[t]{\firsttwo\linewidth}
		\centering
		\includegraphics[width=\linewidth]{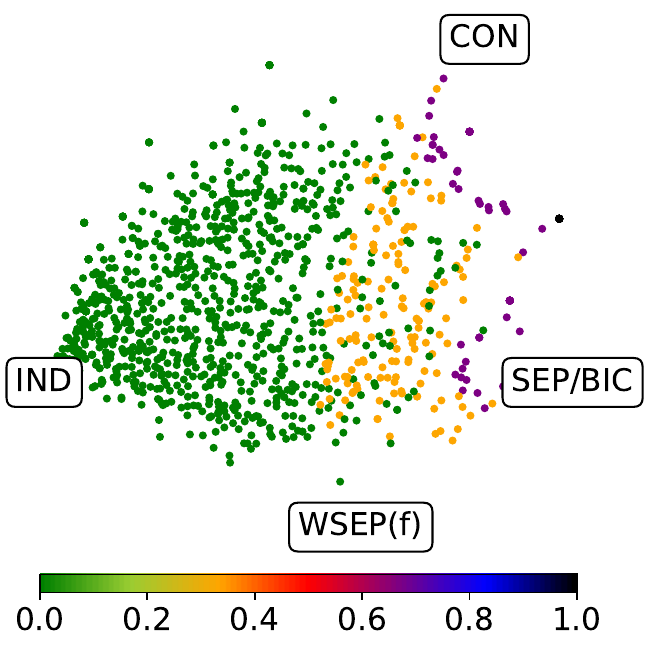}
		\caption{$3\times 6$, valuation distance}
	\end{subfigure}\hfill%
	\begin{subfigure}[t]{\firsttwo\linewidth}
		\centering
		\includegraphics[width=\linewidth]{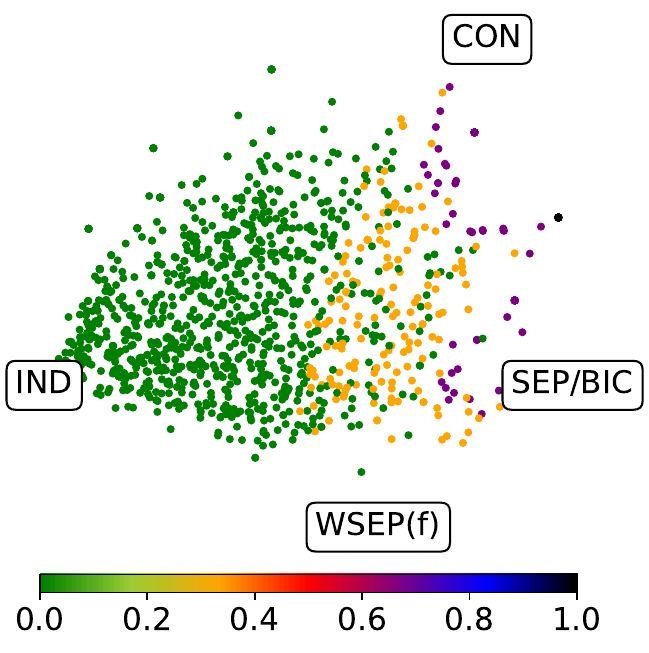}
		\caption{$3\times 6$, demand distance}
	\end{subfigure}\hfill%
	\begin{subfigure}[t]{\fpeval{(0.96 - \firsttwo - \firsttwo)*\linewidth}pt}
		\centering
		\includegraphics[width=\linewidth]{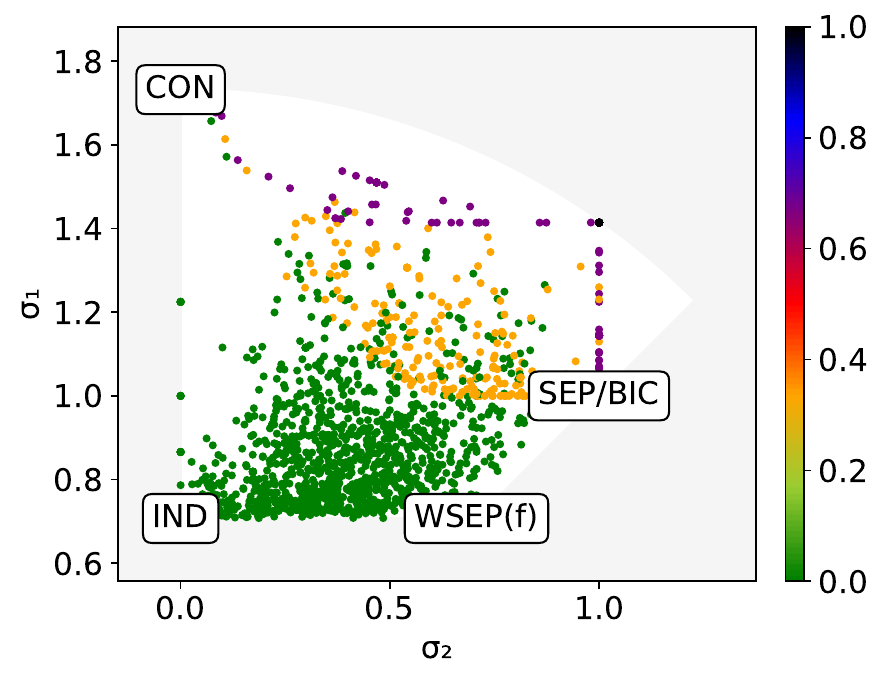}
		\caption{$3\times 6$}
	\end{subfigure}\\
	\begin{subfigure}[t]{\firsttwo\linewidth}
		\centering
		\includegraphics[width=\linewidth]{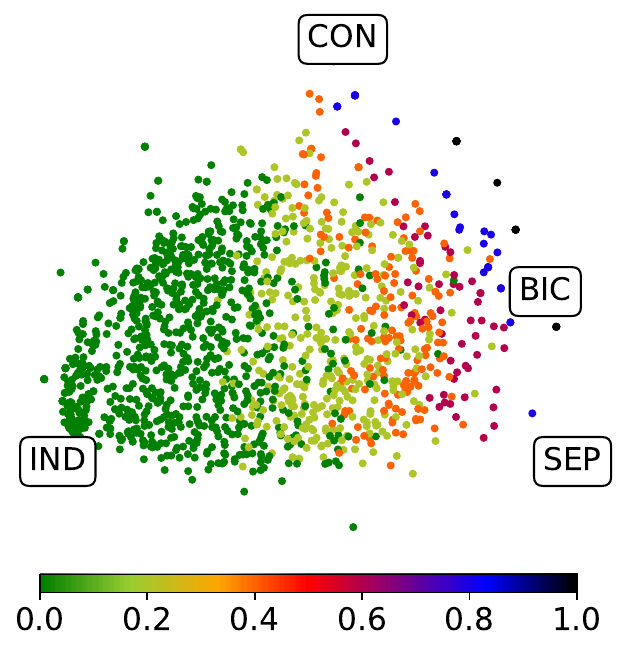}
		\caption{$5\times 5$, valuation distance}
	\end{subfigure}\hfill%
	\begin{subfigure}[t]{\firsttwo\linewidth}
		\centering
		\includegraphics[width=\linewidth]{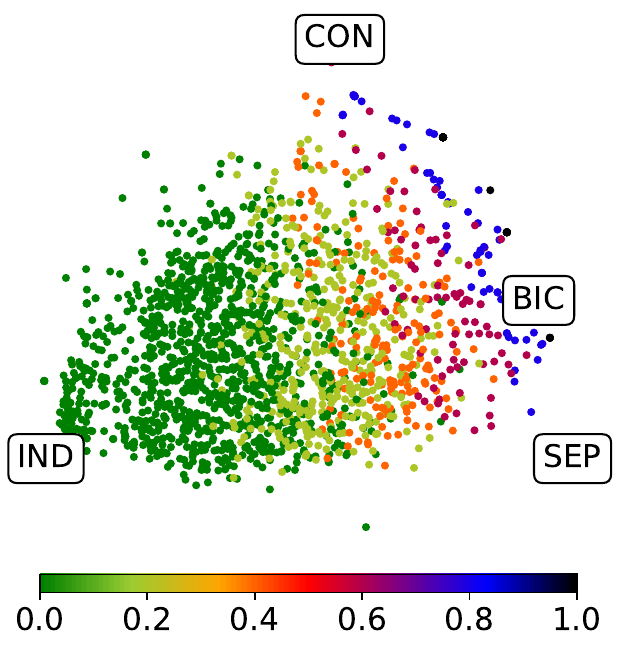}
		\caption{$5\times 5$, demand distance}
	\end{subfigure}\hfill%
	\begin{subfigure}[t]{\fpeval{(0.96 - \firsttwo - \firsttwo)*\linewidth}pt}
		\centering
		\includegraphics[width=\linewidth]{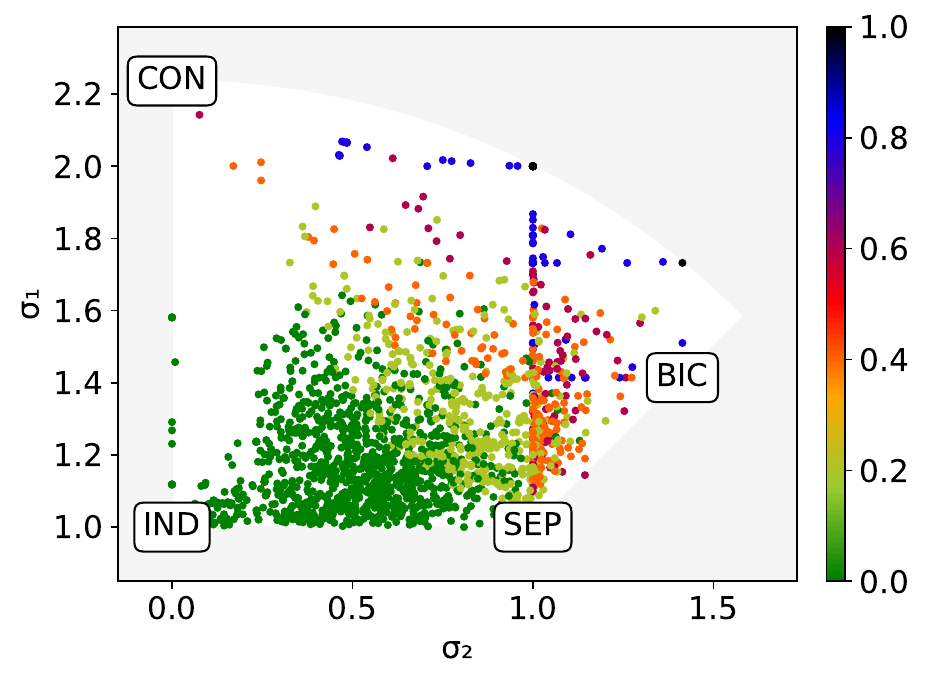}
		\caption{$5\times 5$}
	\end{subfigure}\\
	\begin{subfigure}[t]{\firsttwo\linewidth}
		\quad
	\end{subfigure}\hfill%
	\begin{subfigure}[t]{\firsttwo\linewidth}
		\centering
		\includegraphics[width=\linewidth]{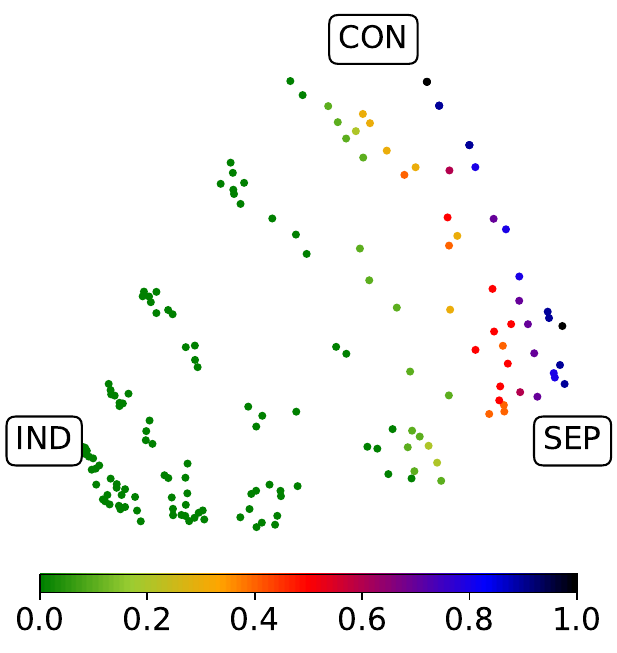}
		\caption{$10\times 20$, demand distance}
	\end{subfigure}\hfill%
	\begin{subfigure}[t]{\fpeval{(0.96 - \firsttwo - \firsttwo)*\linewidth}pt}
		\centering
		\includegraphics[width=\linewidth]{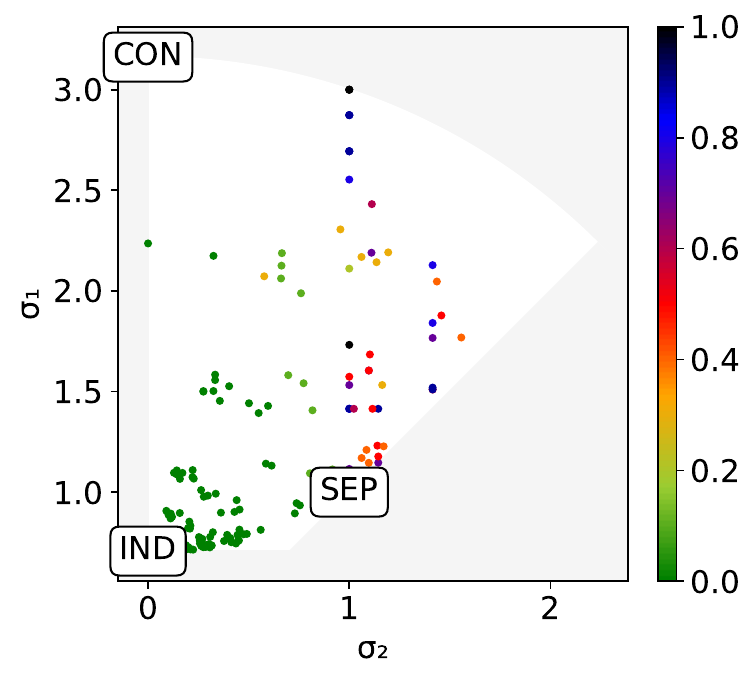}
		\caption{$10\times 20$}
	\end{subfigure}\\
	\caption{Distribution of the agents who are single-minded on our
	distance-embedding map using the valuation distance (first column), using the
demand distance (second column), and on our explicit map (third column).
Computing the valuation distance for instances of the~$10 \times 20$~dataset was
computationally too demanding, hence the blank space in the third row.}
  \label{fig:all:fracsingleminded} 
	\let\firsttwo\undefined
\end{figure*}

\begin{figure*}\centering
	\def\firsttwo{.29}
	\begin{subfigure}[t]{\firsttwo\linewidth}
		\centering
		\includegraphics[width=\linewidth]{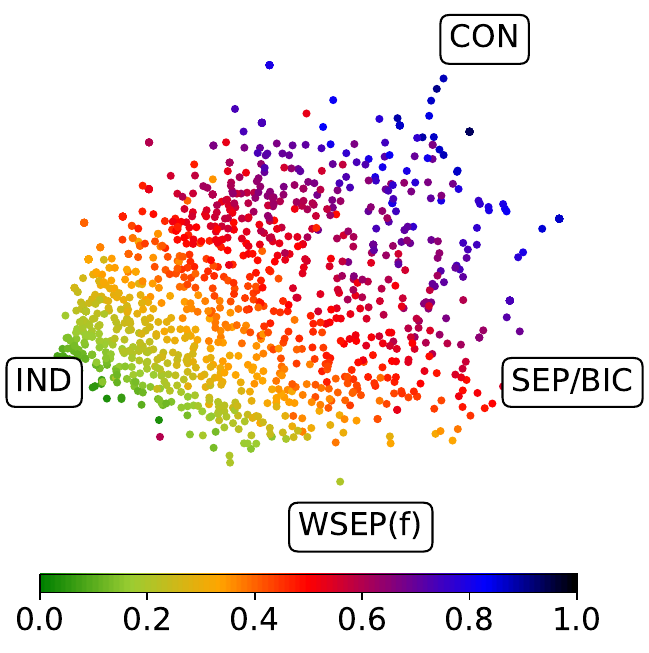}
		\caption{$3\times 6$, valuation distance}
	\end{subfigure}\hfill%
	\begin{subfigure}[t]{\firsttwo\linewidth}
		\centering
		\includegraphics[width=\linewidth]{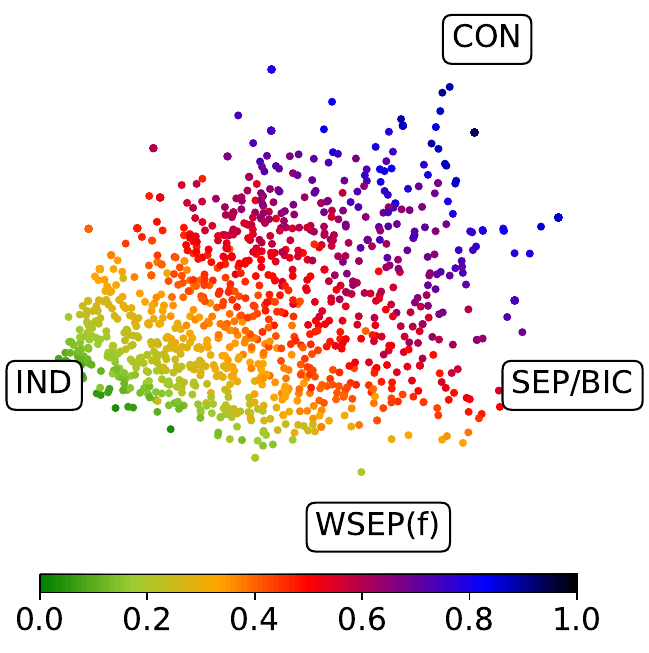}
		\caption{$3\times 6$, demand distance}
	\end{subfigure}\hfill%
	\begin{subfigure}[t]{\fpeval{(0.96 - \firsttwo - \firsttwo)*\linewidth}pt}
		\centering
		\includegraphics[width=\linewidth]{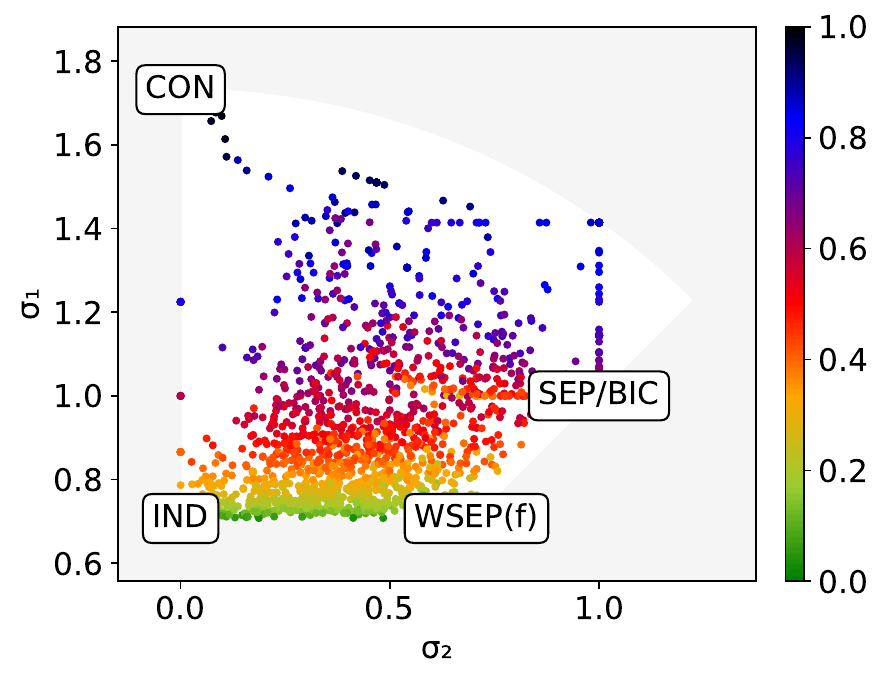}
		\caption{$3\times 6$}
	\end{subfigure}\\
	\begin{subfigure}[t]{\firsttwo\linewidth}
		\centering
		\includegraphics[width=\linewidth]{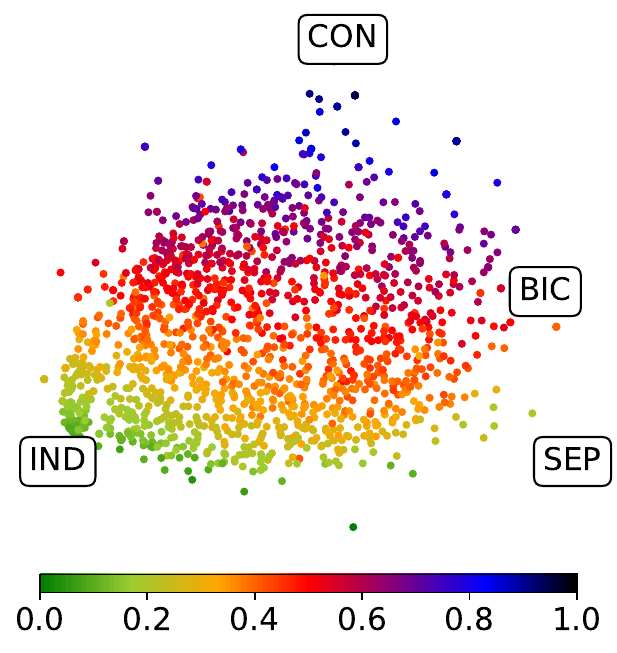}
		\caption{$5\times 5$, valuation distance}
	\end{subfigure}\hfill%
	\begin{subfigure}[t]{\firsttwo\linewidth}
		\centering
		\includegraphics[width=\linewidth]{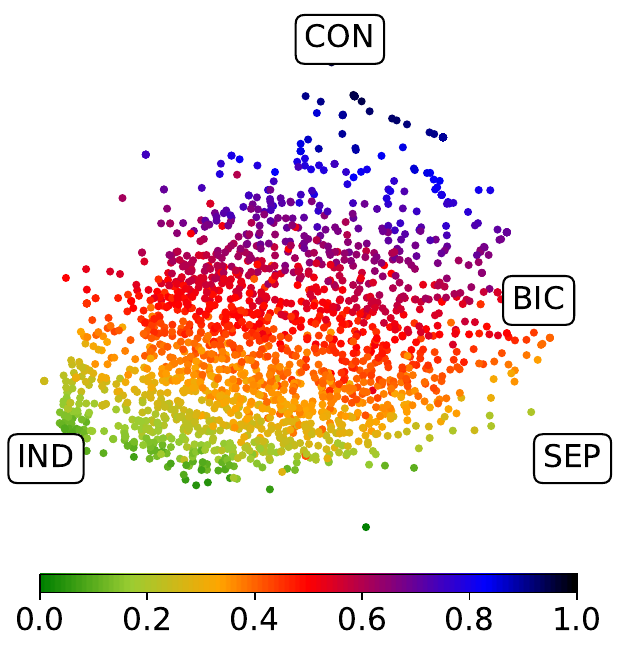}
		\caption{$5\times 5$, demand distance}
	\end{subfigure}\hfill%
	\begin{subfigure}[t]{\fpeval{(0.96 - \firsttwo - \firsttwo)*\linewidth}pt}
		\centering
		\includegraphics[width=\linewidth]{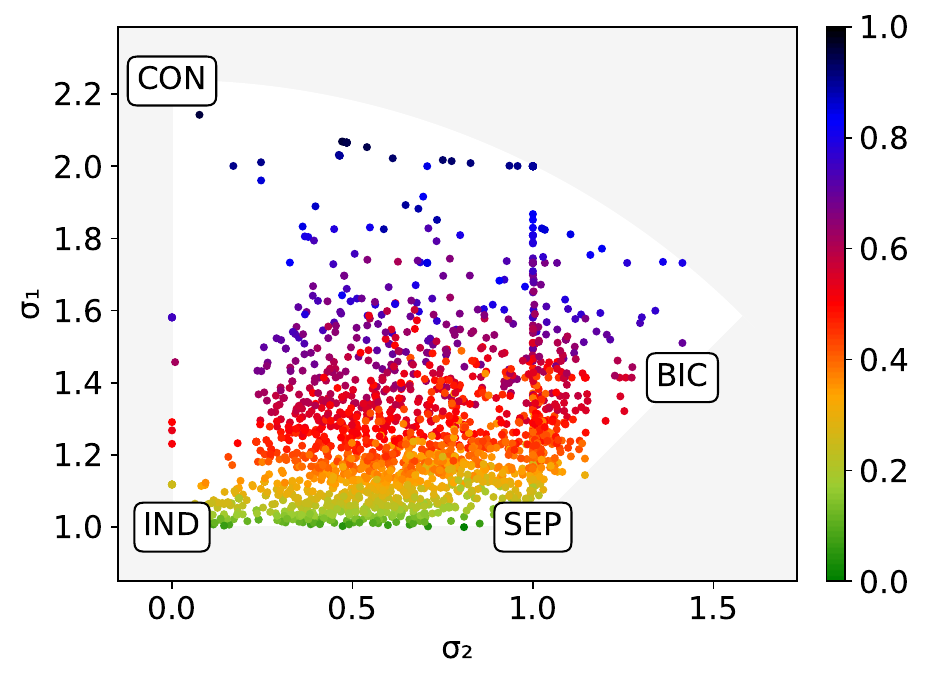}
		\caption{$5\times 5$}
	\end{subfigure}\\
	\begin{subfigure}[t]{\firsttwo\linewidth}
		\quad
	\end{subfigure}\hfill%
	\begin{subfigure}[t]{\firsttwo\linewidth}
		\centering
		\includegraphics[width=\linewidth]{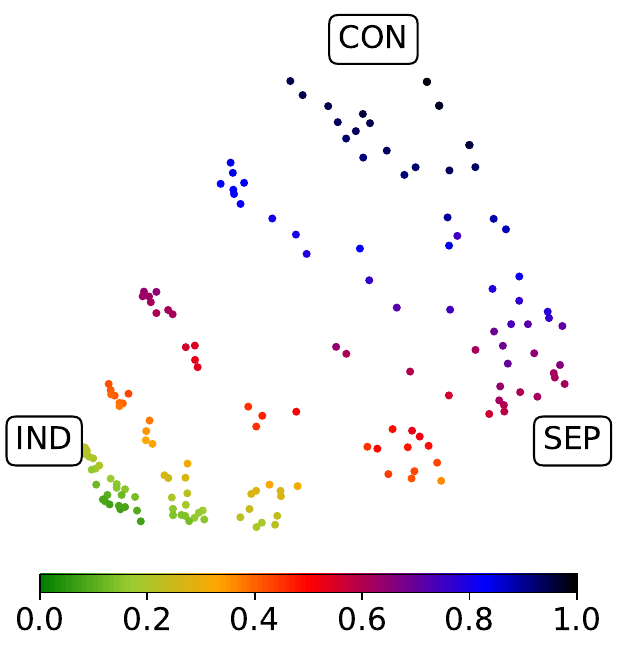}
		\caption{$10\times 20$, demand distance}
	\end{subfigure}\hfill%
	\begin{subfigure}[t]{\fpeval{(0.96 - \firsttwo - \firsttwo)*\linewidth}pt}
		\centering
		\includegraphics[width=\linewidth]{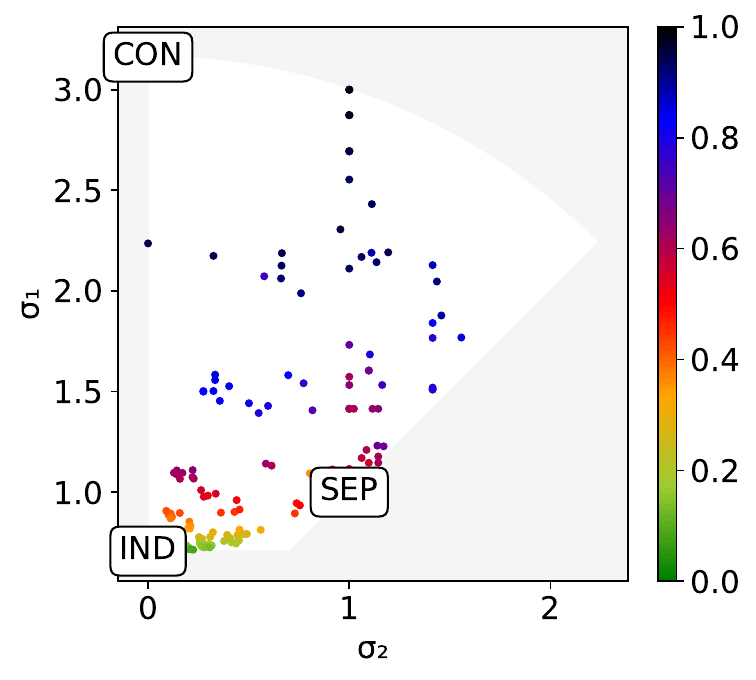}
		\caption{$10\times 20$}
	\end{subfigure}\\
	\caption{Distribution of the diversity of demand on our distance-embedding map
	using the valuation distance (first column), using the demand distance (second
column), and on our explicit map (third column). Computing the valuation
distance for instances of the~$10 \times 20$~dataset was computationally too
demanding, hence the blank space in the third row.}
  \label{fig:all:divdemand} 
	\let\firsttwo\undefined
\end{figure*}

\begin{figure*}\centering
	\def\firsttwo{.29}
	\begin{subfigure}[t]{\firsttwo\linewidth}
		\centering
		\includegraphics[width=\linewidth]{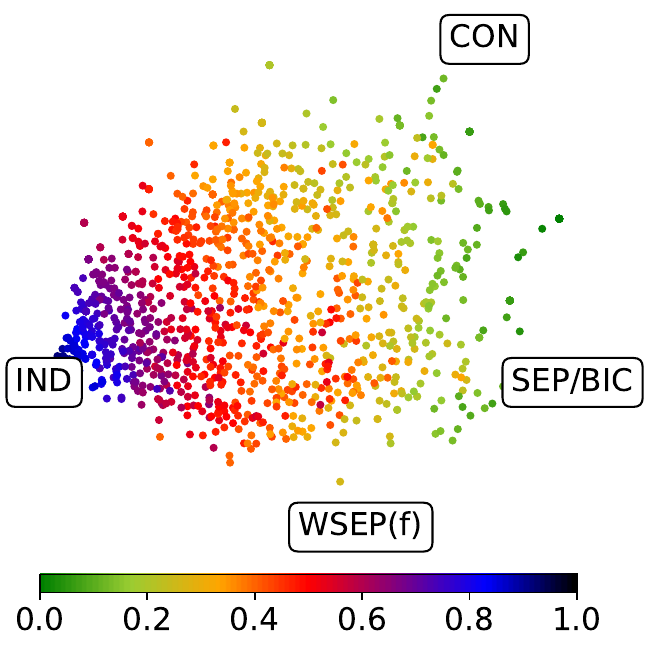}
		\caption{$3\times 6$, valuation distance}
	\end{subfigure}\hfill%
	\begin{subfigure}[t]{\firsttwo\linewidth}
		\centering
		\includegraphics[width=\linewidth]{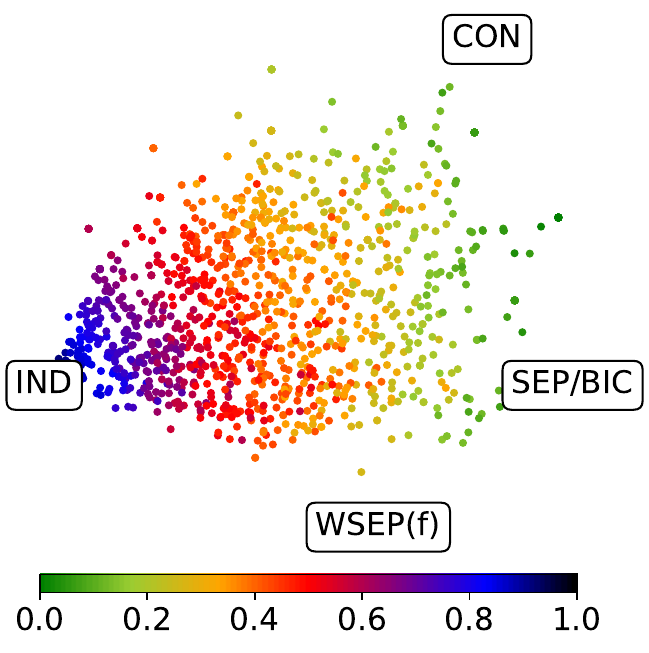}
		\caption{$3\times 6$, demand distance}
	\end{subfigure}\hfill%
	\begin{subfigure}[t]{\fpeval{(0.96 - \firsttwo - \firsttwo)*\linewidth}pt}
		\centering
		\includegraphics[width=\linewidth]{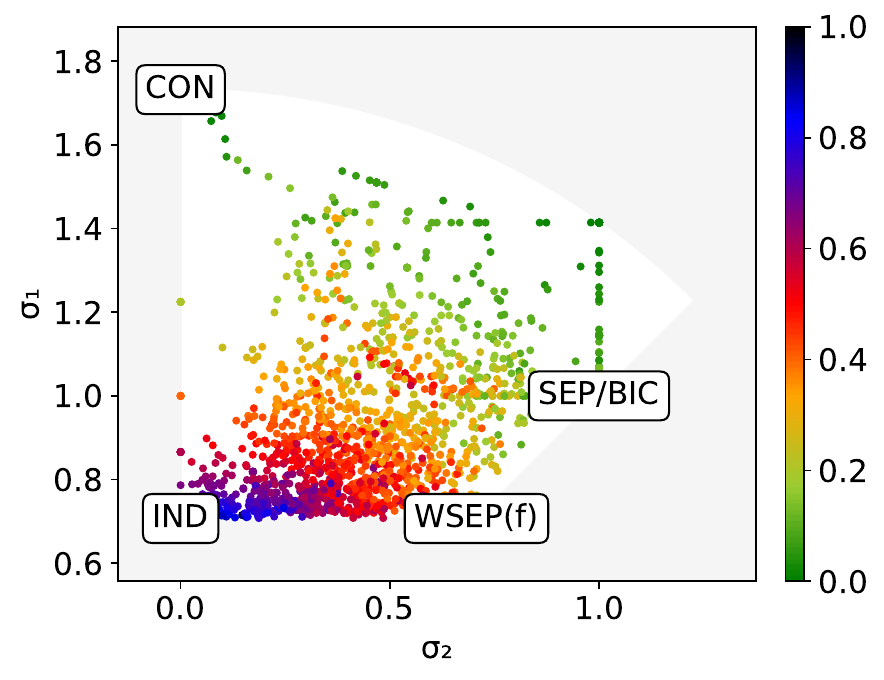}
		\caption{$3\times 6$}
	\end{subfigure}\\
	\begin{subfigure}[t]{\firsttwo\linewidth}
		\centering
		\includegraphics[width=\linewidth]{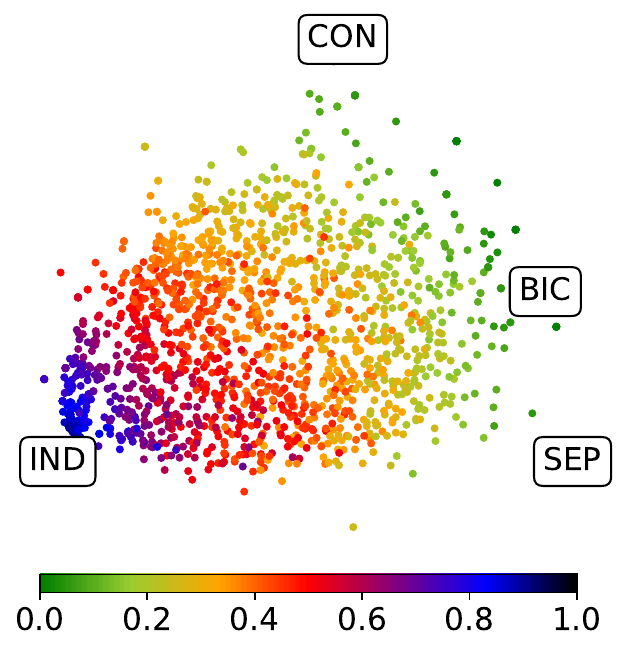}
		\caption{$5\times 5$, valuation distance}
	\end{subfigure}\hfill%
	\begin{subfigure}[t]{\firsttwo\linewidth}
		\centering
		\includegraphics[width=\linewidth]{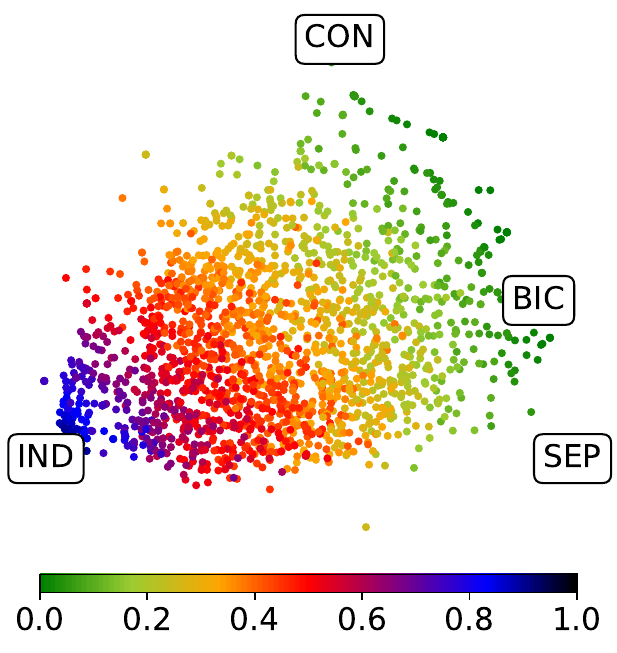}
		\caption{$5\times 5$, demand distance}
	\end{subfigure}\hfill%
	\begin{subfigure}[t]{\fpeval{(0.96 - \firsttwo - \firsttwo)*\linewidth}pt}
		\centering
		\includegraphics[width=\linewidth]{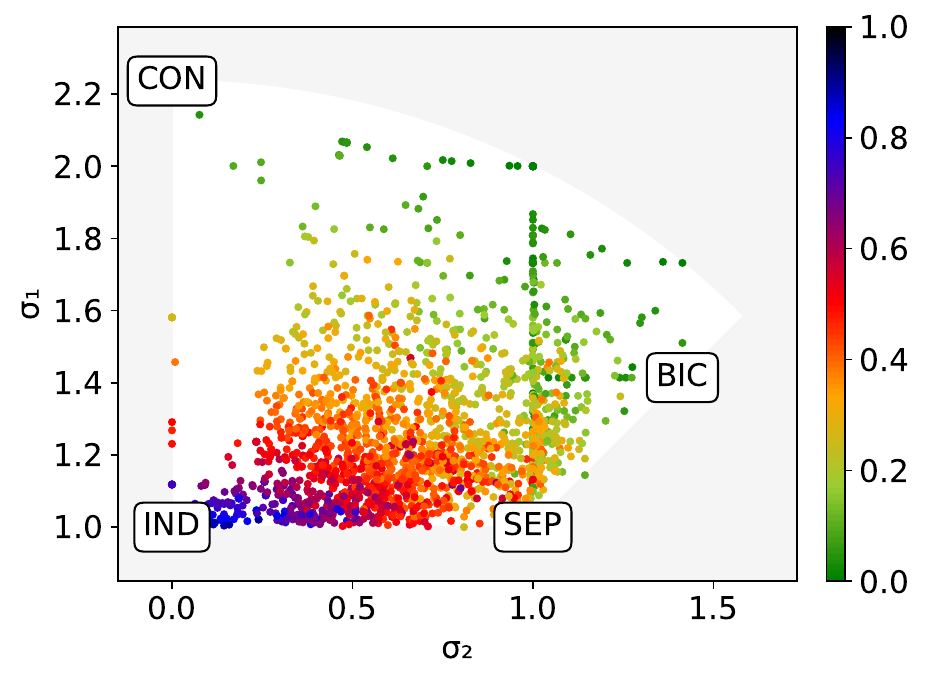}
		\caption{$5\times 5$}
	\end{subfigure}\\
	\begin{subfigure}[t]{\firsttwo\linewidth}
		\quad
	\end{subfigure}\hfill%
	\begin{subfigure}[t]{\firsttwo\linewidth}
		\centering
		\includegraphics[width=\linewidth]{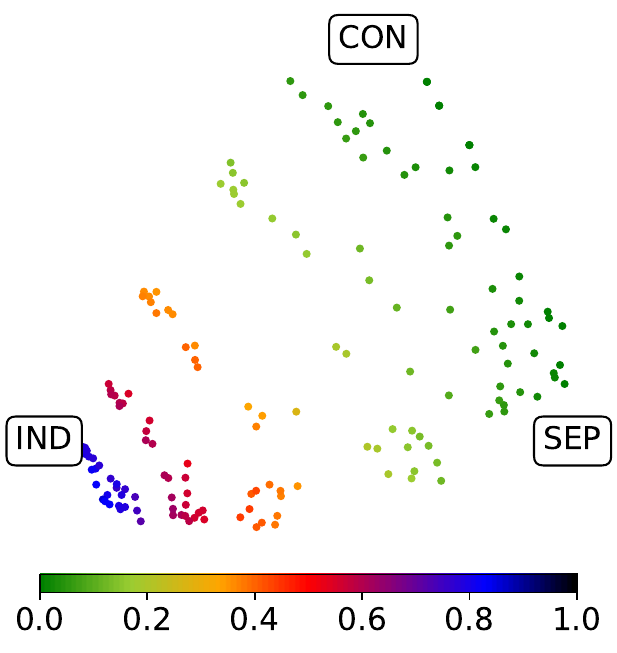}
		\caption{$10\times 20$, demand distance}
	\end{subfigure}\hfill%
	\begin{subfigure}[t]{\fpeval{(0.96 - \firsttwo - \firsttwo)*\linewidth}pt}
		\centering
		\includegraphics[width=\linewidth]{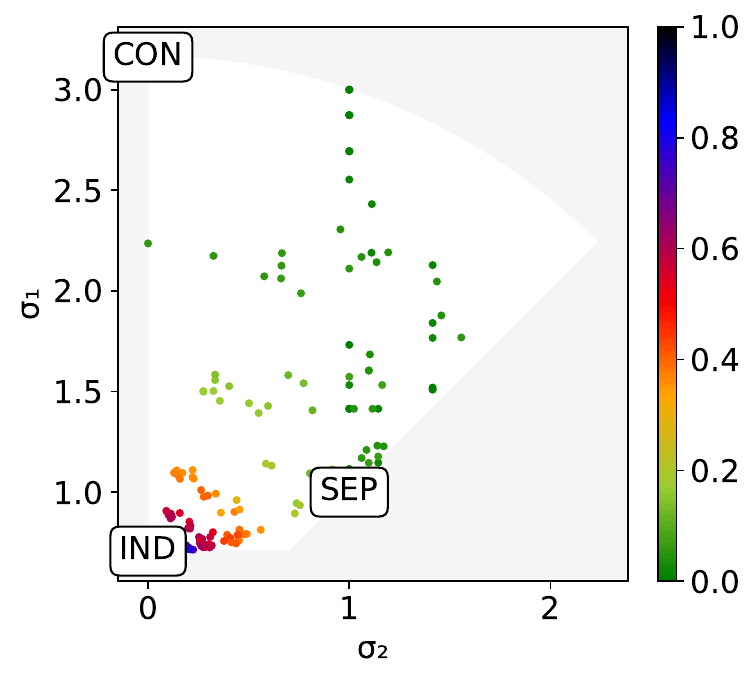}
		\caption{$10\times 20$}
	\end{subfigure}\\
	\caption{Distribution of one minus pickiness on our distance-embedding map
	using the valuation distance (first column), using the demand distance (second
column), and on our explicit map (third column). Computing the valuation
distance for instances of the~$10 \times 20$~dataset was computationally too
demanding, hence the blank space in the third row.}
  \label{fig:all:minuspicki} 
	\let\firsttwo\undefined
\end{figure*}

Lastly, we highlight each of the different instance sources separately for
both the $3\times 6$ and $5\times 5$ instances in
Fig.~\ref{fig:exponstri} to \ref{fig:candiescompar}---maps of the $10\times 20$ instances are missing here, as each instance apart from the corner points are from the resampling distribution---, which show the
observations about the instance sources on the distance-embedding map in
\cref{subsec:stude} more clearly on both distance-embedding and explicit
maps, both of which show similar distributions.

\begin{figure*}\centering
	\begin{subfigure}[t]{.32\linewidth}
		\centering
		\includegraphics[height=4cm]{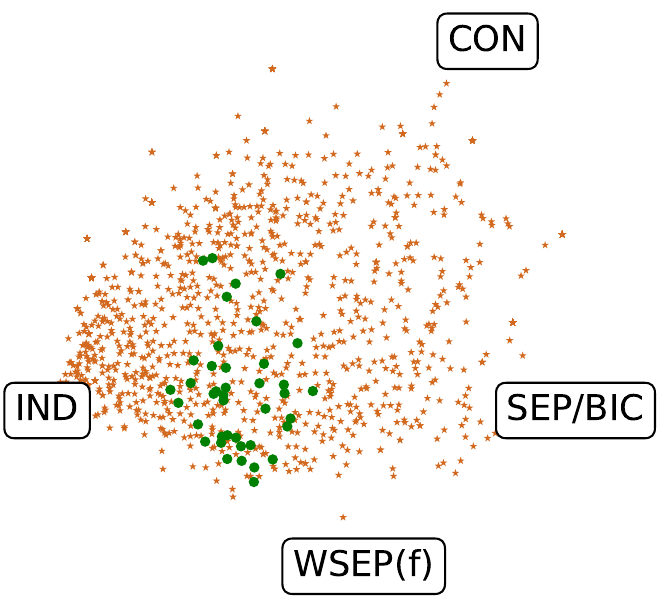}
		\caption{$3\times 6$, valuation distance}
	\end{subfigure}%
	\begin{subfigure}[t]{.32\linewidth}
		\centering
		\includegraphics[height=4cm]{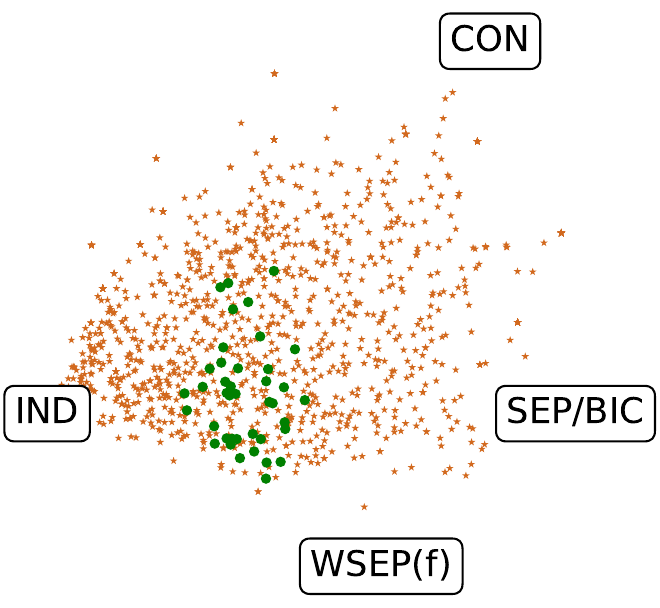}
		\caption{$3\times 6$, demand distance}
	\end{subfigure}%
	\begin{subfigure}[t]{.33\linewidth}
		\centering
		\includegraphics[height=4cm]{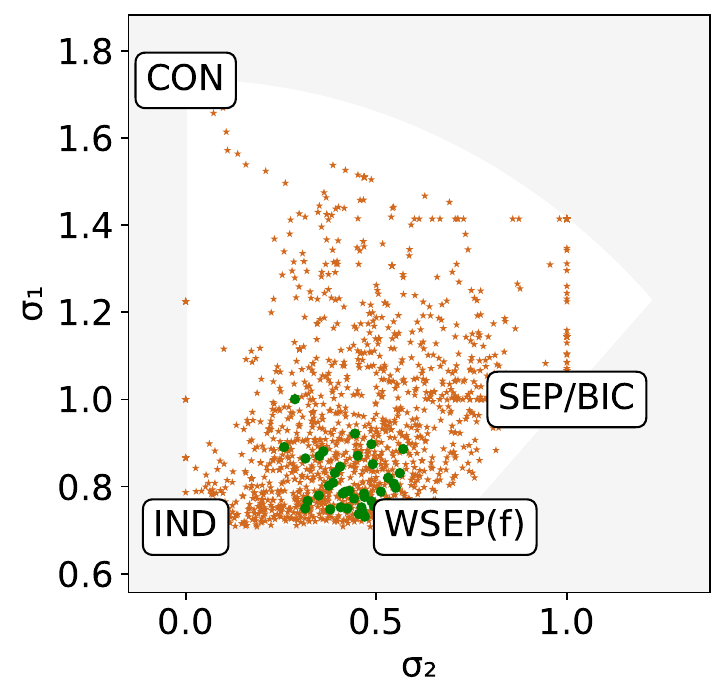}
		\caption{$3\times 6$}
	\end{subfigure}\\
	\begin{subfigure}[t]{.32\linewidth}
		\centering
		\includegraphics[height=4cm]{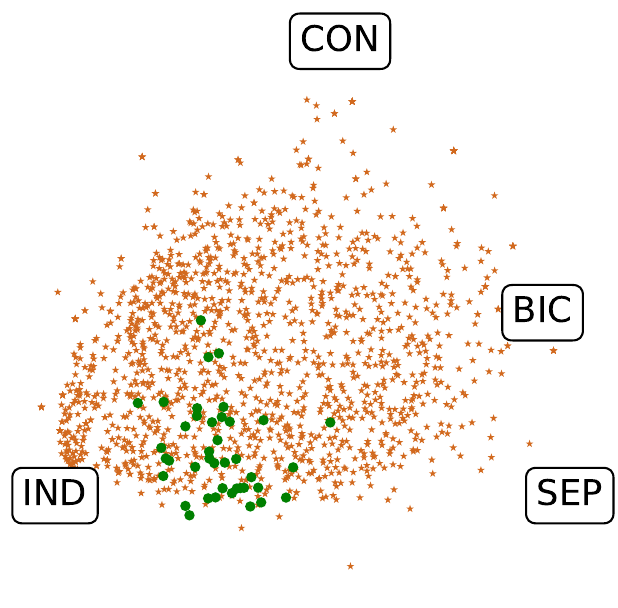}
		\caption{$5\times 5$, valuation distance}
	\end{subfigure}%
	\begin{subfigure}[t]{.32\linewidth}
		\centering
		\includegraphics[height=4cm]{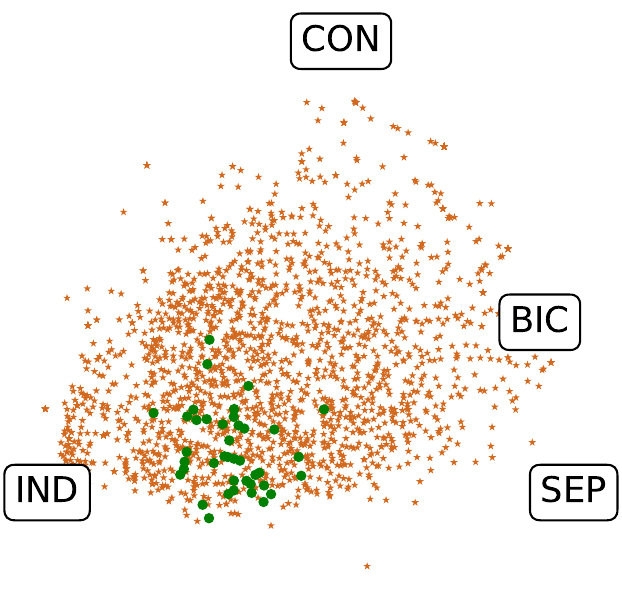}
		\caption{$5\times 5$, demand distance}
	\end{subfigure}%
	\begin{subfigure}[t]{.33\linewidth}
		\centering
		\includegraphics[height=4cm]{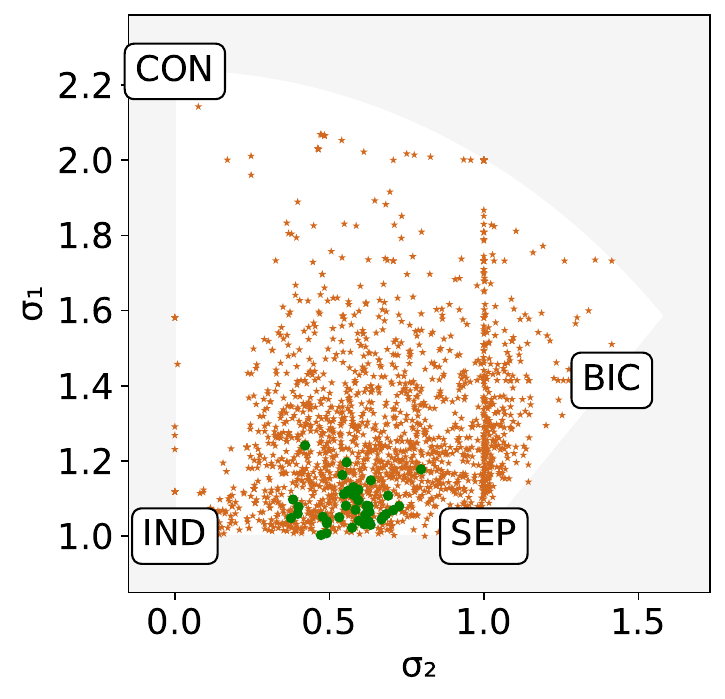}
		\caption{$5\times 5$}
	\end{subfigure}\\
	\caption{The instances from the i.i.d. distribution using the exponential distribution are marked as green dots, while all other instances are marked as brown stars.}
	\label{fig:exponstri}
\end{figure*}

\begin{figure*}\centering
	\begin{subfigure}[t]{.32\linewidth}
		\centering
		\includegraphics[height=4cm]{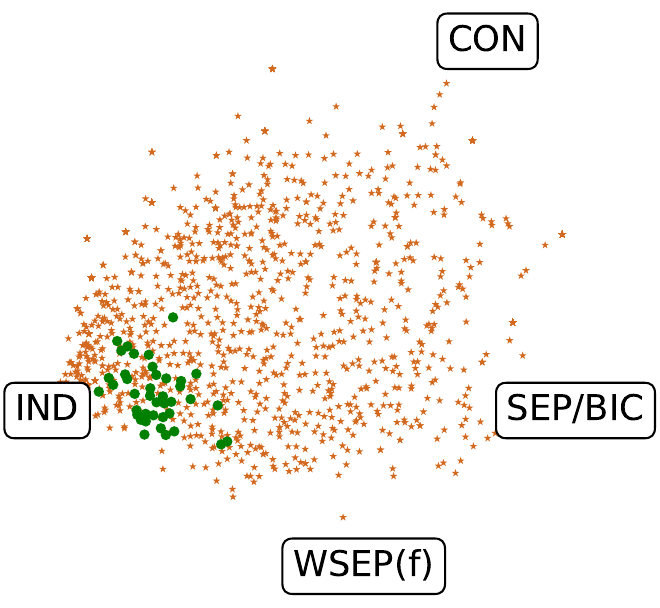}
		\caption{$3\times 6$, valuation distance}
	\end{subfigure}%
	\begin{subfigure}[t]{.32\linewidth}
		\centering
		\includegraphics[height=4cm]{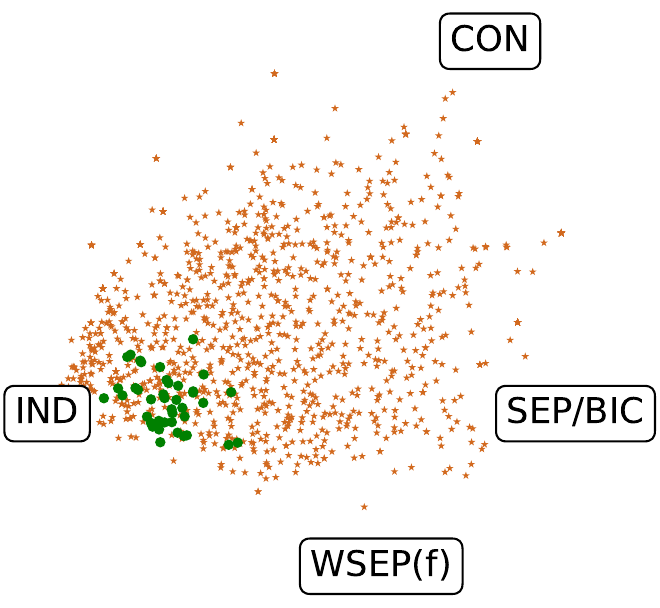}
		\caption{$3\times 6$, demand distance}
	\end{subfigure}%
	\begin{subfigure}[t]{.33\linewidth}
		\centering
		\includegraphics[height=4cm]{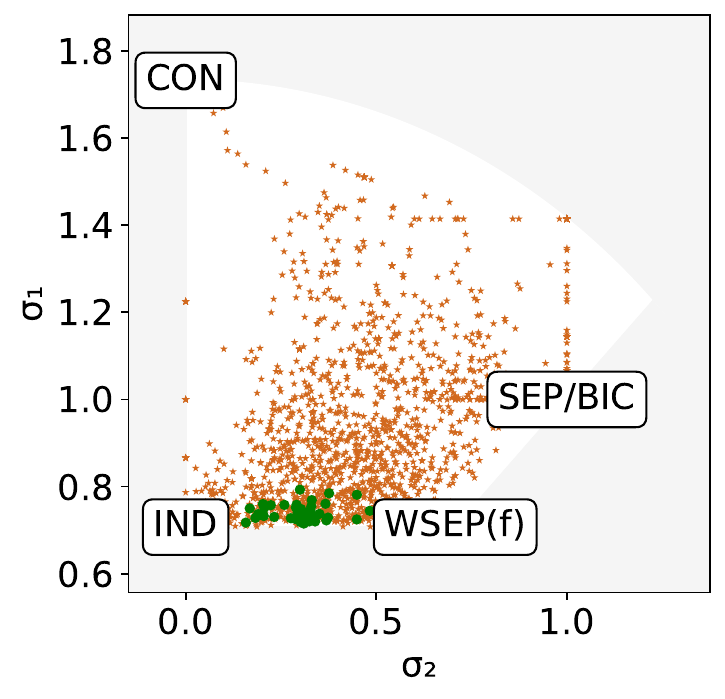}
		\caption{$3\times 6$}
	\end{subfigure}\\
	\begin{subfigure}[t]{.32\linewidth}
		\centering
		\includegraphics[height=4cm]{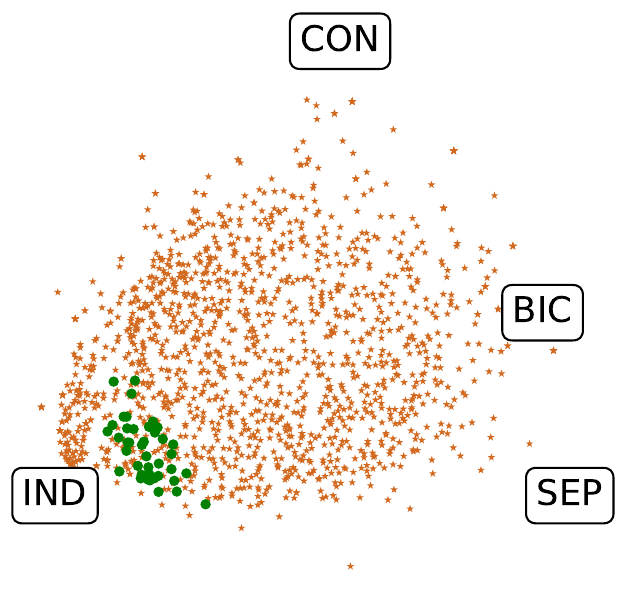}
		\caption{$5\times 5$, valuation distance}
	\end{subfigure}%
	\begin{subfigure}[t]{.32\linewidth}
		\centering
		\includegraphics[height=4cm]{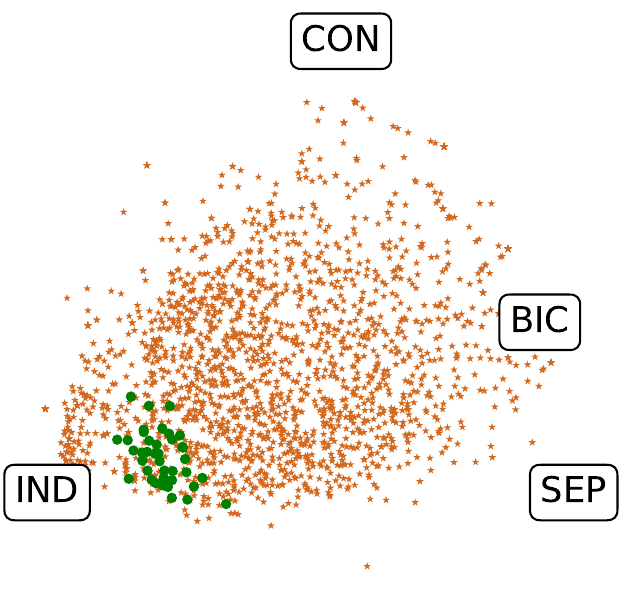}
		\caption{$5\times 5$, demand distance}
	\end{subfigure}%
	\begin{subfigure}[t]{.33\linewidth}
		\centering
		\includegraphics[height=4cm]{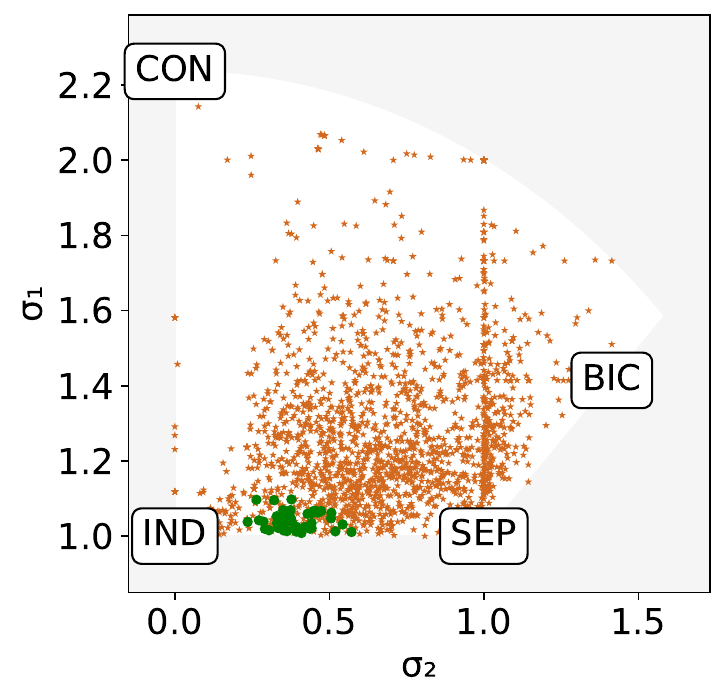}
		\caption{$5\times 5$}
	\end{subfigure}\\
	\caption{The instances from the i.i.d. distribution using the uniform distribution over $[0,1]$ are marked as green dots, while all other instances are marked as brown stars.}
	\label{fig:unistri}
\end{figure*}

\begin{figure*}\centering
	\begin{subfigure}[t]{.32\linewidth}
		\centering
		\includegraphics[height=4cm]{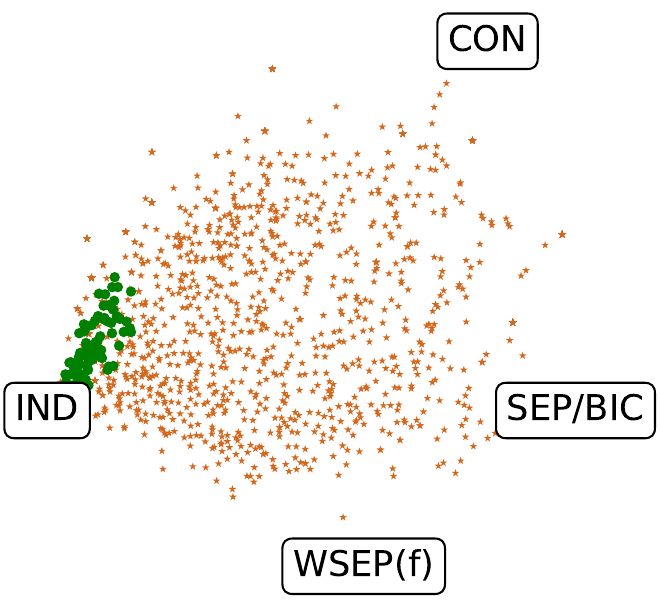}
		\caption{$3\times 6$, valuation distance}
	\end{subfigure}%
	\begin{subfigure}[t]{.32\linewidth}
		\centering
		\includegraphics[height=4cm]{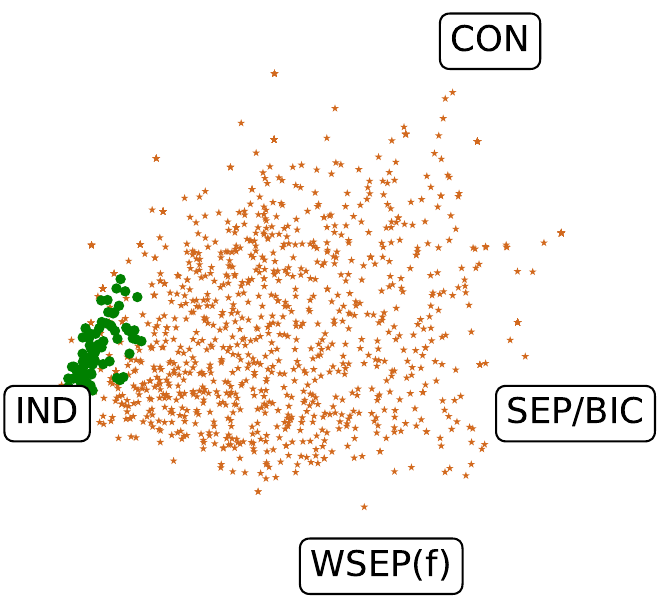}
		\caption{$3\times 6$, demand distance}
	\end{subfigure}%
	\begin{subfigure}[t]{.33\linewidth}
		\centering
		\includegraphics[height=4cm]{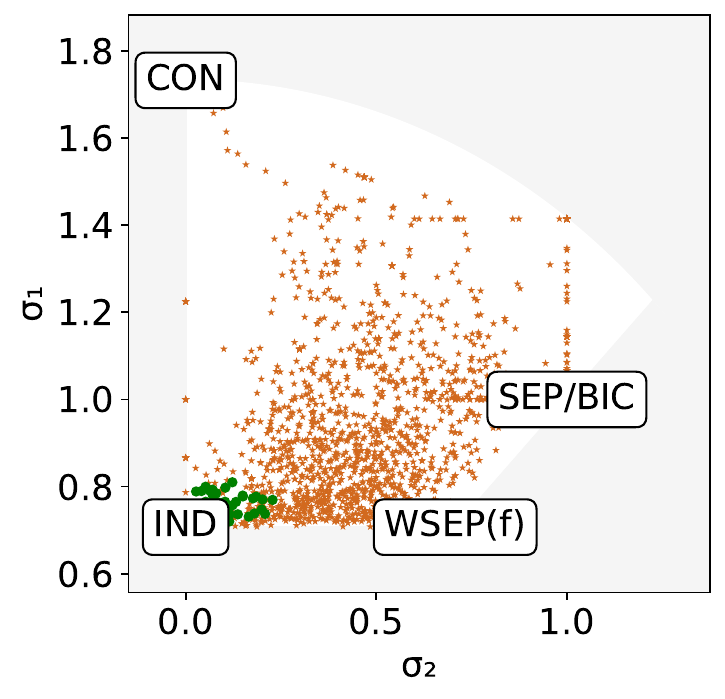}
		\caption{$3\times 6$}
	\end{subfigure}\\
	\begin{subfigure}[t]{.32\linewidth}
		\centering
		\includegraphics[height=4cm]{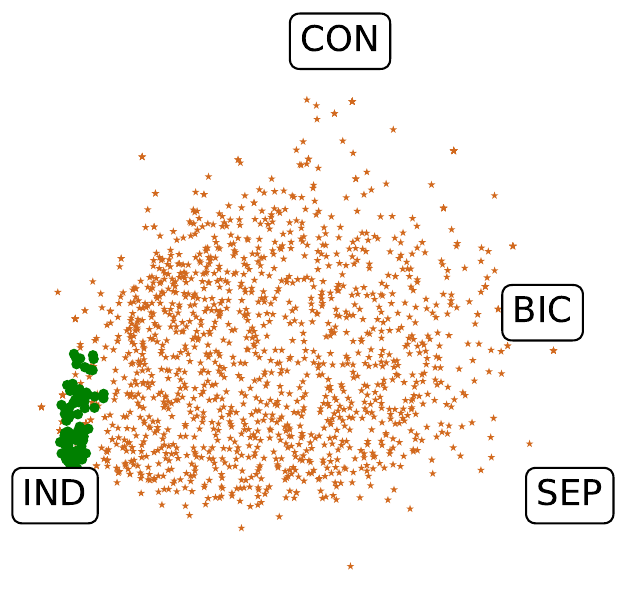}
		\caption{$5\times 5$, valuation distance}
	\end{subfigure}%
	\begin{subfigure}[t]{.32\linewidth}
		\centering
		\includegraphics[height=4cm]{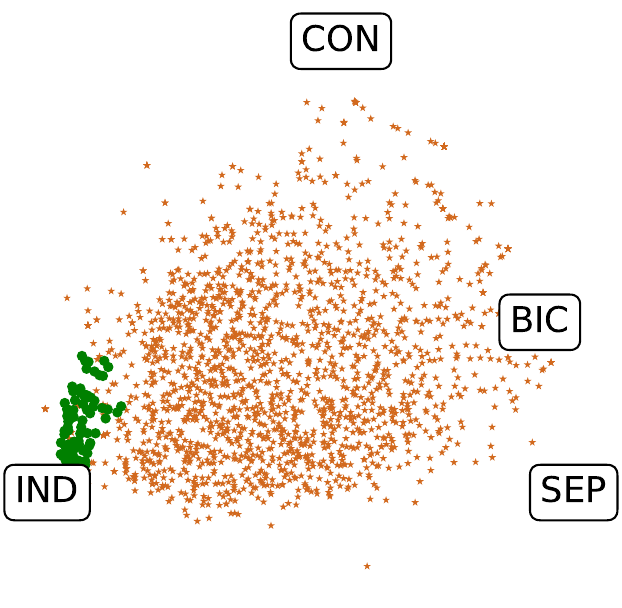}
		\caption{$5\times 5$, demand distance}
	\end{subfigure}%
	\begin{subfigure}[t]{.33\linewidth}
		\centering
		\includegraphics[height=4cm]{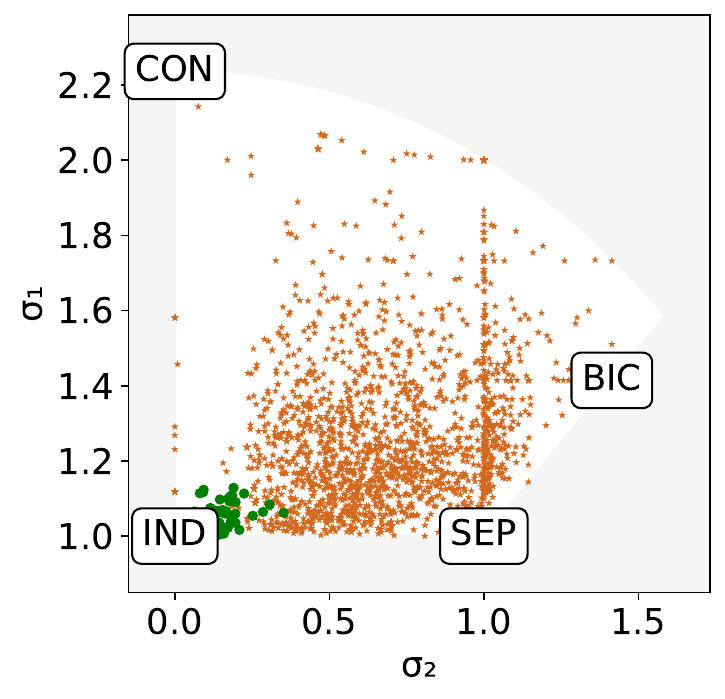}
		\caption{$5\times 5$}
	\end{subfigure}\\
	\caption{The instances from the attributes distribution are marked as green dots, while all other instances are marked as brown stars.}
	\label{fig:attristri}
\end{figure*}

\begin{figure*}\centering
	\begin{subfigure}[t]{.32\linewidth}
		\centering
		\includegraphics[height=4cm]{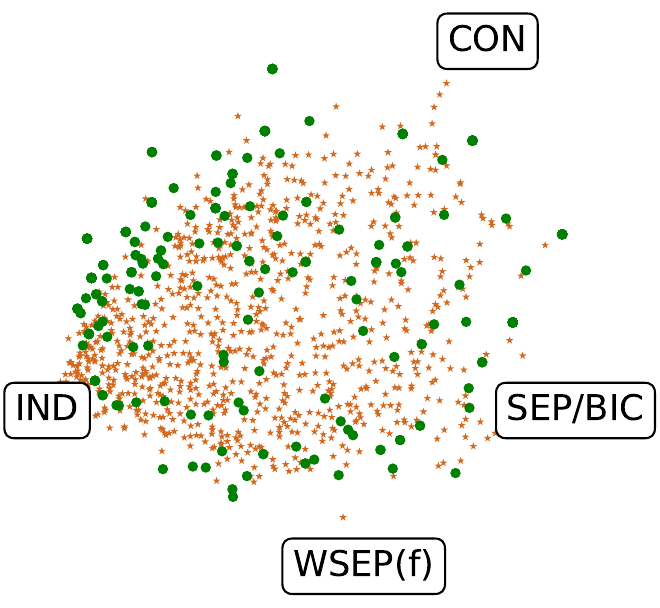}
		\caption{$3\times 6$, valuation distance}
	\end{subfigure}%
	\begin{subfigure}[t]{.32\linewidth}
		\centering
		\includegraphics[height=4cm]{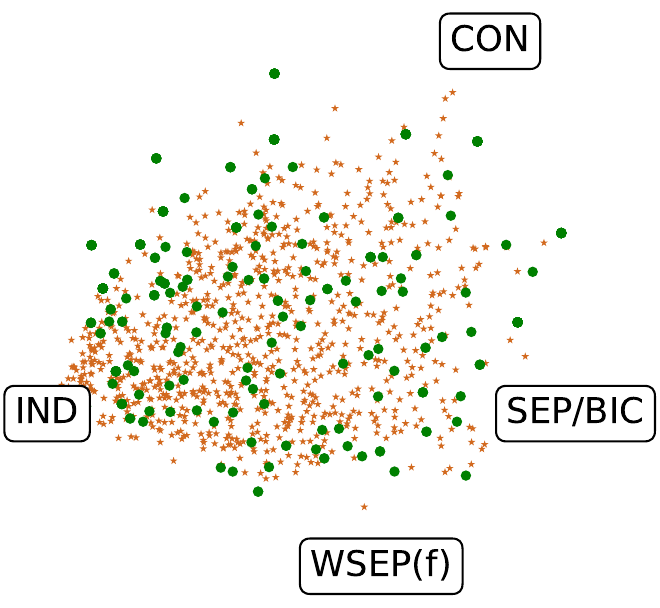}
		\caption{$3\times 6$, demand distance}
	\end{subfigure}%
	\begin{subfigure}[t]{.33\linewidth}
		\centering
		\includegraphics[height=4cm]{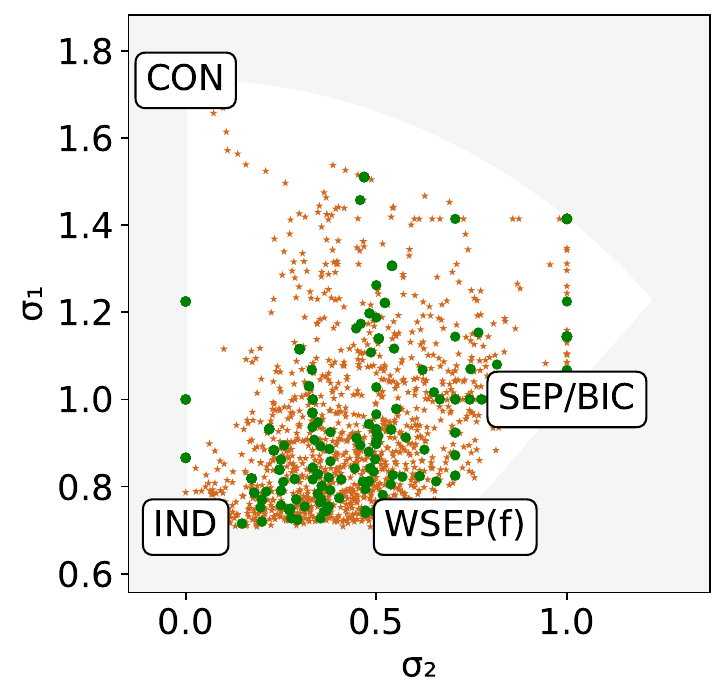}
		\caption{$3\times 6$}
	\end{subfigure}\\
	\begin{subfigure}[t]{.32\linewidth}
		\centering
		\includegraphics[height=4cm]{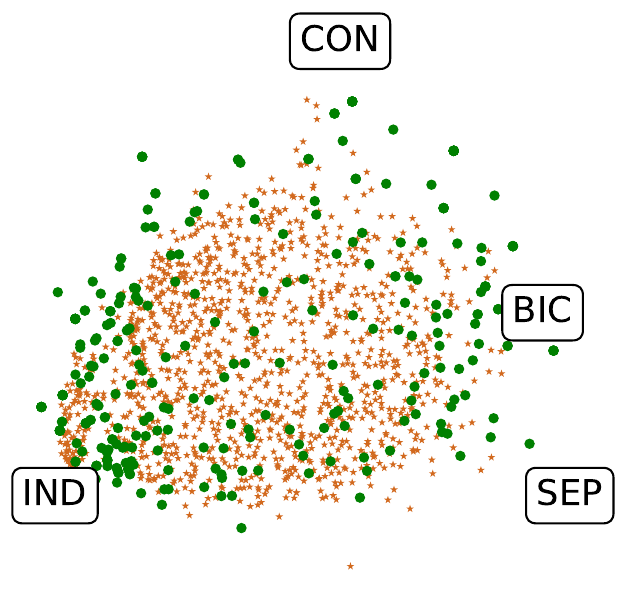}
		\caption{$5\times 5$, valuation distance}
	\end{subfigure}%
	\begin{subfigure}[t]{.32\linewidth}
		\centering
		\includegraphics[height=4cm]{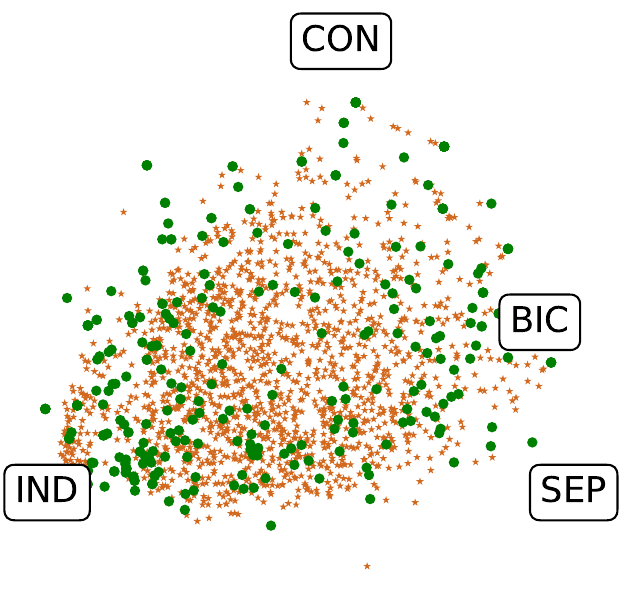}
		\caption{$5\times 5$, demand distance}
	\end{subfigure}%
	\begin{subfigure}[t]{.33\linewidth}
		\centering
		\includegraphics[height=4cm]{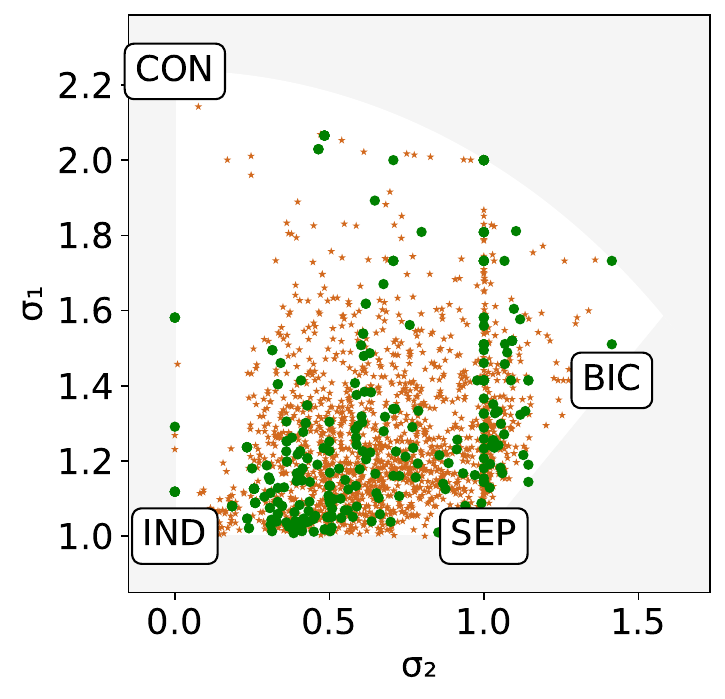}
		\caption{$5\times 5$}
	\end{subfigure}\\
	\caption{The instances from the resampling distribution are marked as green dots, while all other instances are marked as brown stars.}
	\label{fig:rsmplstri}
\end{figure*}

\begin{figure*}\centering
	\begin{subfigure}[t]{.32\linewidth}
		\centering
		\includegraphics[height=4cm]{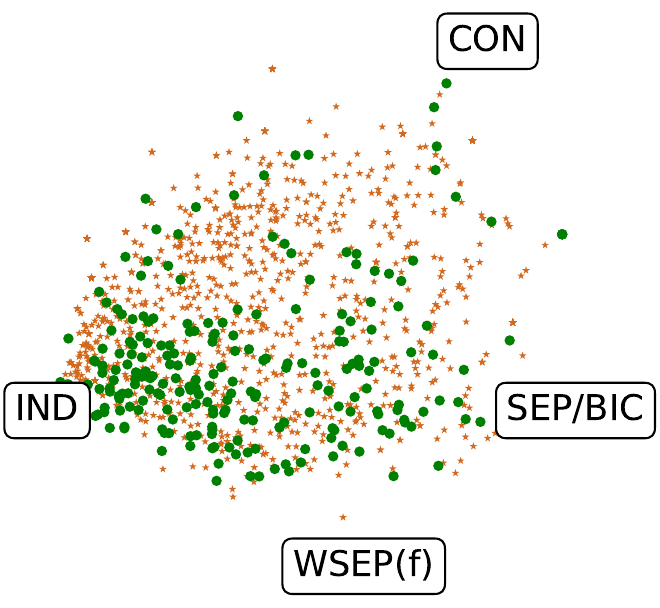}
		\caption{$3\times 6$, valuation distance}
	\end{subfigure}%
	\begin{subfigure}[t]{.32\linewidth}
		\centering
		\includegraphics[height=4cm]{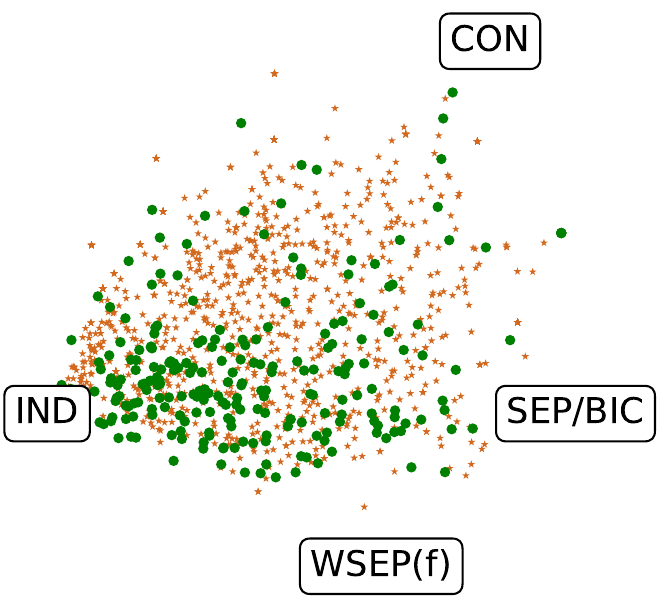}
		\caption{$3\times 6$, demand distance}
	\end{subfigure}%
	\begin{subfigure}[t]{.33\linewidth}
		\centering
		\includegraphics[height=4cm]{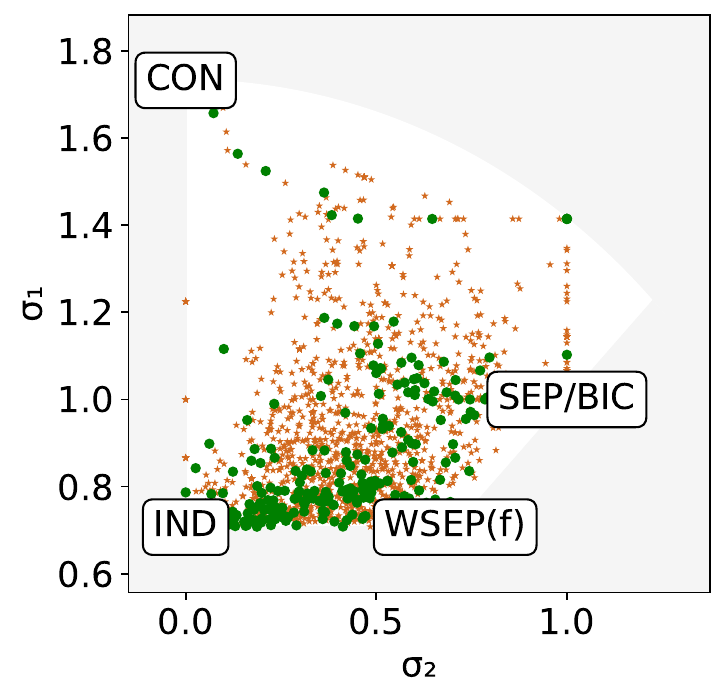}
		\caption{$3\times 6$}
	\end{subfigure}\\
	\begin{subfigure}[t]{.32\linewidth}
		\centering
		\includegraphics[height=4cm]{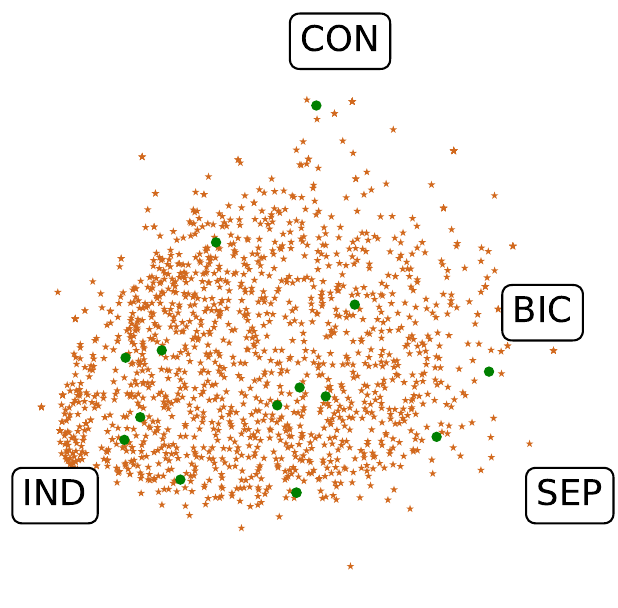}
		\caption{$5\times 5$, valuation distance}
	\end{subfigure}%
	\begin{subfigure}[t]{.32\linewidth}
		\centering
		\includegraphics[height=4cm]{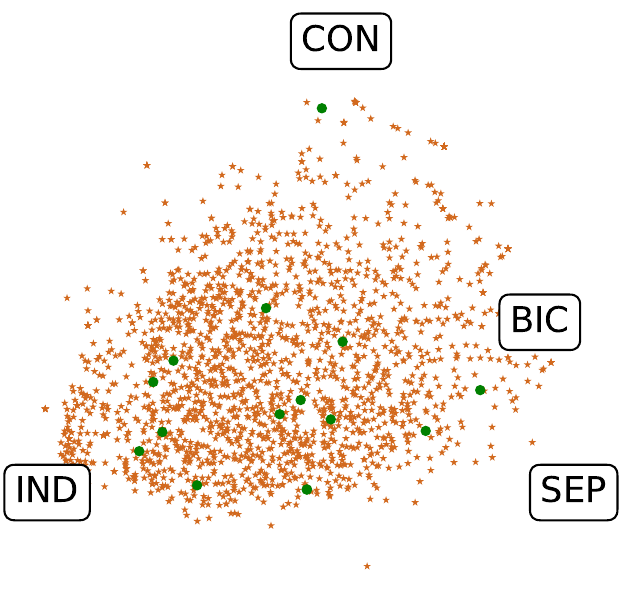}
		\caption{$5\times 5$, demand distance}
	\end{subfigure}%
	\begin{subfigure}[t]{.33\linewidth}
		\centering
		\includegraphics[height=4cm]{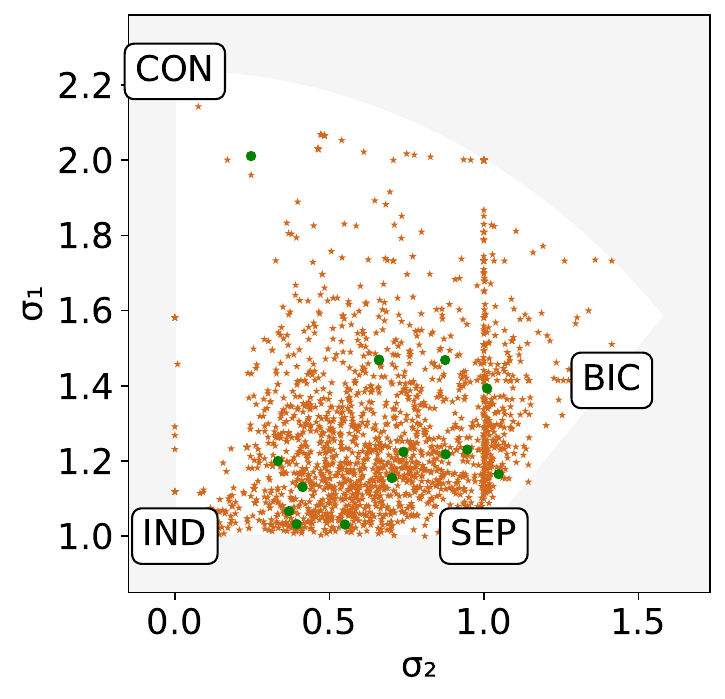}
		\caption{$5\times 5$}
	\end{subfigure}\\
	\caption{The instances from the Spliddit dataset are marked as green dots, while all other instances are marked as brown stars}
	\label{fig:spldstri}
\end{figure*}

\begin{figure*}\centering
	\begin{subfigure}[t]{.32\linewidth}
		\centering
		\includegraphics[height=4cm]{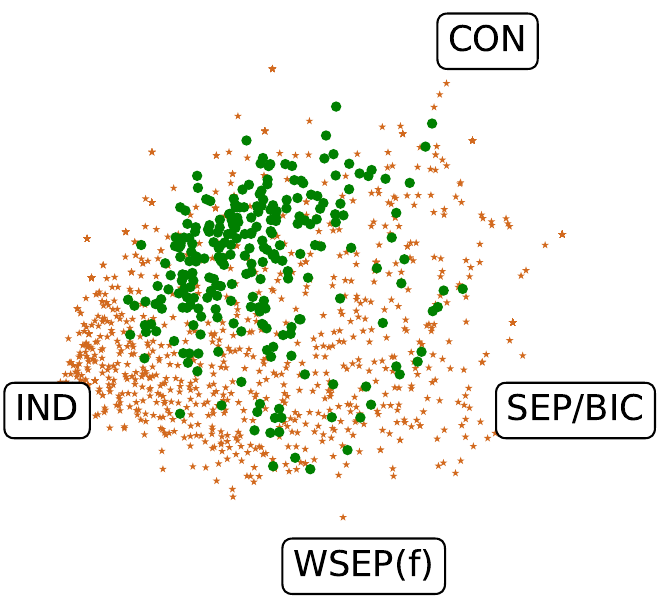}
		\caption{$3\times 6$, valuation distance}
	\end{subfigure}%
	\begin{subfigure}[t]{.32\linewidth}
		\centering
		\includegraphics[height=4cm]{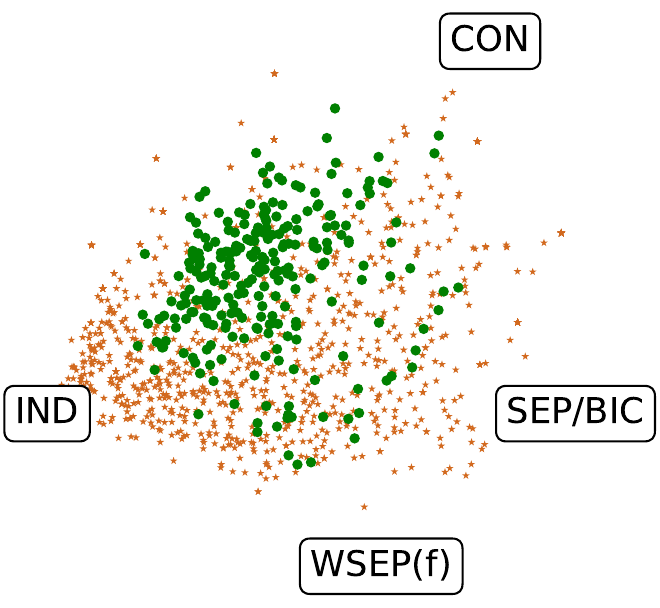}
		\caption{$3\times 6$, demand distance}
	\end{subfigure}%
	\begin{subfigure}[t]{.33\linewidth}
		\centering
		\includegraphics[height=4cm]{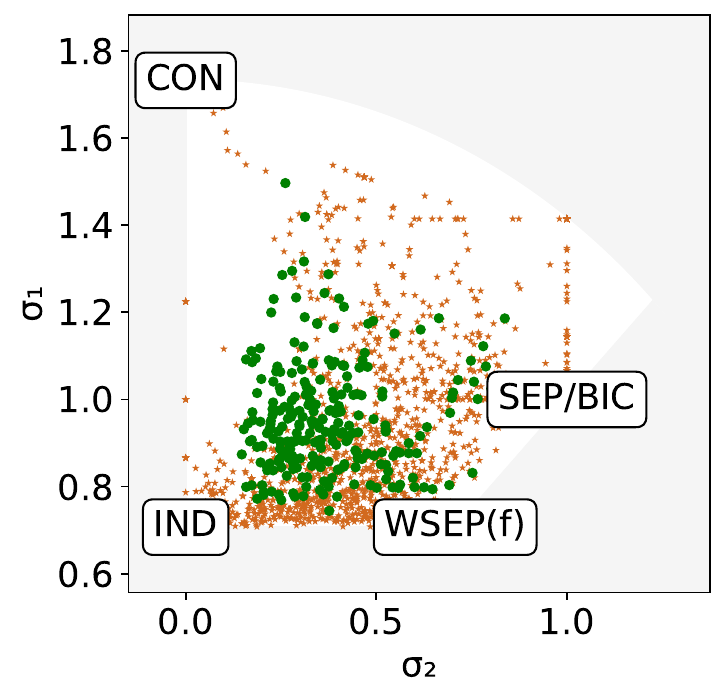}
		\caption{$3\times 6$}
	\end{subfigure}\\
	\begin{subfigure}[t]{.32\linewidth}
		\centering
		\includegraphics[height=4cm]{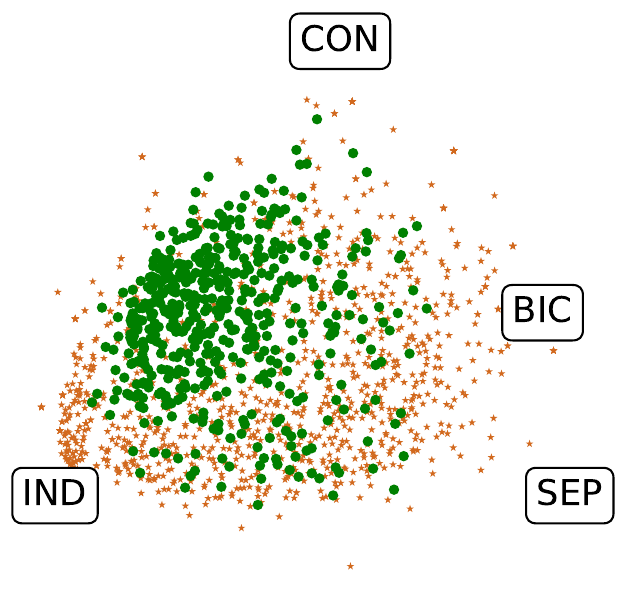}
		\caption{$5\times 5$, valuation distance}
	\end{subfigure}%
	\begin{subfigure}[t]{.32\linewidth}
		\centering
		\includegraphics[height=4cm]{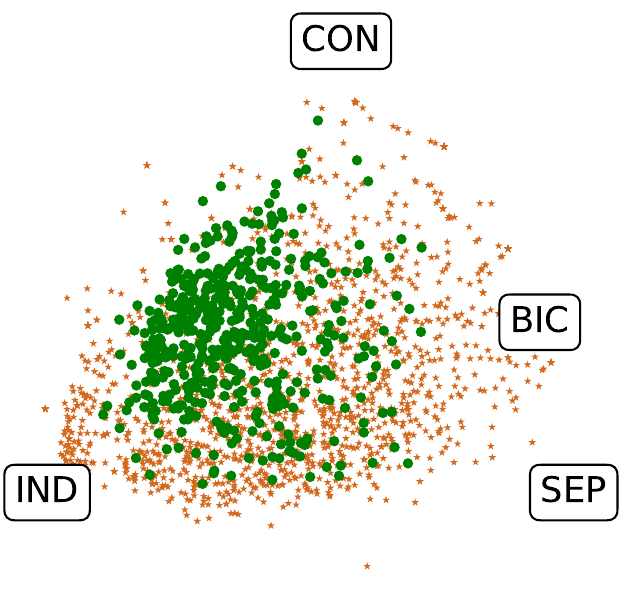}
		\caption{$5\times 5$, demand distance}
	\end{subfigure}%
	\begin{subfigure}[t]{.33\linewidth}
		\centering
		\includegraphics[height=4cm]{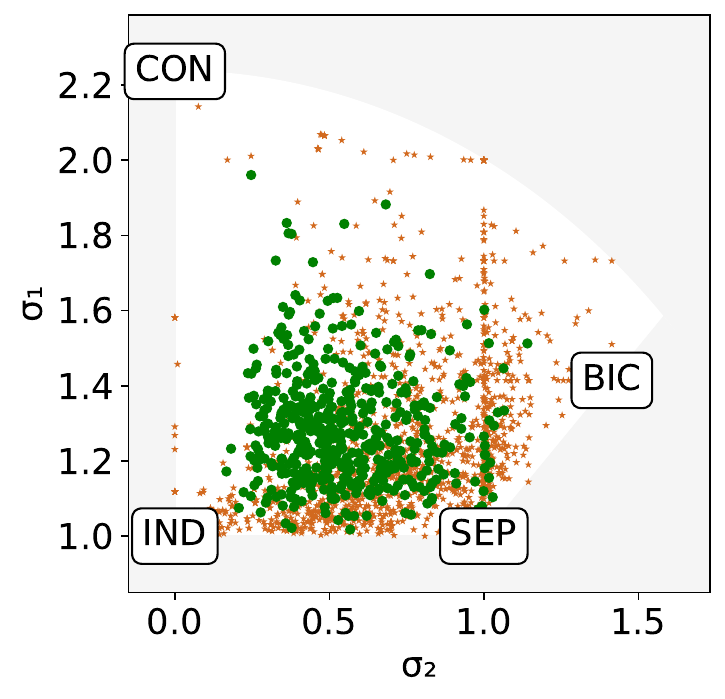}
		\caption{$5\times 5$}
	\end{subfigure}\\
	\caption{The instances from the island dataset are marked as green dots, while all other instances are marked as brown stars.}
	\label{fig:islstri}
\end{figure*}

\begin{figure*}\centering
	\begin{subfigure}[t]{.32\linewidth}
		\centering
		\includegraphics[height=4cm]{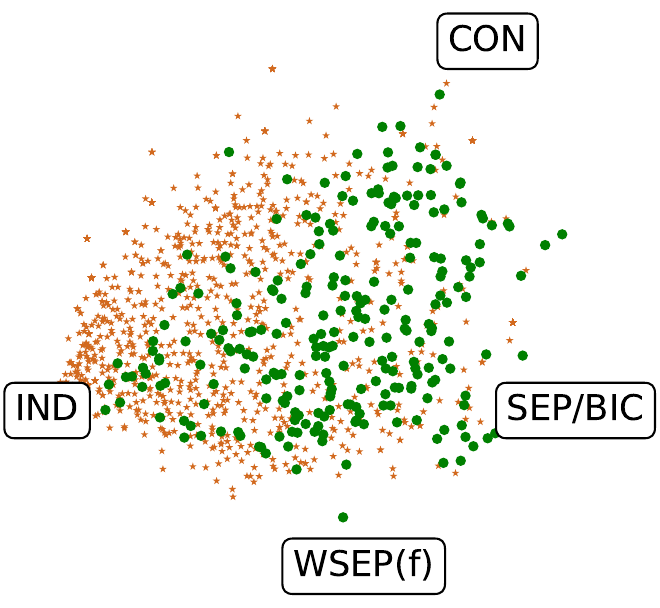}
		\caption{$3\times 6$, valuation distance}
	\end{subfigure}%
	\begin{subfigure}[t]{.32\linewidth}
		\centering
		\includegraphics[height=4cm]{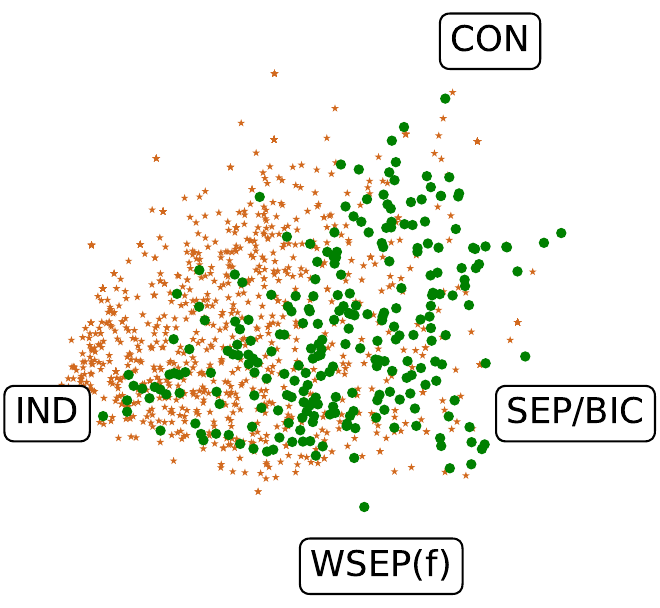}
		\caption{$3\times 6$, demand distance}
	\end{subfigure}%
	\begin{subfigure}[t]{.33\linewidth}
		\centering
		\includegraphics[height=4cm]{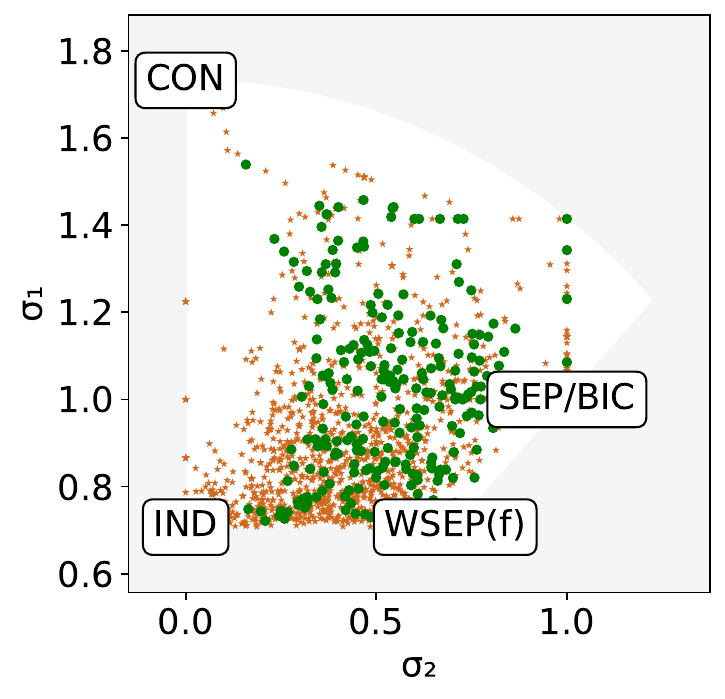}
		\caption{$3\times 6$}
	\end{subfigure}\\
	\begin{subfigure}[t]{.32\linewidth}
		\centering
		\includegraphics[height=4cm]{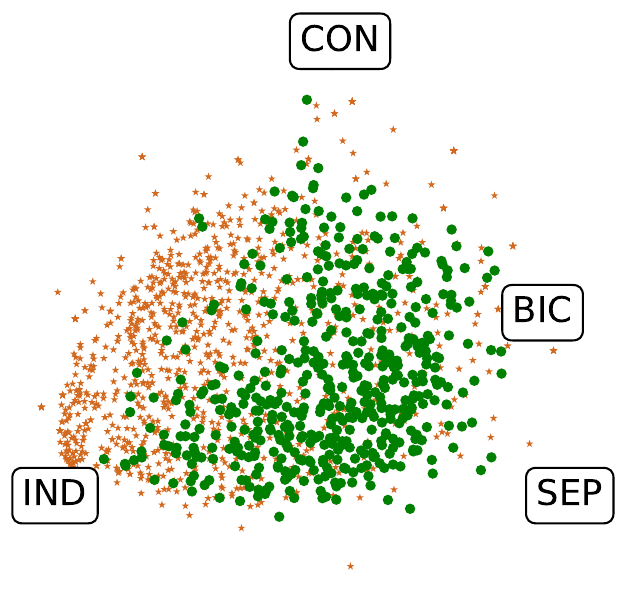}
		\caption{$5\times 5$, valuation distance}
	\end{subfigure}%
	\begin{subfigure}[t]{.32\linewidth}
		\centering
		\includegraphics[height=4cm]{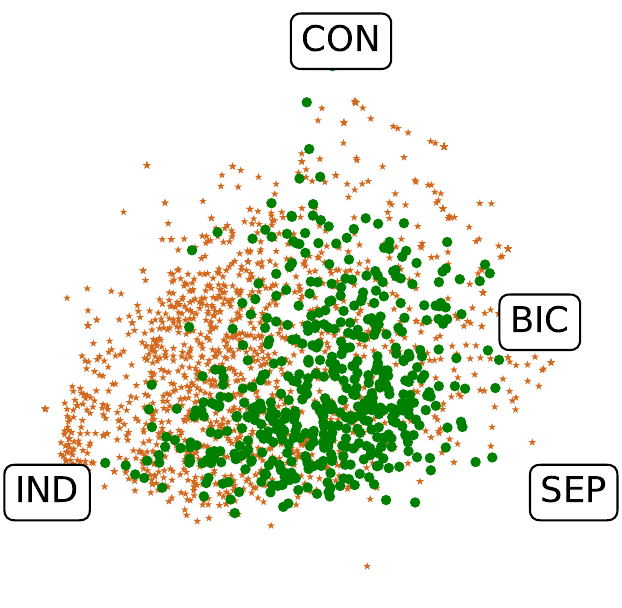}
		\caption{$5\times 5$, demand distance}
	\end{subfigure}%
	\begin{subfigure}[t]{.33\linewidth}
		\centering
		\includegraphics[height=4cm]{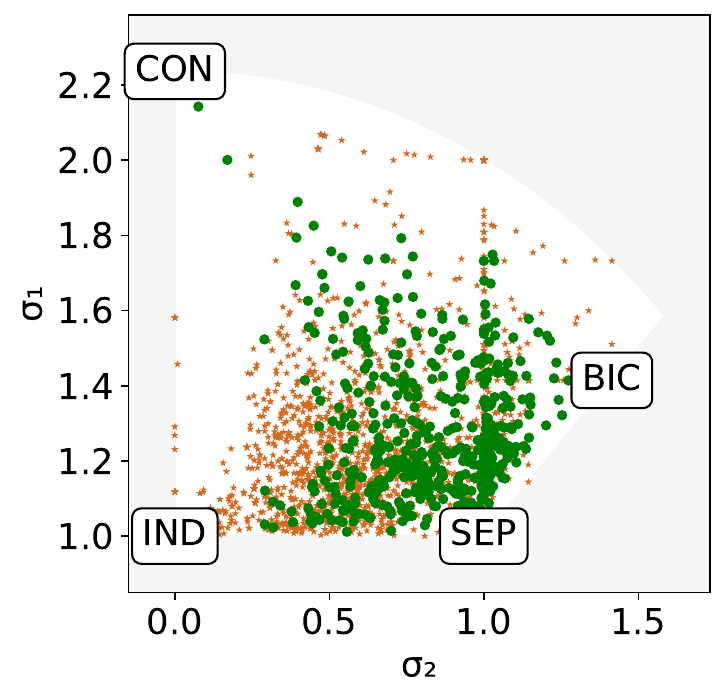}
		\caption{$5\times 5$}
	\end{subfigure}\\
	\caption{The instances from the candies dataset are marked as green dots, while all other instances are marked as brown stars.}
	\label{fig:candiescompar} 
\end{figure*}

\end{document}